\documentclass[12pt]{report}
\usepackage[paper=a4paper,dvips,top=3.0cm,left=3.0cm,right=3.0cm,
    foot=1.0cm,bottom=2.0cm]{geometry}
\usepackage{graphicx}
\usepackage{minitoc} 
\usepackage{subfigure}
\usepackage[colorlinks=true]{hyperref}
\hypersetup{urlcolor=blue, citecolor=blue, linkcolor= blue}
\hypersetup{pdfborder={0 0 0}}
\usepackage{abstract}
\usepackage[utf8]{inputenc}
\usepackage[T1]{fontenc}
\usepackage[usenames]{color}
\usepackage{booktabs}
\usepackage[blocks]{authblk}
\usepackage[fleqn]{amsmath}
\usepackage{amsfonts}
\usepackage{amssymb}
\usepackage{mathrsfs}
\usepackage{amsthm}
\usepackage{array}
\newtheorem{Lemma}{Lemma}
\newtheorem{Remark}{Remark}

\newtheorem{Definition}{Definition}

\newtheorem{Example}{Example}

\title{\Huge \bf {Diseases on complex networks. Modeling from a database and a protection strategy proposal.}\footnote{Work supported and funded by ANID COVID0739 project.}}
\author{Ronald Manr\'iquez}
\affil{Laboratorio de Investigaci\'on Lab[e]saM,\\ Departamento de Matem\'atica y Estad\'istica, \\ Universidad de Playa Ancha, Valpara\'iso, Chile. \\ ronald.manriquez@upla.cl}
\author{Camilo Guerrero-Nancuante}
\affil{Escuela de Enfermer\'ia,\\ Universidad de Valpara\'iso, Vi\~{n}a del Mar, Chile. \\ camilo.guerrero@uv.cl}

\date{}
\begin{document}
\maketitle

\dominitoc
\tableofcontents

\chapter*{Introduction}

Infectious diseases have been the object of study throughout the history of mankind. Multiple disciplines have contributed to the understanding of these health phenomena, in particular the sources and types of infections, as well as the negative consequences on the population.

From an epidemiological and health perspective, humanity has experienced a series of infectious disease events, including Cholera, Malaria and AIDS \cite{1}. Infectious diseases have an epidemic potential due to the dissemination of microorganisms, generally viruses, that develop in a host and later seek another living being to continue with their survival process \cite{2} \cite{3}.  Therefore, the spread of this type of disease occurs through contact between living beings, humans or animals, which present significant loads of pathogenic microorganisms. Consequently, when massive infections occur, we are facing an epidemic outbreak. The concept of an epidemic is established when the infectious outbreak affects a specific geographic area and a pandemic is related to an event spread over extensive continental areas \cite{4}.

An example of the above is the current COVID-19 pandemic context, the study of the spread of diseases being of interest \cite{bhapkar_virus_2020} and \cite{CROCCOLO2020110077}. The beginning of the pandemic was registered in the city of Wuhan-China \cite{article}. The consequences of the COVID-19 pandemic have been evidenced in a series of dimensions, including the collapse of health systems in some countries, the stoppage of production, the impoverishment of communities, unemployment, among other social and economical consequences \cite{8}.

In this sense, the current and historical contributions of mathematical models are important. The compartmental models are useful to establish in a simple way the projections and evolution of infectious diseases. They are characterized by compartmentalizing the population depending on whether the disease generates immunity or not \cite{3}.
One of the classic compartmental mathematical models is SIR, developed by Kermack and Mc Kendrick in 1927 for the understanding of epidemics \cite{Kermack_1927} and the current use of computational simulations, are of great relevance to analyze the behavior, in this case, of SARS-CoV2 (see for instance \cite{guerrero-manriquez_proyeccion_2020}).
The SIR model compartmentalizes or divides the population into Susceptible (S), Infected (I) and Removed (R). This compartmentalization allows to analyze the population with these states and is useful to determine projections in relation to the total number of patients and the duration of the disease \cite{11}. The SIR model approach is eminently deterministic, however, it has also been used from a stochastic perspective, improving the representation of the dynamics of infectious diseases through the probability of the appearance of epidemic outbreaks \cite{11}. There are other mathematical models in epidemiology that have been developed from the SIR model, adding variables such as exposure and the effect of quarantine measures such as the SEIR and SEQIJR model respectively \cite{12}. With this, nations and governments can count on information to establish mitigation measures for the consequences of the virus, such as: safeguarding employment, strengthening health system responses, developing community actions, among other measures.

However, these models are limited when the extension and heterogeneity of the data is wide, so they fail to detect changes in the population structure and the variation in contact dynamics over time \cite{12}.

On the other hand, globalization and high population concentration have led to the inclusion of other ways of representing the spread of infectious diseases. Models with stochastic approximations have the advantage of establishing probabilities of person-to-person contact \cite{3}. One of them is the network model, which is based on the theory of graphs studied from the observations of Leonhard Euler with the problem of the seven bridges. The model proposes the formation of individuals (nodes) and their relationship with others (edge), so the result is a network \cite{12} (see for instance figure \ref{fig_000}).

\begin{figure}[hbt!]
\centering
{\includegraphics[width=0.5\textwidth]{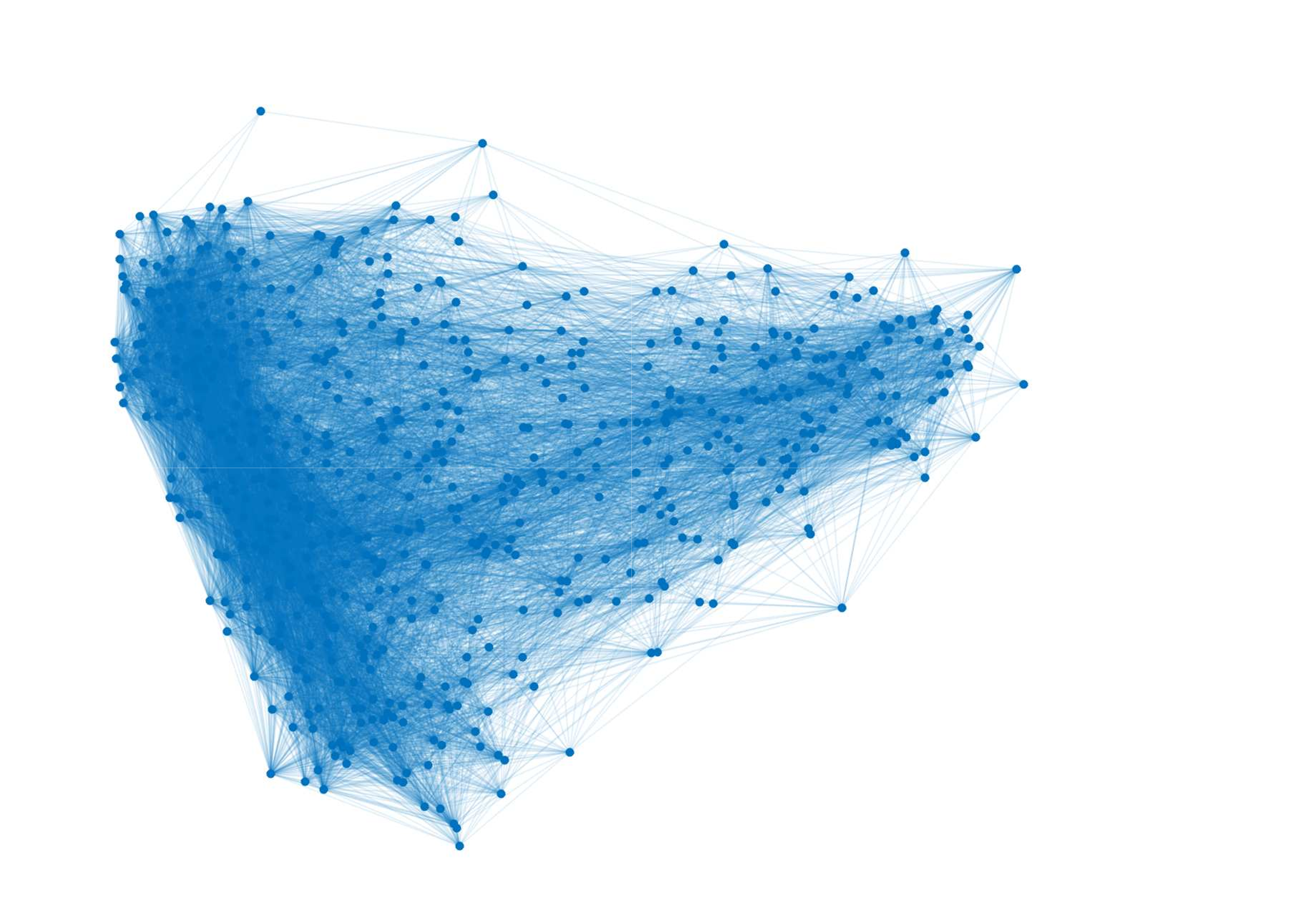}}
\caption{Network model of brain functional coactivations}\label{fig_000}
\end{figure}


Networks study, in the form of mathematical graph theory, is one of the most important objets from discrete mathematics \cite{newman2003structure}. The importance of its various applications in issues of traffic, economy, brain networks, spread of diseases, social networks, electrical power grid, etc. has made this field of research get the attention of many researchers since the last 20 years (see for instance \cite{almasi2019measuring},  \cite{an2014synchronization}, \cite{Crossley11583}, \cite{guangzeng2009novel}, \cite{montenegro2019linear} and \cite{pastor-satorras_epidemic_2001}).


In a complex network there are more important vertices than others that influence or have a greater impact on their structure or function. This importance establishes some vertices ranking that can be based on different viewpoints, this is to say, we can have an importance based on the neighborhood, on the paths, on the link, connectivity or sensibility (more details in \cite{lu2016vital}). For example, the degree of a vertex, allows ordering or making a vertices ranking according to the number of neighbors it has. In general, we have the well known centrality measures that give us a vertices importance ranking. However, identifying vital nodes is not an easy job because the criteria for vital nodes are diverse due to their context-dependent nature, and consequently it is not possible to find a universal index that better quantifies the importance of nodes in all situations. Even harder, if you want to find a set of important nodes, it is a NP-Hard problem (see \cite{1318579}). 

On the other hand, modeling the spread of disease over networks has many applications in real life. For~example, rumors or spam that spread on large social networks such as Facebook or Twitter~\cite{guess2019less}, viruses that spread through computer systems~\cite{yang2012towards}, or~diseases that spread over a population such as it does the actual SARS-CoV-2 (see for instance~\cite{ronald1}). This leads to solve the problem of totally or partially controlling these spreads.

Scientific evidence supports immunization strategies being useful in both homogeneous and heterogeneous networks \cite{nian2018immunization}. The strategies seek, first, to develop a dynamic experimental model that includes the classic compartmental systems in epidemiology (SIR, SEIR, SIS, among others) and then to establish immunization measures across the nodes. Immunization measures can be applied in static and dynamic networks, in a random or targeted immunization \cite{gomez2006immunization}.

Random immunization refers to random strategies, without determining a particular population in the protection process \cite{cohen2003efficient}. In contrast, targeted immunization recognizes nodes with a higher degree of connection with other nodes \cite{xia2018improved}. Nian et al. conclude, through computational simulations in a free scale graph (Model BA), that targeted immunization is more effective than randomized \cite{nian2018immunization}. In the same direction, Wang et al. conclude that, even if the immunization strategy is imperfect or incomplete, it manages to generate positive impacts on the protection of the network \cite{wang2009imperfect}. Another investigation by Xia et al. determined that targeted immunization in two rounds of selection provides a greater protective effect compared to progressive strategies \cite{xia2018improved}.

According to the above, a study conducted by Ghalmane et al. proposed an immunization strategy considering the influence of the nodes, the number of communities and the links between them \cite{ghalmane2019immunization}. Along the same lines, Gupta et al. analyzed the importance of protection strategies using information from community networks  \cite{gupta}. Since edge-weighted graphs have been an important way of understanding epidemic diseases, Manr\'iquez et al. measure the effectiveness of DIL-W$^{\alpha}$ with the recognition of bridge nodes. The effectiveness of DIL-W$^{\alpha}$ is high compared to other proposals on four real networks~\cite{ronald2}.

For all the above, studies that include immunization measures together with community/local network models are highly relevant for public health, both for establishing measures to mitigate epidemic diseases and for the development of vaccination optimization policies to more quickly achieve herd immunization and, therefore, overcome diseases such as COVID-19. This is supported by Zhao et al., who argue that this framework allows for the optimization of immunization resources \cite{zhao2014immunization}.

Regarding the COVID-19 pandemic, a series of models have been developed which allow the projection and establishment of the progress of this infectious disease based on the data available from different information sources (see, for instance, \cite{guerrero-manriquez_proyeccion_2020}). In this sense, Manríquez et al. propose the use of weighted graphs at the edges, giving the network model, from a stochastic approach, the possibility of identifying the most important variables in the spread of COVID-19 with real data in a city in Chile \cite{ronald1}. The same authors, in order to classify the importance of the nodes, propose the generalization of the measure of importance of the line with the use of degree of centrality (DIL-W$^{\alpha}$), improving the understanding of the network from a local perspective \cite{ronaldranking}.\\

The objective of this report is to summarize the research carried out by the authors in the context of the COVID0739 project funded by supported and funded by Agencia Nacional de Investigación y Desarrollo (ANID) based on the works \cite{ronald1}, \cite{ronaldranking}, \cite{ronald2}, and \cite{ronald3}.

The report is organized as follows. Chapter \ref{basic} contains generalities about graph theory and SIR model. 

Chapter \ref{modeldes} is divided in three parts. First, we give the construction of the graph from a database. In the second part, we describe the Graph-based SIR model. In Section \ref{determi} we give a deterministic approach to the Graph-based SIR model. At the end of the chapter, we provide a discussion.\\

In Chapter \ref{chap2} the generalization of the importance vertex measure, denoted by DIL-W$^{\alpha}$, is given. In Section \ref{Comparison}, we compare our proposal DIL-W$^{\alpha}$ with the one given by Almasi \& Hu (denoted by DIL-W) considering $\alpha=1$ and $\alpha=0.5$ on some real networks (Zacharys karate club network, US air transport network, meta-analysis network of human whole-brain functional coactivations, colony of ants network, and wild birds network). Moreover, we analyze the DIL-W$^{\alpha}$ rank with different values of $\alpha$, and provide a discussion about the results. Finally, in Section \ref{conclusionchap2}, we give the final conclusions.

In Chapter \ref{pro}, we provide definitions of the protection of a graph when a disease spreads on it. Moreover, we state the protection strategy. Section~\ref{simulation} is devoted to run simulations on test networks (Zacharys karate club network, Wild birds network, Sandy authors network and CAG-mat72 network) and scale-free networks and it also provides a discussion about the results. Moreover, we address the variation in the survival rate when the protection budget is modified. Finally, we give the conclusions in Section~\ref{conclusionpro}.

Finally, the Chapter \ref{covid} aims to analyze the protection effect against COVID-19 using the DIL-W$^{\alpha}$ ranking with real data from a city in Chile (Olmu\'e-City), obtained from the Epidemiological Surveillance System of the Ministry of Health of Chile, from which we obtain an edge-weighted graph, denoted by $G_{\mathcal{E}}$, according to the method proposed in \cite{ronald1}. We apply the protection to the $G_{\mathcal{E}}$ network, according to the importance ranking list produced by DIL-W$^{\alpha}$. In Section \ref{meth}, we obtain the graph from a real database from a city in Chile (Olmu\'e-City), and we set the protection strategy. In Section \ref{results}, the results of the study are presented. Section \ref{discussion} provides a discussion of the results and potentialities of the method used. Finally, Section \ref{conclusioncovid} provides the conclusions.

\chapter{Basic definitions}\label{basic}
\minitoc
In this section, we establish the definitions and elements used throughout this report.
We summarize the symbols and notations in Table~\ref{simbol}.
{\begin{table}[h!]
\caption{{Summary of Symbols and Notations.}}
\begin{tabular}{cl}
\toprule
{\textbf{Notations}}   & {\textbf{Definition and Description}      }                                   \\ \midrule
$G$                  & Graph or network.                                                            \\
$(G,w)$		&Edge-weighted graph.                                                   \\
$v_i$                & Vertex or node.                                                             \\
$N(v_i)$             & Neighborhood of a vertex $v$.                                               \\
$e_{ij}$             & Edge between vertex $v_i$ and vertex $v_j$.                                 \\
$w_{ij}$             & Weight of the edge $e_{ij}$.                                                \\
$deg(v_i)$           & Degree of the vertex $v_i$.                                                 \\
$S(v_i)$             & Strength of the vertex $v_i$.                                               \\
$\alpha$             & Real number. Tuning parameter.                                              \\
$C_D^{w\alpha}(v_i)$ & Degree centrality of $v_i\in V$ of an edge-weighted graph $(G,w)$.\\
DIL-W$^{\alpha}$     & Ranking based on Degree and importance of line.                              \\
$I^{\alpha}(e_{ij})$ & Importance of edge $e_{ij}$.                                                \\
$W^{\alpha}(e_{ij})$ & Contribution that $v_i$ makes to the importance of the edge $e_{ij}$.       \\
$L^{\alpha}(v_i)$    & The importance of a vertex $v_i$.                                           \\
$\mathcal{E}$        & Database.                                                                   \\
$\mathcal{X}_k$      & Variable of a database.                                                     \\
$p_k$                & Weight of the variable $\mathcal{X}_k$.                                     \\
$k$                  & Protection budget (the number of nodes in graph $G$ that can be protected). \\
$\sigma$             & Ratio of surviving nodes.             \\ \bottomrule
\end{tabular}
\label{simbol}
\end{table}}

\subsection{Graphs}
The following definitions come from \cite{chartrand_graphs_1996} and \cite{west_introduction_2001}.
\begin{Definition}
A graph $G$ is a finite nonempty set $V$ of objects called vertices together with a possibly empty set $E$ of 2-element subsets of $V$ called edges.
\end{Definition}

To indicate that a graph $G$ has vertex set $V$ and edge set $E$, we write $G=(V,E)$. If~the set of vertices is $V=\{v_1,v_2,\ldots,v_n\}$, then the edge between vertex $v_i$ and vertex $v_j$ is denoted by $e_{ij}$.

If $e_{ij}$ is an edge of $G$, then $v_i$ and $v_j$ are adjacent vertices. Two adjacent vertices are referred to as neighbors of each other. The~set of neighbors of a vertex $v$ is called the open neighborhood of $v$ (or simply the neighborhood of $v$) and is denoted by $N(v)$. If~$e_{ij}$ and $e_{jk}$ are distinct edges in $G$, then $e_{ij}$ and $e_{jk}$ are adjacent edges.

\begin{Definition}
The number of vertices in a graph $G$ is the order of $G$ and the number of edges is the size of $G$.
\end{Definition}

\begin{Definition}
The degree of a vertex $v$ in a graph $G$, denote by $deg(v)$, is the number of vertices in $G$ that are adjacent to $v$. Thus, the degree of $v$ is the number of vertices in its neighborhood $N(v)$.
\end{Definition}

\begin{Definition}
Let $G$ be a graph of order $n$, where $V(G)=\{v_1,v_2,\ldots, v_n\}$. The adjacency matrix of $G$ is the $n\times n$ zero-one matrix $A(G)=[a_{ij}]$, or simply $A = [a_{ij}]$, where
$$a_{ij}=\begin{cases}
1 & \text{if}\;\; e_{ij}\in E(G)\\
0 & \text{if}\;\; e_{ij}\notin E(G).
\end{cases}$$
\end{Definition}

On the other hand, an important generalization of the simple graph consists of the definition of weighted graph, more specifically edge-weighted graph. Informally, an edge-weighted graph is a graph whose edges have been assigned a weight. 

\begin{Definition}
An edge-weighted graph is a pair $(G,W)$ where $G=(V,E)$ is a graph and $W: E\rightarrow \mathbb{R}$ is a weight function. If $e_{ij}\in E$ then $W(e_{ij})=w_{ij}$.
\end{Definition}

\begin{Definition}
The strength of a vertex $v_i$, denoted by $S(v_i)$, is defined as the sum of the weights of all edges incident to it, this is to say
$$S(v_i)=\sum_{v_j\in N(v_i)}w_{ij}.$$
\end{Definition}

The following definition comes from \cite{opsahl2010node}. 

\begin{Definition}[Degree centrality \cite{opsahl2010node}]\label{dcent}
The degree centrality of $v_i\in V$ of an edge-weighted graph $(G,w)$, denoted by $C_D^{w\alpha}(v_i)$, is defined as
\begin{align}
C_D^{w\alpha}(v_i)&=deg(v_i)^{(1-\alpha)}\cdot S(v_i)^{\alpha},
\end{align}
where $\alpha\in[0,1]$. 
\end{Definition}

The parameter $\alpha$ is called {\it tuning parameter}. Notice that when $\alpha=0$ then $C_D^{w\alpha}(v_i)=deg(v_i)$ and when $\alpha=1$ then $C_D^{w\alpha}(v_i)=S(v_i)$. 

\subsection{SIR model}
In the entire spectrum of epidemiological models that currently exist, the SIR model is the basis or the simplest of all these.\\
The classical Kermack-McKendrick SIR model \cite{Kermack_1927}, developed in the early 1900s (see \cite{anderson_discussion_1991, ross_sir_2013}), consists of a system of non-linear ordinary differential equations, which expresses the spread among the population of a constant size, denoted by $N$, for all time $t$. The population is divided into three groups: susceptible individuals, infected individuals, and recovered (or removed) individuals. The sizes of these groups at time $t$ are denoted by $S(t)$, $I(t)$, and $R(t)$ respectively such that $N=S(t)+I(t)+R(t)$. The model is the following,
\begin{align}
\dot{S}(t)&=-\beta S(t) I(t)\\
\dot{I}(t)&=\beta S(t) I(t)-\delta I(t)\\
\dot{R}(t)&=-\delta I(t),
\end{align}
$t\in [0, T ]$, subject to the initial conditions $S(0)=S_0$, $I(0)=I_0$ and $R(0)=R_0$.
where the disease transmission rate $\beta>0$ and the recovery rate $\delta>0$ (the duration of infection $\delta^{-1}$).

In summary, the above system describes the relationship between the three groups, this is to say, a susceptible individual changes its state to infected with probability $\beta$, while an infected changes its state to recovered with probability $\delta$.

\chapter{Obtaining a graph from a database}\label{modeldes}
\minitoc

\section{Introduction}
The understanding of  infectious diseases is a priority in the field of public health. This has generated the inclusion of several disciplines and tools that allow to analyze the dissemination of infectious diseases. The aim of this chapter is  to model the spreading of a disease in a population that is registered in a database. From this database, we obtain an edge-weighted graph. The spreading was modeled with the classic SIR model. The model proposed with edge-weighted graph allows to identify the most important variables in the dissemination of epidemics. \\

To build a network model, the variables that are relevant to the spread of a disease are established. Among the multiple models developed, sociocentric studies stand out, which allow a broad exploration of the complete network that is generated to understand the spread of a pathogen. Therefore, network models are useful to understand the development of different infectious diseases \cite{cardinal-fernandez_medicina_2014}. It is not new, especially if the spread of an epidemic disease in network structures is studied, which in an abstract way is the main object of graph theory, see for instance in \cite{RIZZO2016212} where the authors use the model as a predictive tool, to emulate the dynamics of Ebola virus disease in Liberia, and  in \cite{shafer} where the transmission connectivity networks of people infected with highly contagious Middle East respiratory syndrome coronavirus (MERS-CoV) in Saudi Arabia were assessed to identify super-spreading events among the infected patients between 2012 and 2016. The relevance of these studies is related to the possibility of preventing nodes (people) from continuing to infect, an issue that is treated in \cite{wijayanto2019effective} with the graph protection methods proposed by Wijayanto \& Murata. Deepening in this line of studies, there have been included new mathematical models with variables related to the behavior of people associated with information and emotions during epidemics \cite{17}. Likewise, dynamic models have covered the influence of the infectious disease itself on the network of contacts and, therefore, changes in the dynamics of spread of the epidemic over time \cite{12}.

In general, network models can be analyzed through static graphics such as snapshots. However, for an adequate approximation and correcting the loss of data generated by the snapshot, data modeling techniques must be used including the weighting of the edges and, consequently, better estimated, given the information obtained from the relationships between individuals and the spreading of the disease \cite{12}.

In the current context, organized information has an important value for the management of epidemics. The databases elaborated from the information of individuals can contribute in the characterization and knowledge of a determined area, which are important to know the evolution of the diseases \cite{margevicius_advancing_2014}.

In the case of infectious pathologies, depending on the type of database, it is possible to determine through the variables whether two or more individuals are linked to each other, such as people who live in the same neighborhood or work in the same place. Given the characteristics of network models and the obtaining of information through complex databases, it is relevant to use these models, in particular, with approximations that incorporate weighting on the edges.

For all the above, there are challenges around the possibility of representing and understanding the evolution of infectious diseases through a network model using databases. The relevance of this type of research contribution is based on the possibility of having tools that favor measures to prevent the spread of this type of disease, an issue that takes on greater social and scientific value due to the context of the SARS-CoV2 pandemic.
In consequence, the aim of this chapter is to develop a spreading stochastic model of some disease from a database, particularly using variables that link individuals in a given territory, the probability of contagion among them, and therefore, the spread of the disease through edge-weighted graphs (or edge-weighted networks). For the purposes of this manuscript, a database is understood as a matrix whose columns are the variables, while the rows correspond to the responses of the subjects in relation to the variables consulted.


\section{Graph from a database}
The authors in $\cite{ronald1} $ provide a way to obtain an edge-weighted graph from a database that we briefly detail.
This section is divided into two parts. A part dedicated to build a graph from a database and the second one to describe the dynamics of the disease on the graph-based SIR model. \\
To begin, we need some basic elements to understand what follows. First, we will understand by {\it variable} the characteristic assigned to a person from a predetermined set of values which can be a numerical measure, a category or a classification. For instance, income, age, weight, occupation, address, etc. Second, we will understand by {\it database} a matrix whose columns are variables, while the rows correspond to the responses of the subjects in relation to the variables consulted.\\

Let us consider a database, denoted by $\mathcal{D}$, that stores information on $N$ individuals of a population. 
Let $V$ be the set of the persons registered in the database, equivalent $$V=\left\{v_1,v_2,\ldots,v_N\right\},$$
where $v_i$ is a person registered in $\mathcal{D}$ for $i=1,\ldots,N$.\\
Let $v_i$ be a person registered in $\mathcal{D}$ for $i=1,\ldots ,N$. We set $$EPI=\left\{\mathcal{X}_1,\mathcal{X}_2,\ldots,\mathcal{X}_K\right\},$$ where $K$ is the number of elements of the set $EPI$. ($K$ is the number of variables in $\mathcal{D}$), $\mathcal{X}_k$ is a variable in $\mathcal{D}$ for $k=1,\ldots, K$, and $EPI(i,k)$ is the response of the person $v_i$ to the variable $\mathcal{X}_k$. 

As we want to study the link between the people who are registered in $\mathcal{D}$ through the variables of this database, we must identify which are the variables that allow us to establish these links that promote the spread of the disease. For example, if two people are the same age, they do not necessarily meet and spread the virus unlike two people who live in the same city.

 \begin{Definition}\label{reldefinition}
We will say that $\mathcal{X}\in EPI$ is a relationship variable if and only if it allows us to assume that some person meets another. In another case, we will say that $\mathcal{X}$ is a characteristic variable. 
\end{Definition}

The above allows us to define the following sets:
$$REL=\left\{\mathcal{X}\in EPI : \mathcal{X}\text{ is a relationship variable}\right\}$$
$$CHAR=\left\{\mathcal{X}\in EPI : \mathcal{X}\text{ is a characteristic variable}\right\}.$$
Let us denote by $K_1$ and $K_2$ the cardinality of $REL$ and $CHAR$ respectively. Notice that $K_1+K_2=K$.

\begin{Definition}\label{graph}
We will say that a person $v_i$ is related to a person $v_j$ if and only if there exists $\mathcal{X}_k \in REL$ for $k\in\{1,2,\ldots,K_1\}$ such that $EP I(i, k) = EP I(j, k)$ and $i\neq j$
\end{Definition}
It is clear that if $v_i$ is related to $v_j$ then $v_j$ is related to $v_i$, this is to say the relation is a symmetric relation.

The previous definition allows us to construct a graph $G$ of links given by the relationship variables of the $\mathcal{D}$. $G$ will be considered as an undirected graph without loops or multiple edges.

On the other hand, it is possible that $k\in\left\{1,\ldots,K_1\right\} $ is not unique because more than one variable may coincide. This induces us to define the weight of the link between $v_i$ and $v_j$.
\subsection{Weighting variables}
In order to define the weight of each link between two vertices, we assume that each $\mathcal{X}\in REL$ has an associated inherent weight, this is to say, it is possible to discriminate some hierarchical order between the variables. Let $p_k$ be the weight associated to the variable $\mathcal{X}_k\in REL$ for $k=1,\ldots, K_1$. 

On the other hand, to better understand the definition that follows and its consequences, suppose that $X$ is a set of $100$ people. If we define the relationship in the set: {\it person Q is related to person W if and only if they are the same age}, then we could group the people in the set by age. In addition, an interesting fact is that thanks to this relationship everyone would be part of a group and no one could be in two groups at the same time. This type of relations defined on a set are called {\it equivalence relations} and each group that is defined is called an {\it equivalence class}.
 \begin{Definition}
We will say that for $\mathcal{X}_{j},\mathcal{X}_{t}\in REL$, $\mathcal{X}_{j}$ is related to $\mathcal{X}_{t}$, denoted by $\mathcal{X}_{j} \mathcal{R}\mathcal{X}_{t}$, if and only if $p_{j}=p_{t}$.
\end{Definition}

\begin{Lemma}
The relation $\mathcal{R}$ defined on $REL$ is an equivalence relation.
\end{Lemma}
\begin{proof}
Directly.
\end{proof}

Thanks to the relation $\mathcal{R}$, we can consider the different equivalence classes which are composed of the variables that have the same weight. Hence, we have the same number of classes as different weights.

 \begin{Definition}\label{weightvariable}
Let $\displaystyle A_1, A_2,\ldots, A_c $ be the different classes that are defined by the different weights $\displaystyle p _1, p _2, \ldots, p _c $ and $\displaystyle\alpha_1, \alpha_2,\ldots, \alpha_c $ its respective cardinalities. Hence,

\begin{equation}
p _j=\frac{\alpha_j}{K_1},
\end{equation}
for all $j\in\{1,2,\ldots,c\}$.
\end{Definition}

\subsection{Weighted link}
The aim in this subsection is to introduce the definition of weight link.\\
Let $v_i, v_j\in V$ such that $v_i$ is related to $v_j$. We set 
\begin{align}
H&=\left\{k\in\left\{1,\ldots,K_1\right\}:EPI(i,k)=EPI(j,k)\right\}.
\end{align}
We denoted by $h_{i,j}$ the cardinality of the set $H$. Notice that $h_{i,j}$ is simply the number of times that one person is related to another (or the number of variables that matches between them). Since our proposal considers undirected graph, we have that $h_{i,j}=h_{j,i}$.

 \begin{Definition}\label{pesosbordes1}
Let $v_i, v_j\in V$ be such that $v_i$ is related to $v_j$ and $p_{k_r}$ the weight of the variable in which $v_i$ and $v_j$ match, for $r=1,\ldots,h_{i,j}$. We will say that
\begin{equation}\label{pesosbordes}
\widetilde{w}_{ij}=\sum_{r=1}^{h_{i,j}}p_{k_r},
\end{equation}
is the weight of the link between $v_i$ and $v_j$.
\end{Definition}

Finally, the weighted adjacency matrix of $G$ is the $n\times n$ matrix $A(G)=[a_{ij}]$, where
$$a_{ij}=\begin{cases}
\widetilde{w}_{ij} & \text{if}\;\; e_{ij}\in E(G)\\
0 & \text{if}\;\; e_{ij}\notin E(G).
\end{cases}$$

The idea of having a graph with weights in its edges is to be able to differentiate or measure, in some way, the strength or closeness between individuals. For example, it is not the same saying that two individuals share the same city than saying they share the same house they live in, it follows that the latter makes the relationship closer and consequently the contagion of the disease is intuitively more likely.\\

\begin{Example}
{In the following example, Table \ref{table1} simulates a database with 20 registered people. The data hosted correspond to the city in which they live (City), the workplace (considering school and university as a workplace), gender (Gen.), age, extracurricular activity (EC activity), address, whether they drink alcohol (Drin.), whether they are smokers (Sm.) and marital status (MS).
Let us consider $A$ and $B$ as two different cities, and $x, y, z, w, u, v, r, s, q, t, p, k, d, g$ and $h$ as different people's addresses. Moreover, in the table, Y = Yes, N = No, IC = in couple, M = married, S = single, W = widower. }

\begin{table}[h]
\small
\caption{Database $\mathcal{E}$.}
\label{table1}
\begin{tabular}{cccccccccc}
\toprule
\textbf{Person} & \textbf{City} & \textbf{Workplace}  & \textbf{E. C. Activity}& \textbf{Address} & \textbf{Sm.} & \textbf{Dri.} & \textbf{Gen.} & \textbf{M. S}& \textbf{Age} \\ \midrule
1  & A & Workplace 1  & Theater       & y & Y & Y & F & IC & 35 \\
2  & A & Workplace 3  & Cinema        & y & Y & Y & M & IC & 35 \\
3  & B & School B     & Football      & z & N & N & F & S  & 10 \\
4  & B & Workplace 1  & Photography   & x & N & N & F & M  & 48 \\
5  & A & Workplace 5  & Does not have & u & Y & N & F & W  & 65 \\
6  & A & Workplace 4  & Does not have & v & Y & Y & M & IC & 27 \\
7  & B & Workplace 2  & Does not have & x & Y & N & M & M  & 46 \\
8  & A & University 1 & Photography   & v & N & N & M & IC & 29 \\
9  & A & University 2 & Does not have & w & Y & Y & M & IC & 19 \\
10 & B & School B     & Karate        & x & N & N & M & S  & 10 \\
11 & A & Workplace 4  & Ping-pong     & r & Y & Y & F & M  & 54 \\
12 & A & School A     & Football      & s & N & N & M & S  & 8  \\
13 & A & Workplace 5  & Dance         & r & Y & Y & F & M  & 57 \\
14 & B & School A     & Handball      & q & N & N & M & S  & 11 \\
15 & A & University 1 & Does not have & t & N & N & F & S  & 25 \\
16 & A & Workplace 7  & Singing       & p & Y & Y & F & S  & 60 \\
17 & A & Workplace 8  & Music         & k & N & Y & F & S  & 28 \\
18 & A & Workplace 3  & Does not have & d & N & N & M & S  & 47 \\
19 & B & School A     & Music         & g & N & N & F & S  & 8  \\
20 & A & Workplace 6  & Does not have & h & Y & Y & M & S  & 30 \\ \bottomrule
\end{tabular}
\end{table}

{From Table \ref{table1}, we have that $EPI=\{\mathcal{X}_1, \mathcal{X}_2, \mathcal{X}_3, \mathcal{X}_4, \mathcal{X}_5, \mathcal{X}_6, \mathcal{X}_7, \mathcal{X}_8, \mathcal{X}_9\}$, where $\mathcal{X}_1=$ City, $\mathcal{X}_2=$ Workplace, $\mathcal{X}_3=$ E.P. activity, $\mathcal{X}_4=$ Address, $\mathcal{X}_5=$ Sm., $\mathcal{X}_6=$ Dri., $\mathcal{X}_7=$ Gen., $\mathcal{X}_8=$ M.S. and $\mathcal{X}_9=$ Age. Then, we obtain the sets:
\begin{enumerate}
\item $\displaystyle REL=\{\mathcal{X}_1, \mathcal{X}_2, \mathcal{X}_3, \mathcal{X}_4\}$ and
\item $\displaystyle CHAR=\{\mathcal{X}_5, \mathcal{X}_6, \mathcal{X}_7, \mathcal{X}_8, \mathcal{X}_9 \}$.
\end{enumerate}

In our criteria, the hierarchical order of the variables $\mathcal{X}_1, \mathcal{X}_2, \mathcal{X}_3, \mathcal{X}_4$ in descending form is $\mathcal{X}_4, \mathcal{X}_2, \mathcal{X}_3$, and $\mathcal{X}_1$. Moreover, we consider that the variables $\mathcal{X}_4$ and  $\mathcal{X}_2$ have the same weight. Hence,
$A_1=\{\mathcal{X}_2, \mathcal{X}_4\}$, $A_2=\{\mathcal{X}_3\}$, and $A_3=\{\mathcal{X}_1\}$ are the different classes that are defined by the different weights. Hence, by Definition \ref{weightvariable}
$$p_1=\frac{1}{2},\quad p_2=\frac{1}{4},\quad p_3=\frac{1}{4}.$$

To construct the graph, we must resort to Definition \ref{graph}. For instance, person $17$ is related to all the people who live in city $ A $ or who work at Workplace 8 or who have music as an extra curricular activity or whose address is $k$. With respect to the weights of the edges, Equation 6 in Definition \ref{pesosbordes1} gives us the answer. For instance, person $6$ matches person $11$ in the answers of the variables $\mathcal{X}_1$ and $\mathcal{X}_2$, this is to say, both people live in city $A$ and have the same workplace. Then, the edge $v_6v_{11}$ has weight $w_{6\;11}=0.5+0.25=0.75$. Figure \ref{graphex1} shows the obtained graph.}

\begin{figure}[h!]
\centering
\includegraphics[width=0.5\textwidth]{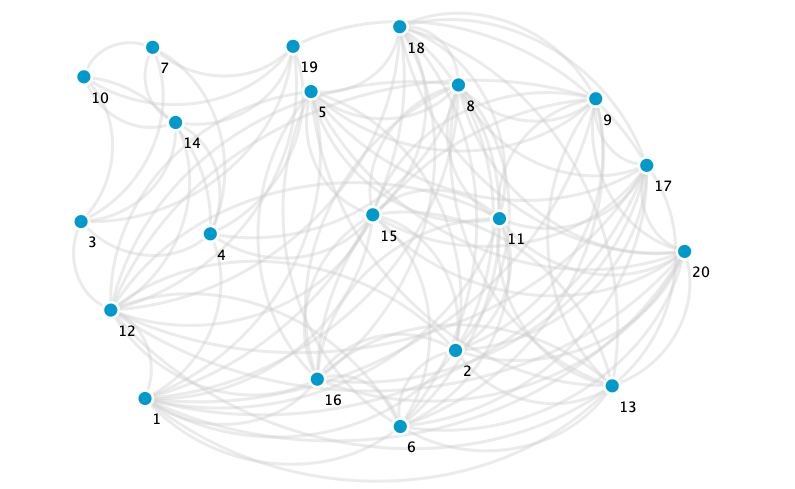}
\caption{Graph obtained from $\mathcal{E}$.} \label{graphex1}
\end{figure}
\end{Example}

\subsection{Graph-based SIR model}
Having described a population with a network model, the spreading of an epidemic is modeled by a dynamic system that uses the graph (in our case an edge-weighted graph) as its support. The class of chosen model is the probabilistic cellular automata (see \cite{noauthor_theory_2005}), this is to say, the model in which the events happen at times $t=0, \Delta t, 2\Delta t,\ldots$, where $\Delta t$ is the discretization interval. 
In this work, we use a graph-based SIR model in which each individual is represented by a vertex in an edge-weighted graph. At time $t$, each vertex $v_i$ is in a state $v_i^t$ belonging to $\mathcal{S}=\{0,1,-1\}$, where $0,1$ and $-1$  represent the three discrete states:  susceptible (S), infected (I) and recovered (R). 

Let $G$ be the edge-weighted graph obtained from a database $\mathcal{D}$ and $v_i\in V(G)$. We set 
\begin{equation}
N_I(v_i)=\left\{v\in N(v_i):v\in I\right\}.
\end{equation}
At time $t +\Delta t$, the vertex $v_i$ will change of state according to probabilistic rules:

\begin{enumerate}
\item The probability ($P_I(v_i) $) that a susceptible vertex $v_i$ is infected by one of its neighbors is given by
\begin{equation}\label{probability_infected}
P_I(v_i)=\sum_{v_j\in N_I(v_i)}\rho\Delta t\cdot\widetilde{w}_{ij}, 
\end{equation}
where $\rho$ is a purely biological factor and representative of the disease. 
\item The probability ($P_R(v_i) $) that an infected vertex $v_i$ at time $t$ will recover is given by
\begin{equation}\label{probability_recover}
P_R(v_i)=\delta\Delta t, 
\end{equation}
where $\delta$ is the recovery rate.
\end{enumerate}
Moreover, we assume that the disease is present for a certain period of time and when individuals recover, they are immune.\\ 
\begin{Remark}
Notice that the expression (\ref{probability_infected}) can be deduced from the infection model called $q$-influence, assuming $q=\rho$. (see \cite{Zhang_2015}, \cite{pastor-satorras_epidemic_2015}).
\end{Remark}

In \cite{ronald1} the authors show in detail a study of the effects produced by the parameters (order of the graph, average of the degree of the vertex, representative factor of the disease, number of relationship variables in the database and quantity of classes) by changing only one of these in the total number of infected and the duration of the disease that is transmitted.

\section{Deterministic approximation}\label{determi}

We will show an approximation of Susceptible, Infected and Recovered curves through a differential equation system. The idea is to obtain a differential equation system with the parameters that define the graph.

Let $I(t)\in [0, N]$ be the number of infected individuals between two consecutive (discrete) times, i.e
\begin{equation}\label{eqq1}
I(t+\Delta t)=I(t)-\delta\Delta t I(t)+\rho\Delta t\varphi(t)\left(N-I(t)-R(t)\right),
\end{equation}
where $\varphi(t)$ is an estimate of the mean value of strength of infected neighbours for every susceptible individual. In order to get $\varphi(t)$, we assume that $m$ is the neighbours average and that a proportion of these neighbours, we say $\displaystyle \frac{I(t)}{N}$, is infected. On the other hand,
\begin{equation}
\overline{S}=\displaystyle\frac{\sum_{i=1}^NS(v_i)}{N}=\frac{\sum_{i=1}^N\sum_{v_j\in N(v_i)}\widetilde{w}_{ij}}{N},
\end{equation}
is the  average strength of the graph. Then 
\begin{equation}\label{eqq2}
\varphi(t)=\overline{S}\cdot \frac{I(t)}{N}.
\end{equation}
Replacing the equation (\ref{eqq2}) in (\ref{eqq1}) and as $N=S(t)+I(t)+R(t)$, then

\begin{equation}\label{eqq3}
I(t+\Delta t)=I(t)-\delta\Delta t I(t)+\rho\Delta t\cdot \overline{S}\frac{I(t)}{N}S(t).
\end{equation}
Dividing by $\Delta t$ and taking the limit as $\Delta t\rightarrow 0$ we obtain
\begin{equation}\label{eqq4}
\dot{I}(t)=\frac{\rho\cdot\overline{S}}{N}I(t)S(t)-\delta I(t).
\end{equation}
In the same way we obtain the equations

\begin{align}
\dot{S}(t)&=-\frac{\rho\cdot\overline{S}}{N}I(t)S(t)\label{eqq5}\\
\dot{R}(t)&=-\delta I(t).\label{eqq6}
\end{align}
We can see that the equations (\ref{eqq4}) ,(\ref{eqq5}) and (\ref{eqq6}) are the same as those defined for the SIR model. From there we can see that $\displaystyle\beta=\frac{\rho\cdot \overline{S}}{N}$.

If the previous deductions are done on a not edge-weighted graph, then we obtain that $\beta=\displaystyle\frac{\rho\cdot m}{N}$. If $\overline{w}$ is the average of weights then $\overline{S}=\overline{w}\cdot m$. It is clear that if $\overline{w}=1$ then we are in the case where the graph is not an edge-weighted graph.

The above is valid only if the population is mixed, the graph has a fixed contact structure, all vertices have approximately the same number of neighbours and approximately the same strength. However, the last condition is a stronger condition and certainly, it is not true, even to assume $\overline{S}=\overline{w}\cdot m$ can be a mistake because $\overline{w}$ could be unrepresentative. For instance in Section \ref{simulation}, the average strengths are not representative in the initial simulation since the strengths are strongly heterogeneous (see Figure \ref{fig_20}).
\begin{figure}[h!]
\centering
\includegraphics[width=0.6\textwidth]{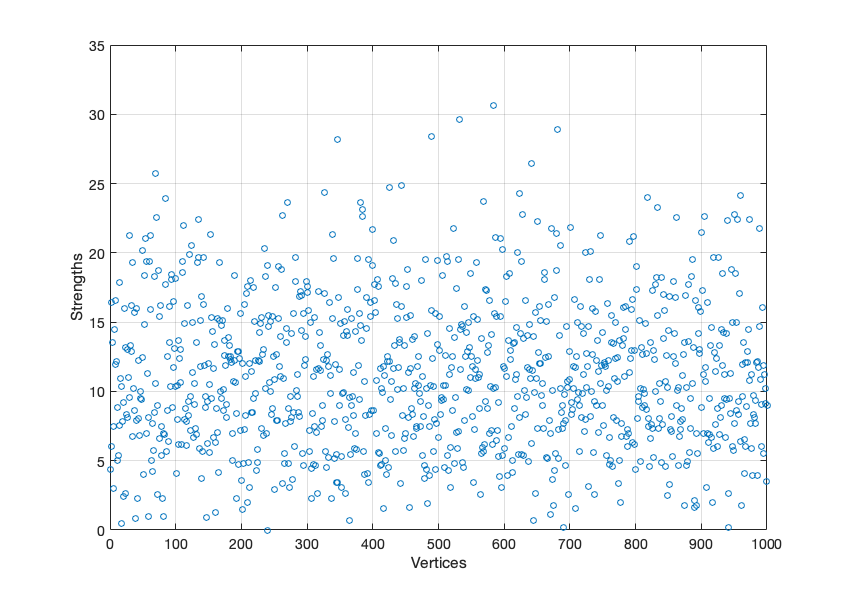}
\caption{Strength of each vertex} \label{fig_20}
\end{figure}

When the strength and weight average are not representative there is an overestimation of the infected individuals and the duration of the disease is lower.(see Figure \ref{fig_500}(a)). Moreover, if we consider $\overline{w}=1$ like in the case of the non edge-weighted graph, we have a sub-estimation of the infected individuals and the duration of the disease is upper (see Figure \ref{fig_500}(b)). Certainly, without considering $\overline{S}=\overline{w}\cdot m$ the fit is better, but not really good (because $\overline{S}$ is not a good representative) (see Figure \ref{fig_22}).

\begin{figure}[h!]
\centering
\subfigure[Fit for $\overline{S}$ and $\overline{w}$ not representative.]        
{\includegraphics[width=0.46\textwidth]{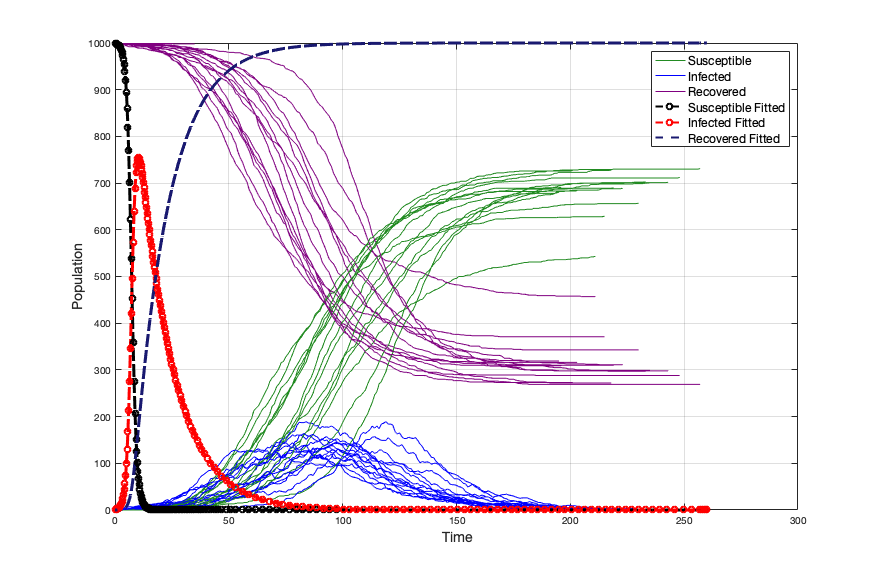}}        
\label{fig_500a}
\subfigure[Fit for $\overline{w}=1$.]        
{\includegraphics[width=0.49\textwidth, height=4.4 cm]{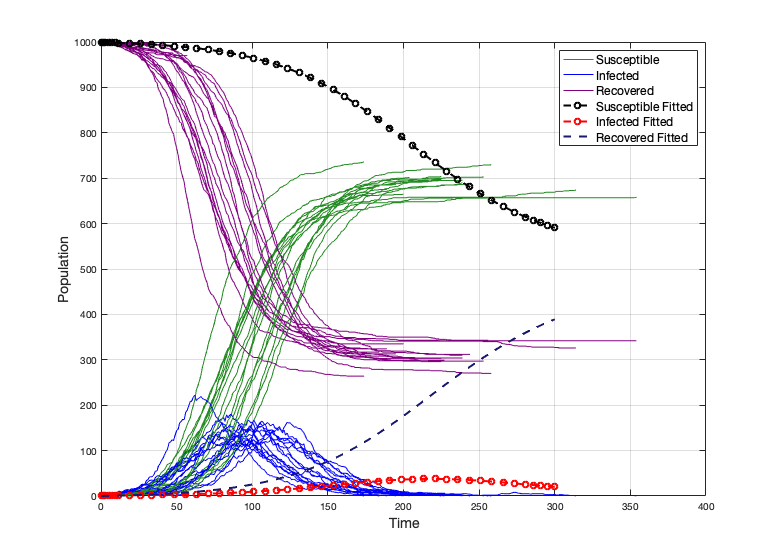}}        
\label{fig_500b}
 \caption{Different fit for $\overline{w}$.}\label{fig_500}
\end{figure}

\begin{figure}[h!]
\centering
\includegraphics[width=0.6\textwidth]{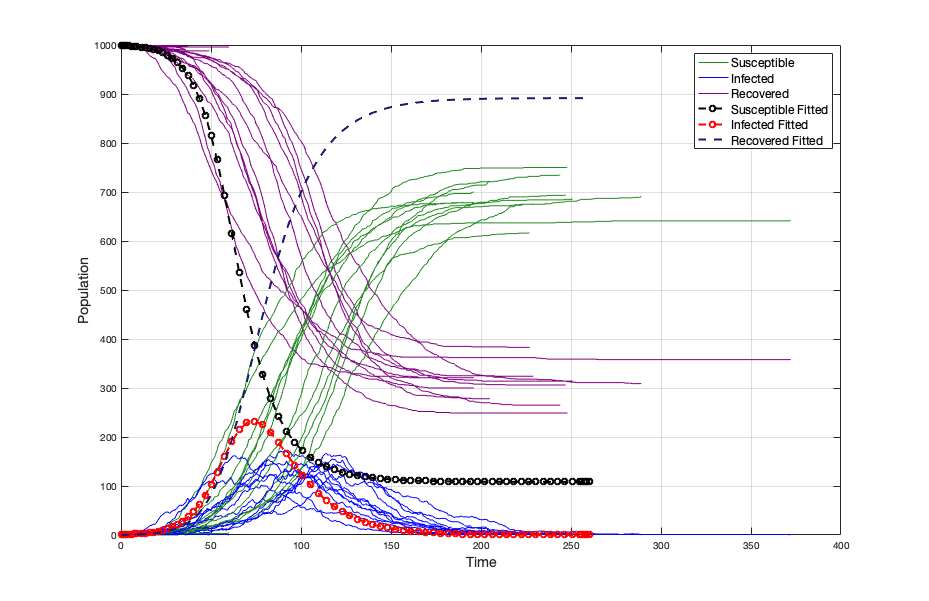}
\caption{Curves fitted without considering $\overline{S}=\overline{w}\cdot m$.} \label{fig_22}
\end{figure}
    
After many simulations (1000), we have noticed that a simple approximation to $\beta$ is the average between $\overline{S}$ and $m$, in the case where the strength of vertices is not strongly homogeneous (see Figure \ref{fig_11}). So, if $\displaystyle\psi=\frac{\overline{S}+m}{2}$, we have

\begin{align}
\dot{S}(t)&=-\frac{\rho \psi}{N}I(t)S(t)\\
\dot{I}(t)&=\frac{\rho \psi}{N}I(t)S(t)-\delta I(t)\\
\dot{R}(t)&=-\delta I(t).\label{eqq7}
\end{align}

\begin{figure}[h!]
\centering
\includegraphics[width=0.6\textwidth]{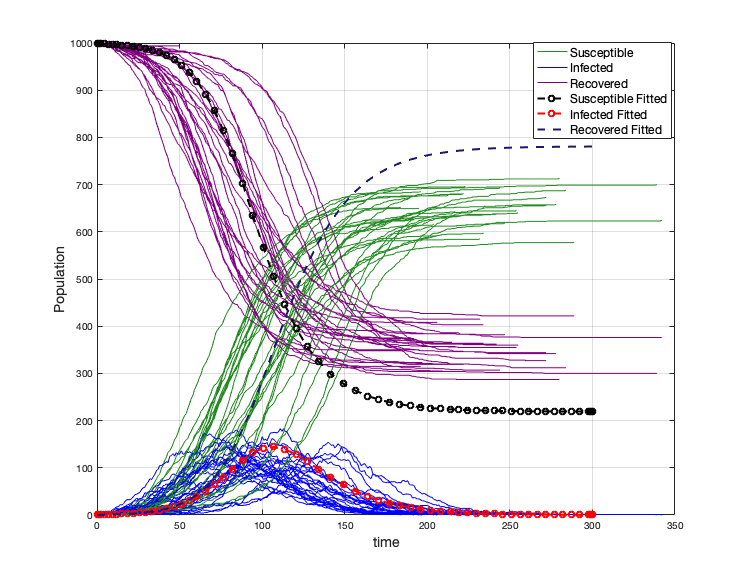}
\caption{Fitted curves in a random graph with $\displaystyle\psi=\frac{\overline{S}+m}{2}$.} \label{fig_11}
\end{figure}

A similar problem is treated in \cite{pastor-satorras_epidemic_2001} in the case of a non edge-weighted graph, where a graph has a heterogeneous distribution degree of vertex (scale-free graph).

\chapter[Generalization of the importance of vertices]{A generalization of the importance of vertices for an undirected weighted graph}\label{chap2}
\minitoc

\section{Introduction}

Establishing a node importance ranking is a problem that has attracted the attention of many researchers in recent decades. For unweighted networks where the edges do not have any attached weight, many proposals have been presented, considering local or global information of the networks. 
In this context, other criteria have emerged to determine the importance of a node. For instance, in \cite{xu2019node} there are taken into account the characteristics of the scale, tightness, and topology of the local area network where the node and its neighbors are located. Liu {\it et al.} \cite{liu2016evaluating} proposed an effective ranking method based on degree value and the importance of edges. Saxena et al. designed a method to identify the best ranked nodes in order to control an epidemic through immunizations \cite{saxena}. On the other hand, the conscious centrality of the community is addressed by Magelinski et al, who propose the relevance of the bridge node in the network quantified in modular vitality \cite{magelinski}. Along the same lines, Ghalmane et al. analyze the non-overlapping community structure network based on an immunization strategy, recognizing the relevance of the node given the characteristics of the communities \cite{ghalmane}, while the same author asserts that modular centrality must include the influences of the nodes both in their own communities and in other \cite{ghalmane2}. Gupta et al. also agree that community structure is important for understanding the spread of an epidemic \cite{gupta}.

In \cite{wang2010new} the authors proposed a method based on degree and degree of their neighborhoods. Ren {\it et al.} in \cite{ren2013node} proposed a measure based on degree and clustering coefficient. For its part, \cite{yang2019novel} is based on combining the existing centrality criteria, where the weight for each criterion is calculated by the entropy weighting method which overcomes the impact of the subjective factor, and the Vlsekriterijumska Optimizacija I Kompromisno Resenje (VIKOR) method is used for ranking nodes importance. Jia-Sheng {\it et al.} improved the method of node importance evaluation based on node contraction \cite{jia2011improved}. Recently,  a new multi-evidence centrality is proposed based on evidence theory. The existing measures of degree centrality, betweenness centrality, efficiency centrality and correlation centrality are taken into consideration in \cite{mo2019identifying}.

Edge-weighted graphs are not exempt from this discussion. For example, there are the rankings which are natural extensions of unweighted graphs, that we could call more representative, they are the strength of the node \cite{barrat2004architecture}, weighted H-index \cite{lu2016h}, weighted degree \cite{garas2012k281, wei2015weighted283}, the weighted closeness centrality \cite{newman2001scientific}, weighted betweenness centrality \cite{brandes2001faster}, weighted PageRank and leaderRank \cite{PhysRevE.75.021102}. Furthermore, Opsahl in \cite{opsahl2010node} proposed a generalization of degree for weighted graph that combines the strength and the number of neighbors.

Also on this type of networks new methods have been proposed to establish the importance or influence of the nodes. Among them, Yang {\it et al.} \cite{yang2017mining} proposed the Two-Way-PageRank method based on PageRank and analyzed the importance of two important factors that affect the importance of the nodes and gave a definition and expression of the importance of the nodes. A new centrality measure (Laplacian centrality) for undirected and edge-weighted graph is proposed in \cite{qi2012laplacian} as well. Tang \cite{tang2020research} proposed the weighted K-order propagation number algorithm to extract the disease propagation based on the network topology to evaluate the node importance. In the field of maximization of influence, Ahmad proposes an approach that includes the weighted sum and a multi-criteria decision method \cite{ahmad}. All this generates the recognition of influential nodes in the network. The authors, through experimentation, conclude that their model is competitive to other similar ones \cite{ahmad}.
Even axiomatic approaches have been proposed (see \cite{SKIBSKI2019151} and \cite{ijcai2017-59}).

In the framework of modeling complex biological systems Almasi \& Hu provide a measure from the importance of vertices in a weighted human disease network \cite{almasi2019measuring}. They propose a vertex importance measure  by extending a centrality measure for unweighted networks proposed by Liu {\it et al.} They named their extension on weighted graphs, the DIL-W centrality.

In this chapter a ranking method for undirected and edge-weighted graph is provided. Particularly, the aim of this article is to generalize the vertex importance measure (DIL) of an undirected and unweighted network given by  Liu {\it et al.}, to an undirected and edge-weighted graph using the generalization of Degree centrality proposed by Opsahl in \cite{opsahl2010node}. Our proposal is slightly different from Almasi \& Hu's one, but when we compare both generalizations in real networks, with respect to the measure the network efficiency, we can see that we obtain a better or the same efficiency measure, as we will show in the following sections.

\section{Measuring vertex importance on an undirected weighted graph}\label{proposal}
In this section the generalization of the importance vertex measure is given, denoted by DIL-W$^{\alpha}$.

Let us consider an undirected weighted graph $(G,w)$ with $G=(V,E)$ and $V=\{v_1,v_2,\ldots,v_n\}$. Let $\alpha\in [0,1]$ be the tuning parameter.
\begin{Definition}[Importance edge]\label{imporedge}
The importance of an edge $e_{ij}\in E$, denoted by $I^{\alpha}(e_{ij})$, is defined as
\begin{align*}
I^{\alpha}(e_{ij})&=\frac{\left(C_D^{w\alpha}(v_i)-p^{\alpha}_i\right)\cdot \left(C_D^{w\alpha}(v_j)-p^{\alpha}_j\right)}{\lambda^{\alpha}},
\end{align*}
where, for $k\in\{i,j\}$, $\displaystyle p^{\alpha}_k=(p+1)^{(1-\alpha)}\cdot t^{\alpha}_k$ with $p$ is the number of triangles, one edge of the triangle is $e_{ij}$, $t^{\alpha}_{k}$ is the weight of the sum of the edges incident to $v_{k}$ that form a triangle with $e_{ij}$ and $\displaystyle\lambda^{\alpha}=\frac{p^{(1-\alpha)}\cdot\left(t_i+t_j\right)^{\alpha}}{2}+1$.
\end{Definition}

\begin{Definition}[Contribution]\label{contri}
The contribution that $v_i\in V$ makes to the importance of the edge $e_{ij}$, denoted by $W^{\alpha}(e_{ij})$, is defined as
\begin{align*}
W^{\alpha}(e_{ij})&=I^{\alpha}(e_{ij})\cdot\frac{C_D^{w\alpha}(v_i)-w_{ij}^{\alpha}}{C_D^{w\alpha}(v_i)+C_D^{w\alpha}(v_j)-2w_{ij}^{\alpha}},
\end{align*}
where $w_{ij}$ is the weight of $e_{ij}$.
\end{Definition}

\begin{Definition}[Importance of vertex DIL-W$^{\alpha}$]\label{importencevertex}
The importance of a vertex $v_i\in V$, denoted by $L^{\alpha}(v_i)$, is defined as
\begin{align*}
L^{\alpha}(v_i)&=C_D^{w\alpha}(v_i)+\sum_{v_j\in N(v_i)}W^{\alpha}(e_{ij}).
\end{align*}
\end{Definition}

From the definition of Degree centrality (Definition \ref{dcent}) proposed by Opsahl in \cite{opsahl2010node}, we can see that when the tuning parameter $\alpha$ is $0$ the Definitions \ref{imporedge}, \ref{contri} and \ref{importencevertex} of importance edge, contribution and importance of vertex respectively, are the same than the proposed by Liu {\it et al.} in \cite{liu2016evaluating} for an undirected and unweighted network. However, when $\alpha=1$, our proposal is slightly different from Almasi \& Hu's proposal in \cite{almasi2019measuring}. In the latter, the authors define the contribution as 
\begin{align}\label{proALMASI}
W(e_{ij})&=I^{\alpha}(e_{ij})\cdot\frac{S(v_i)}{S(v_i)+S(v_j)}.
\end{align}
Notice that, if the strength of the vertex is the traditional degree centrality as in the case of the undirected and unweighted graph, it is not possible to get the proposal of Liu {\it et al.} from the equation (\ref{proALMASI}). This slight difference between the generalization of the vertex importance method of \cite{almasi2019measuring} and our generalization is relevant when we compare both generalizations in real networks with respect to the measure of the network efficiency as we will show in the next Section.

On the other hand, a possibility to extend the vertex importance based on degree and the importance of lines from an unweighted network to the proposed by Almasi \& Hu, is to consider the contribution as 
\begin{align}\label{failprop}
W(e_{ij})&=I^{\alpha}(e_{ij})\cdot\frac{C_D^{w\alpha}(v_i)-(1-\alpha)}{C_D^{w\alpha}(v_i)+C_D^{w\alpha}(v_j)-2(1-\alpha)}.
\end{align}
Indeed, when $\alpha=0$ we obtain the original proposal and when $\alpha=1$ we get the generalization given in \cite{almasi2019measuring}. Despite the above, we do not consider the equation (\ref{failprop}) because, in our opinion, it is not a good generalization proposal since it does not improve the one done by Almasi \& Hu.\\

\begin{figure}[h!]
\centering
\includegraphics[width=0.5\textwidth]{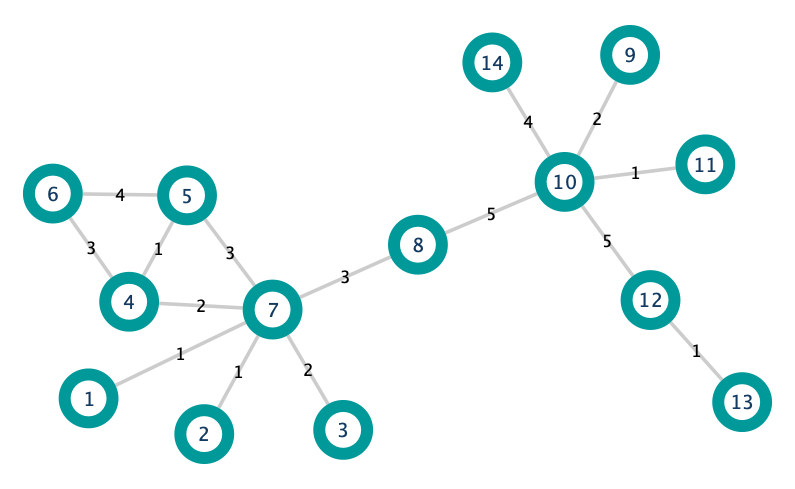}
\caption{Simple graph.} \label{line_toy}
\end{figure}

{In order to illustrate the above Definitions, let us consider the edge-weighted graph in Figure \ref{line_toy}. Moreover, we consider the edges $e_{78}$ and ${e_{75}}$. Notice that they both have the same weight (three). For this example, we set $\alpha=1$. Applying Definition~\ref{dcent}, we get:}
{\begin{align*}
C_D^{w1}(v_7)&=deg(v_7)^{(1-1)}\cdot S(v_t)^{1}=1\cdot 12=12\quad\text{and}\\ 
C_D^{w1}(v_5)&=deg(v_5)^{(1-1)}\cdot S(v_t)^{1}=1\cdot 8=8.
\end{align*}}

{From Definition~\ref{imporedge}:}  
{\begin{align*}
p^{1}_7&=(1+1)^{(1-1)}\cdot t^{\alpha}_7=2^0\cdot 2^1=2,\\
p^{1}_5&=(1+1)^{(1-1)}\cdot t^{\alpha}_5=2^0\cdot 1^1=1,\quad\text{and}\\
\lambda^{1}&=\frac{1^{(1-1)}\cdot\left(2+1\right)^{1}}{2}+1=\frac{5}{2}.
\end{align*}}

{Therefore,
\begin{align*}
I^{1}(e_{75})&=\frac{\left(C_D^{w1}(v_7)-p^{1}_7\right)\cdot \left(C_D^{w1}(v_5)-p^{1}_5\right)}{\lambda^{1}}\\
&=\frac{(12-2)\cdot (8-1)}{\frac{3}{2}+1}=28.
\end{align*}

In the same way with the edge $e_{78}$, we obtain 
\begin{align*}
I^{1}(e_{78})&=96.
\end{align*}

In conclusion, edge $ e_{68} $ is more important than edge $ e_{75}$}.
{The latter is reasonable because the edge $e_{78}$ is a bridging edge of the graph.}

{We have calculated the importance of the edge $ e_{78} $ of the graph in Figure~\ref{line_toy}. The contribution that $v_7$ makes to it is given by Definition \ref{contri}:}
{\begin{align*}
W^{1}(e_{78})&=I^{1}(e_{78})\cdot\frac{C_D^{w1}(v_7)-w_{78}^{1}}{C_D^{w1}(v_7)+C_D^{w1}(v_8)-2w_{78}^1}\\
&=96\cdot \frac{12-3}{12+8-2\cdot 3}\\
&=\frac{432}{7}.
\end{align*}}

{In the same way, the contribution that $v_8$ makes to $I^{1}(e_{78})$ is: }
{\begin{align*}
W^{1}(e_{87})&=I^{1}(e_{78})\cdot\frac{C_D^{w1}(v_8)-w_{78}^{1}}{C_D^{w1}(v_7)+C_D^{w1}(v_8)-2w_{78}^1}\\
&=96\cdot \frac{8-3}{12+8-2\cdot 3}\\
&=\frac{240}{7}.
\end{align*}}

{The above means that the node $v_7$ contributes more to the edge $e_{78}$ than node $v_8$.}

{In order to illustrate the Definition \ref{importencevertex}, we compute the importance of $v_7$ and $v_8$ in the graph of Figure~\ref{line_toy}.}
{\begin{eqnarray*}
L^{1}(v_7)&=&C_D^{w1}(v_7)+\sum_{v_j\in N(v_7)}W^{1}(e_{7j})\\
&=&12+W^{1}(e_{71})+W^{1}(e_{72})+W^{1}(e_{73})+W^{1}(e_{74})+W^{1}(e_{75})+W^{1}(e_{78})\\
&=&12+12+12+24+\frac{75}{7}+18+\frac{432}{7}\\
&=&78+\frac{507}{7}=\frac{1053}{7},
\end{eqnarray*}}
{and}
{\begin{eqnarray*}
L^{1}(v_8)&=&C_D^{w1}(v_8)+\sum_{v_j\in N(v_8)}W^{1}(e_{8j})\\
&=&8+W^{1}(e_{87})+W^{1}(e_{8,10})\\
&=&8+\frac{240}{7}+\frac{136}{5}\\
&=&8+\frac{2152}{35}=\frac{2432}{35}.
\end{eqnarray*}}

{Since $L^{1}(v_7)<L^{1}(v_8)$, then node $v_7$ is more important than node $v_8$ (according to DIL-W$^{1}$ ranking).}\\

On the other hand, one of the good qualities of the DIL ranking is that it recognizes the importance of bridge nodes (see more in \cite{Musial}). This quality is not lost in our extension to weighted graphs. For instance, we can see in Figure \ref{exgraph} that when applying the DIL-W $^{\alpha}$ ranking with $\alpha = 0.5$, vertex 10 has a better ranking than vertex 4, this due to the vertex 10 is a bridge node. Indeed, from Definitions \ref {importencevertex} we get
\begin{align*}
L^{0.5}(v_4)&=C_D^{w\alpha}(v_4)+\sum_{v_j\in N(v_4)}W^{0.5}(e_{4j})=13.7834,
\end{align*}
and 
\begin{align*}
L^{0.5}(v_{10})&=C_D^{w\alpha}(v_{10})+\sum_{v_j\in N(v_10)}W^{0.5}(e_{10j})=17.4523.
\end{align*}
Therefore $L^{0.5}(v_4)<L^{0.5}(v_{10})$.
\begin{figure}[!htb]
\begin{center}
 \includegraphics[width=0.6\textwidth]{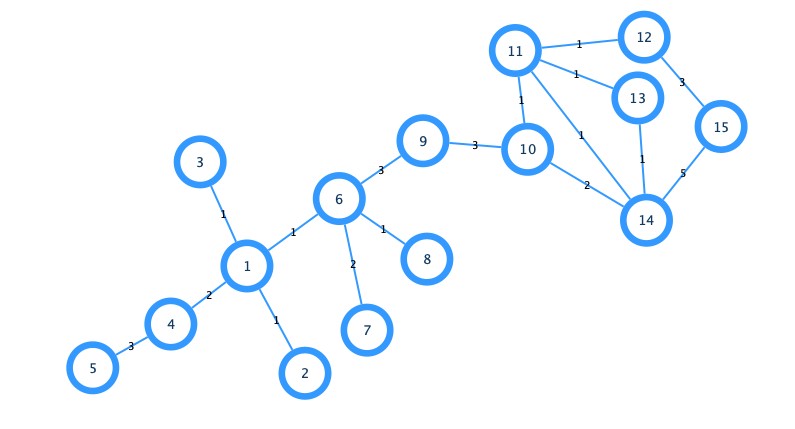}%
   \caption{Simple network example.}
    \label{exgraph}
\end{center}
\end{figure}

On the other hand, the consideration of the centrality measure, that combines degree and strength, picks up the main idea of the original measure, which is the number of neighbors of a vertex, an idea that is lost when considering only the strength.

\section{Comparison and analysis results}\label{Comparison}
\subsection{Data and source code}
In our research, we compare our proposal DIL-W$^{\alpha}$ for $\alpha=1$ and $\alpha=0.5$ with DIL-W on some real networks, these are: Zacharys karate club network \cite{zachary} (with 34 vertices and 78 edges), US air transport network \cite{colizza2007} (with 500 vertices and 2980 edges), 
meta-analysis network of human whole-brain functional coactivations \cite{Crossley11583} (with 638 vertices and 18625 edges), colony of ants network \cite{mersch2013tracking} (with 113 vertices and 4550 edges), and wild bird network \cite{firth2015experimental} (with 131 vertices and 1444 edges).

For the databases of colony of ants and wild birds networks see \cite{nr-aaai15}. For US air transport network see \cite{vitonicoruso}. For brain functional coactivactions network database see \cite{RUBINOV20101059}.

\subsection{Evaluation}\label{evaluation}
According to \cite{PhysRevLett}, network efficiency is an index used to indicate the quality of network connectivity. A high connectivity indicates a higher efficiency of the graph. The efficiency of a graph, denoted by $\eta$, is defined as follows:
\begin{align}
\eta&=\frac{1}{n(n-1)}\sum_{v_i\neq v_j}\frac{1}{d^{w}(v_i,v_j)},
\end{align}
where $d^{w}(v_i,v_j)$ is the distance between the vertex $v_j$ and the vertex $v_j$, defined as,
\begin{align}
d^{w}(v_i,v_j)&=\min \left(\frac{1}{w_{ih}}+\ldots+\frac{1}{w_{hj}}\right),
\end{align}
where $h$ are the indices of the intermediary vertices on paths between $v_i$ and $v_j$, and $w_{ih},\ldots,w_{hj}$ the weights of the edges $e_{ih},\ldots, e_{hj}$.
(see more details in \cite{brandes2001faster} and \cite{newman2001scientific}).
We consider $\alpha=1$ and $\alpha=0.5$ to compare DIL-W$^{\alpha}$ with DIL-W, according to what was discussed at the end of Section \ref{proposal}.

The authors in \cite{lai2004} show that between 5\% and 10\% of important nodes may cause the entire network to falter. To compare the DIL-W$^{\alpha}$ proposal with DIL-W, the 10\% of the vertices are deleted one by one from the networks according to the importance ranking produced by DIL-W$^{\alpha}$ and DIL-W, and we compute the decline rate of network efficiency, denoted by $\mu$, and defined by
\begin{align*}
\mu&=1-\frac{\eta}{\eta_0},
\end{align*}
where $\eta_0$ is the efficiency of the original graph, and $\eta$ is the graph efficiency after some vertices are removed \cite{ren2013node}. We expect a greater decline rate of the network efficiency if we delete a more important vertex of the obtained DIL-W$^{\alpha}$ ranking.

In Figure \ref{figura1}, we see the ranking done by DIL-W, DIL-W$^{1}$, and DIL-W$^{0.5}$ on the example networks versus the decline rate of the network efficiency.  There exists a  correlation between the rankings DIL-W,  DIL-W$^{1}$, and DIL-W$^{0.5}$ and the decline rate of the network efficiency in Zacharys karate club,  US air transport, and colony of ants networks. Whereas in the case of wild birds and brain functional coactivations networks the correlation is not evident.
\begin{figure}[!h]
\begin{center}
  \subfigure
    {\includegraphics[width=0.48\textwidth]{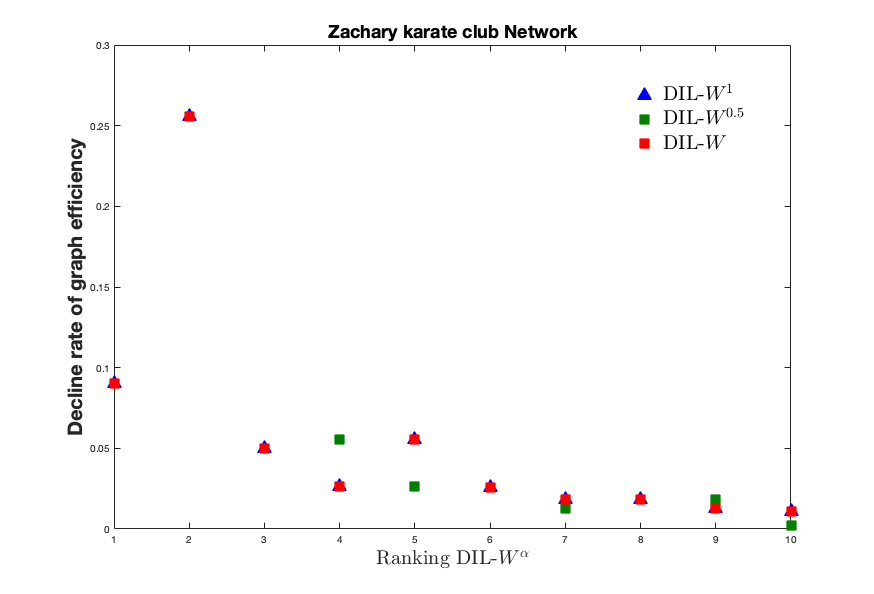}}  
  \subfigure
    {\includegraphics[width=0.48\textwidth]{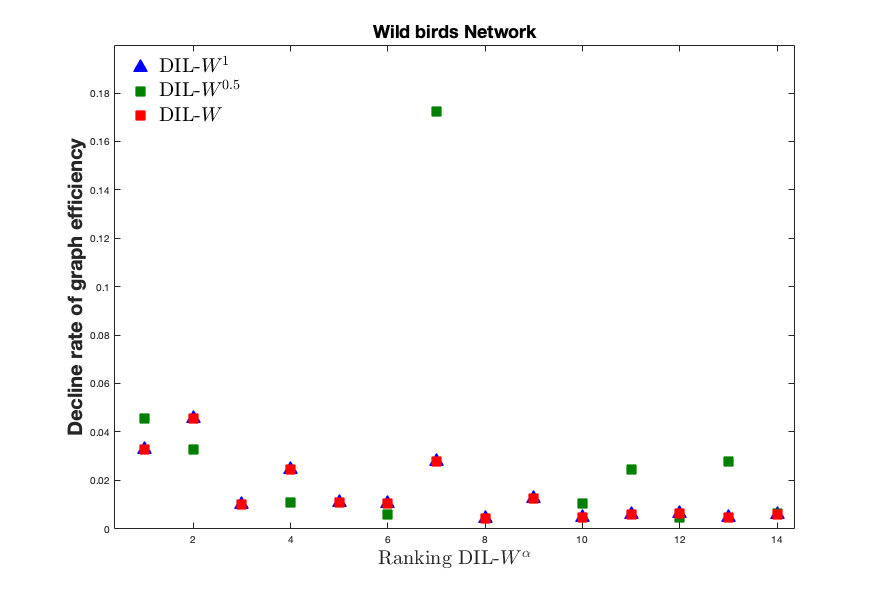}}
   \subfigure
    {\includegraphics[width=0.48\textwidth]{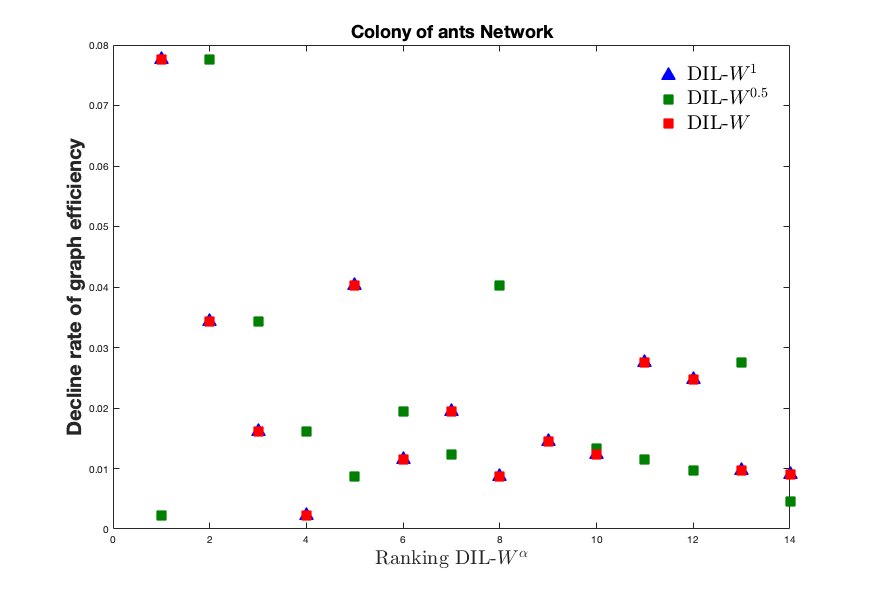}}
     \subfigure
    {\includegraphics[width=0.48\textwidth]{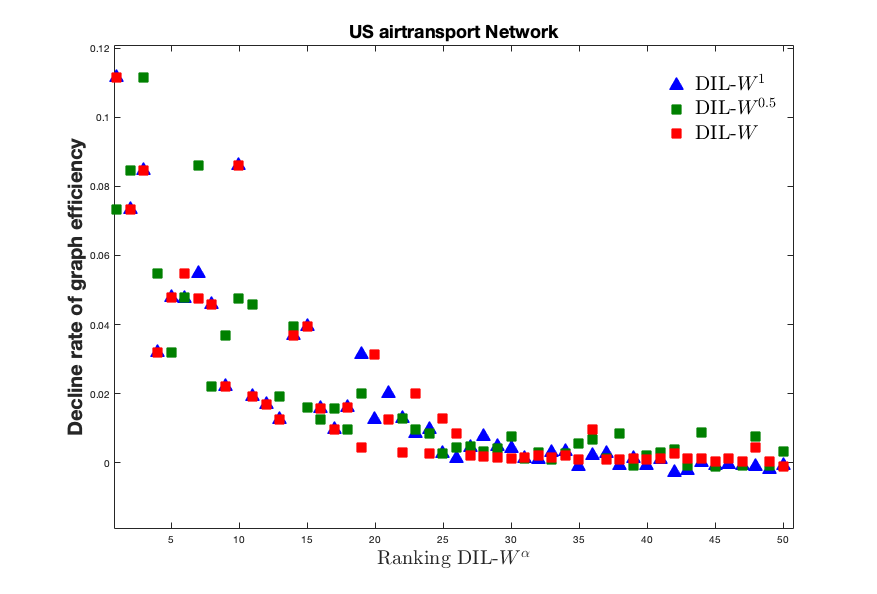}}
     \subfigure
    {\includegraphics[width=0.48\textwidth]{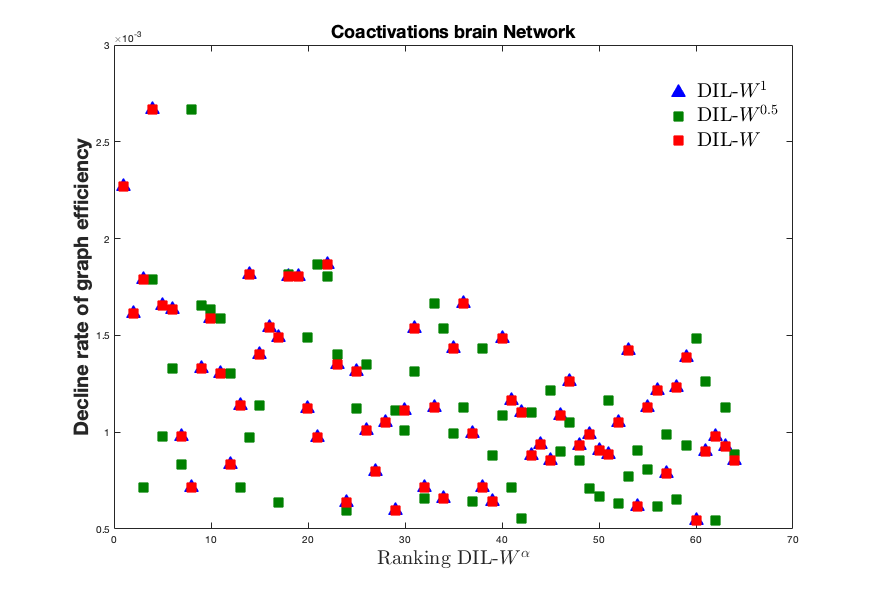}}
  \caption{The relation between decline rate of network efficiency and the ranking  DIL-W$^{1}$,  DIL-W$^{0.5}$ and DIL-W on the example networks.}
  \label{figura1}
  \end{center}
\end{figure}

In each implementation only one node is removed according to the importance ranking list, then we compute the decline rate of the network efficiency.

In Figure \ref{figura3}, we can see the relation between the decline rate of network efficiency and the number of vertices deleted from five real networks.
\begin{figure}[h!tb]
\begin{center}
  \subfigure
    {\includegraphics[width=0.45\textwidth]{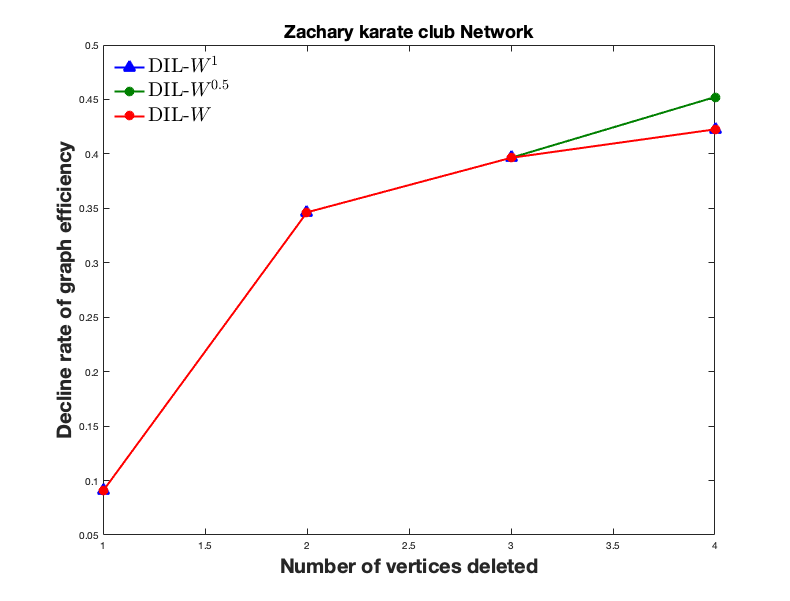}}
   \subfigure
    {\includegraphics[width=0.45\textwidth]{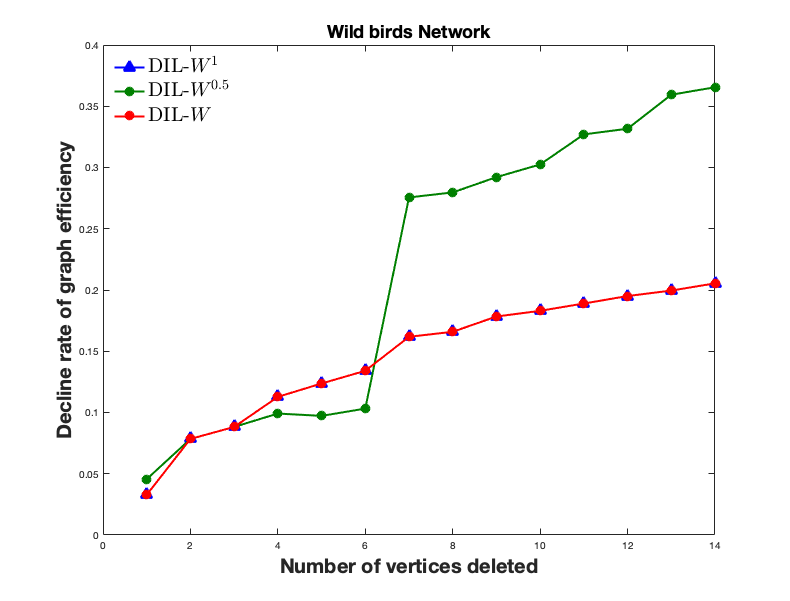}}
  \subfigure
    {\includegraphics[width=0.45\textwidth]{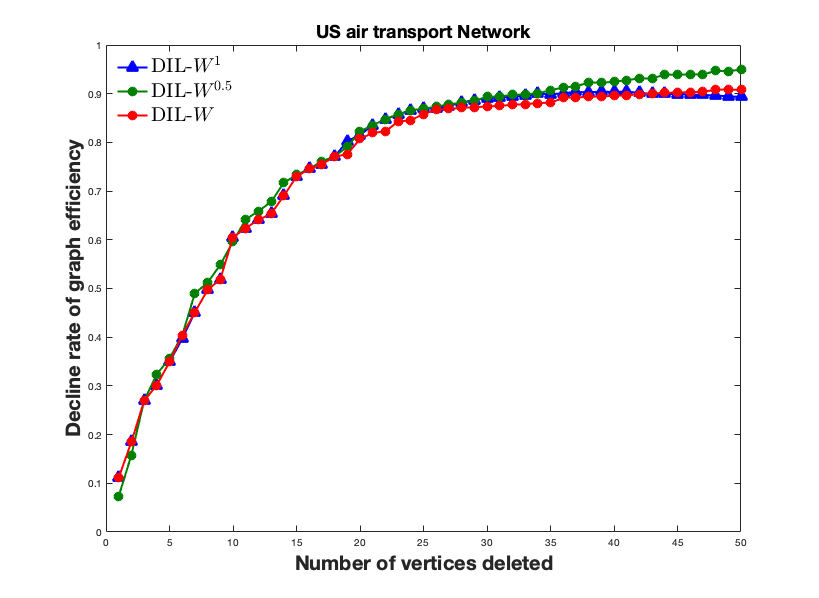}}
    \subfigure
    {\includegraphics[width=0.45\textwidth]{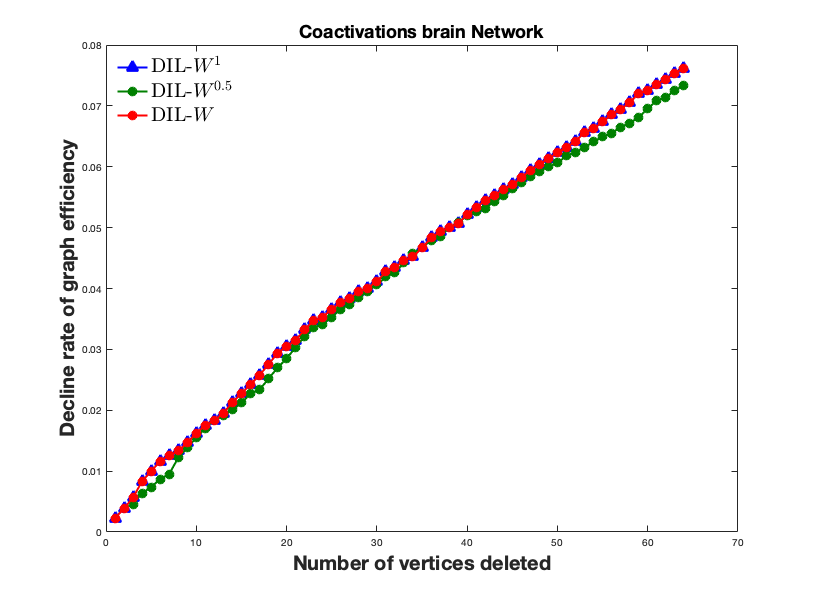}}
    \subfigure
    {\includegraphics[width=0.45\textwidth]{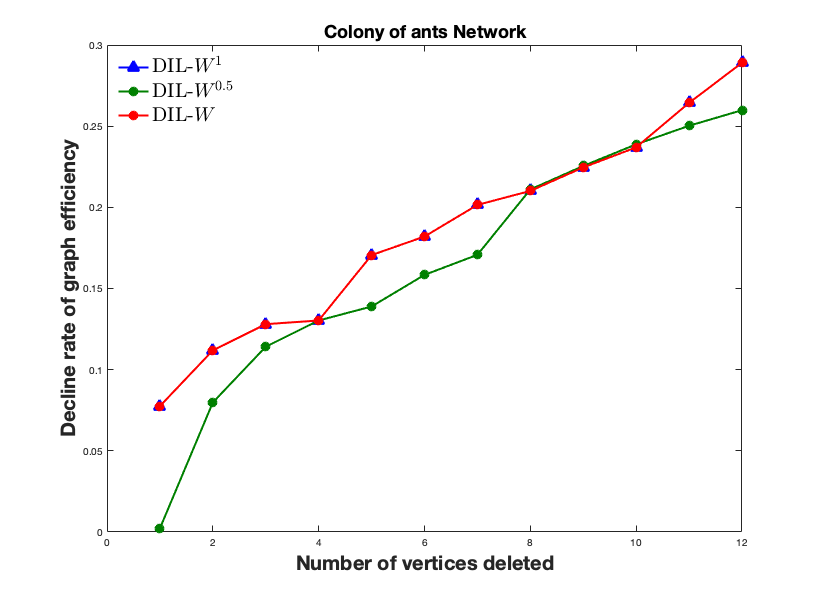}}
   \caption{The decline rate of the network efficiency as a function of deleting the top 10\% of the vertices ranked by DIL-W$^{1}$, DIL-W$^{0.5}$ and DIL-W from five real networks (Zacharys karate club, wild birds, US air transport, brain functional coactivations, and colony of ants).}
  \label{figura3}
 \end{center}
\end{figure}

In Zachary karate club, US air transport, and wild birds networks, DIL-W$^{0.5}$ performs the best between all three rankings. In brain functional coactivations, and colony of ants networks, DIL-W$^{1}$ and DIL-W perform the best. Notice that, only in the case of US air transport, DIL-W$^{1}$ does not match completely with DIL-W. Table \ref{table1} shows the nodes ranking on the wild bird and US air transport networks making evident the same ranking of importance of nodes for DIL-W and DIL-W$^1$ in wild birds network and the slight difference of ranking in US air transport network. 

\begin{table}[h!bt]
\centering
\resizebox{10cm}{!} {
\begin{tabular}{c|c|c|c|c|}
\cline{2-5}
\multicolumn{1}{l|}{}                           & \multicolumn{2}{c|}{\textbf{Wild birds network}} & \multicolumn{2}{c|}{\textbf{US air transport network}} \\ \cline{2-5} 
\multicolumn{1}{l|}{}                           & \textbf{DIL-W}        & \textbf{DIL-W$^1$}       & \textbf{DIL-W}          & \textbf{DIL-W$^1$}          \\ \hline
\multicolumn{1}{|c|}{\textbf{Ranking position}} & ID                    & ID                       & ID                      & ID                          \\ \hline
\multicolumn{1}{|c|}{1}  & 12 & 12 & 1  & 1  \\
\multicolumn{1}{|c|}{2}  & 27 & 27 & 3  & 3  \\
\multicolumn{1}{|c|}{3}  & 23 & 23 & 6  & 6  \\
\multicolumn{1}{|c|}{4}  & 10 & 10 & 10 & 10 \\
\multicolumn{1}{|c|}{5}  & 35 & 35 & 14 & 14 \\
\multicolumn{1}{|c|}{6}  & 77 & 77 & 7  & 5  \\
\multicolumn{1}{|c|}{7}  & 25 & 25 & 5  & 7  \\
\multicolumn{1}{|c|}{8}  & 51 & 51 & 4  & 4  \\
\multicolumn{1}{|c|}{9}  & 26 & 26 & 8  & 8  \\
\multicolumn{1}{|c|}{10} & 94 & 94 & 2  & 2  \\
\multicolumn{1}{|c|}{11} & 6  & 6  & 12 & 12 \\
\multicolumn{1}{|c|}{12} & 56 & 56 & 11 & 11 \\
\multicolumn{1}{|c|}{13} & 61 & 61 & 13 & 13 \\
\multicolumn{1}{|c|}{14} & 11 & 11 & 21 & 21 \\ \hline
\end{tabular}}
\caption{The ranking results of DIL-W and DIL-W$^1$ on wild birds and US airtransport networks.}
\label{table1}
\end{table}

\subsection{Ranking DIL-W$^{\alpha}$ for different values of $\alpha$. }\label{discussion}
In the previous section we could see that, with respect to efficiency, DIL-W $^{\alpha}$ performs better rankings in brain functional coactivations and colony of ants networks, when $\alpha=1$, whereas in US air transport, Zachary karate club, and wild birds networks the ranking is better with $\alpha=0.5$.\\
In this section we compare the different rankings that DIL-W $^{\alpha} $ performs for different values of $\alpha$. We have considered the values $0.1$, $0.3$, $0.5$, $0.7$, $0.9$ and $1$. Figure \ref{figura4} shows the results. 


\begin{figure}[!htb]
\begin{center}
  \subfigure
    {\includegraphics[width=0.45\textwidth]{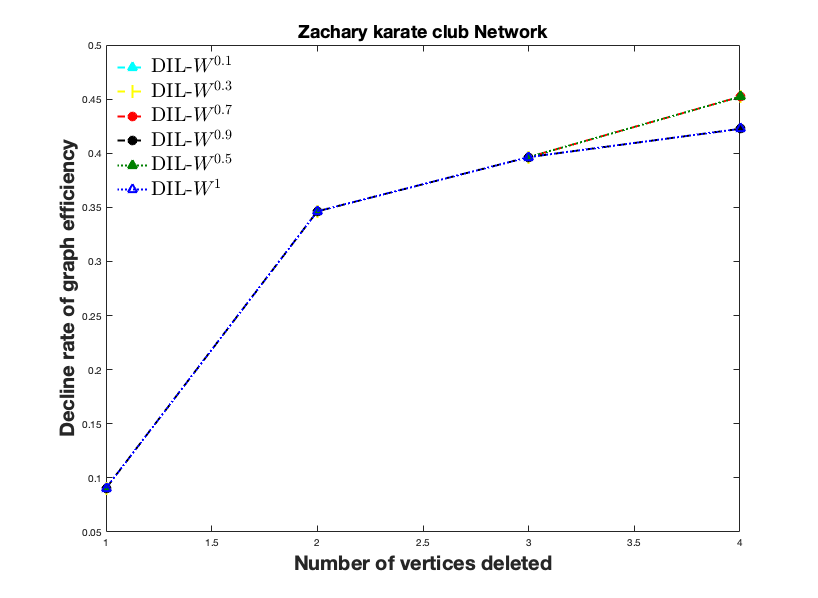}}
   \subfigure
    {\includegraphics[width=0.45\textwidth]{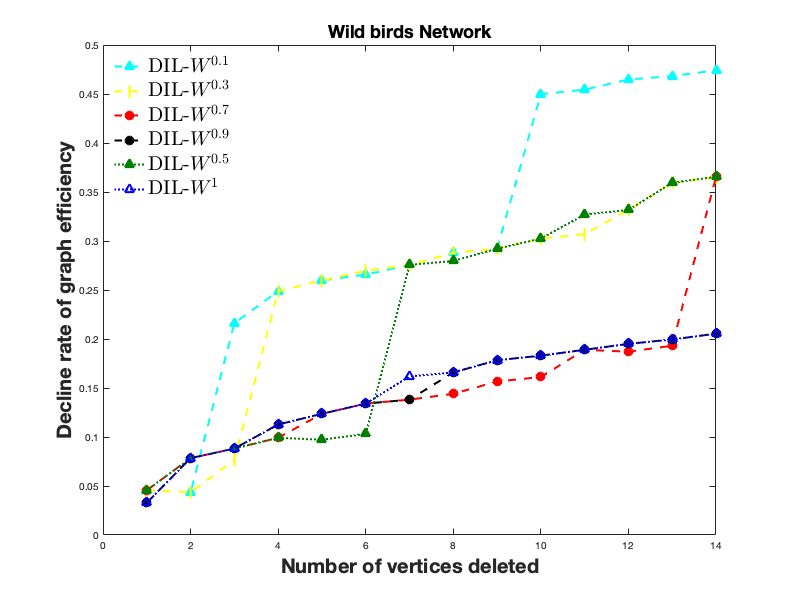}}
  \subfigure
    {\includegraphics[width=0.45\textwidth]{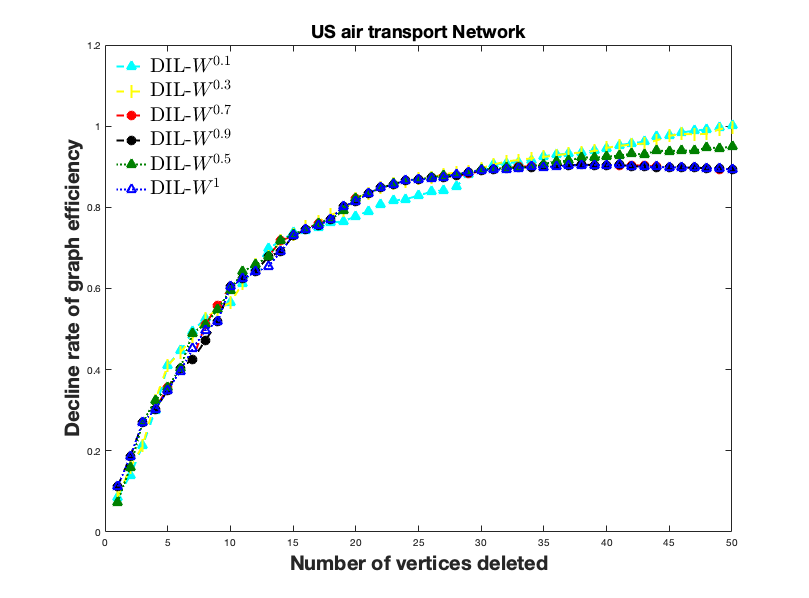}}
    \subfigure
    {\includegraphics[width=0.45\textwidth]{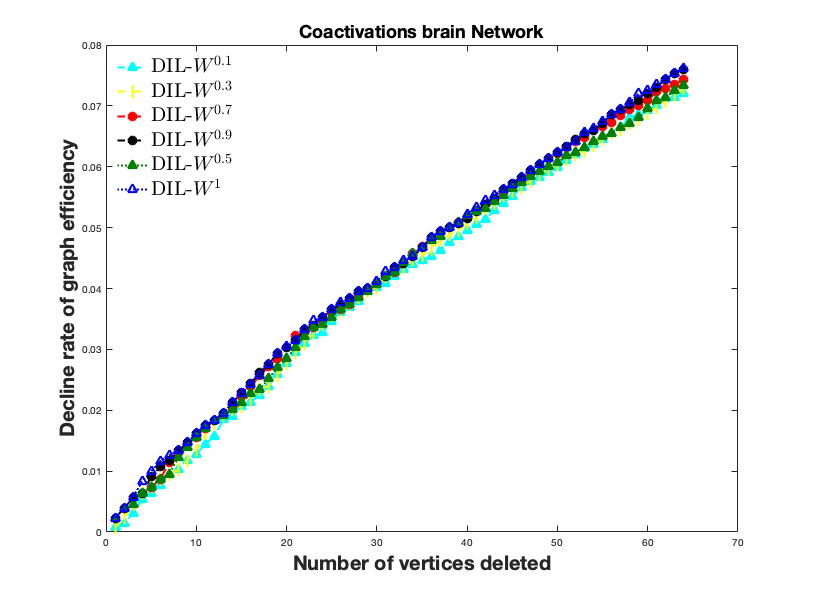}}
    \subfigure
    {\includegraphics[width=0.45\textwidth]{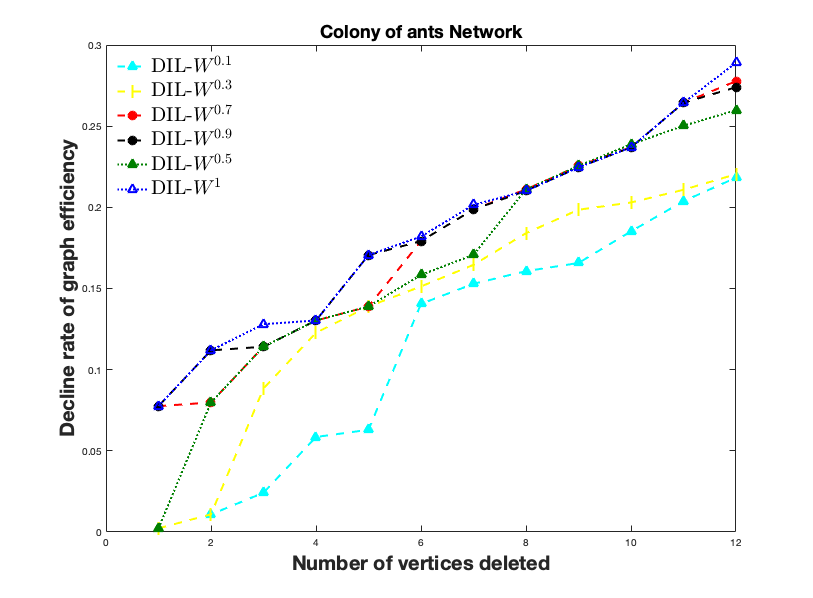}}
   \caption{The decline rate of the network efficiency as a function of deleting the top 10\% of the vertices ranked by DIL-W$^{0.1}$, DIL-W$^{0.3}$, DIL-W$^{0.5}$, DIL-W$^{0.7}$, DIL-W$^{0.9}$ and DIL-W$^1$ from five real networks (Zacharys karate club, wild birds, US air transport, brain functional coactivations, and colony of ants).}
    \label{figura4}
\end{center}
\end{figure}

For the US air transport, Zachary karate club and wild birds networks, it is observed that the smaller the $\alpha$, the better the efficiency ranking is performed. Conversely, in colony of ants and brain coactivations networks the best performance for efficiency is with $\alpha = 1 $ or $\alpha $ close to $1$. 
In conclusion, we can see that DIL-W $^{\alpha}$ is equal to or better in efficiency than the DIL-W proposal given in \cite{almasi2019measuring} in the five networks that we have considered. In particular, they coincide in efficiency in brain coactivations and colony of ants networks  for $\alpha=1 $ and cannot be improved with another value of $\alpha$ as we concluded from Figure \ref{figura4}. In contrast, DIL-W $^{\alpha} $ performs better than DIL-W on the US air transport, Zachary karate club, and wild birds networks, as soon as $\alpha<1$. \\

Moreover, we deepen into the possible relations between DIL and DIL-W$^\alpha$ for the different values of $\alpha $ seen previously (i.e $0.1, 0.3, 0.5, 0.7, 0.9,$ and $1$) on all the networks that we have tested. Figure \ref{coef} shows the relations. In each network it is observed that the correlation coefficient between both rankings are very close to 1. In particular, for the Zachary karate club network we have all coefficients above 0.988. While for US air transport we have the largest range of correlation coefficients among all networks (between 0.7308 and 0.9998).

\begin{figure}[!htb]
\begin{center}
  \subfigure
    {\includegraphics[width=0.45\textwidth]{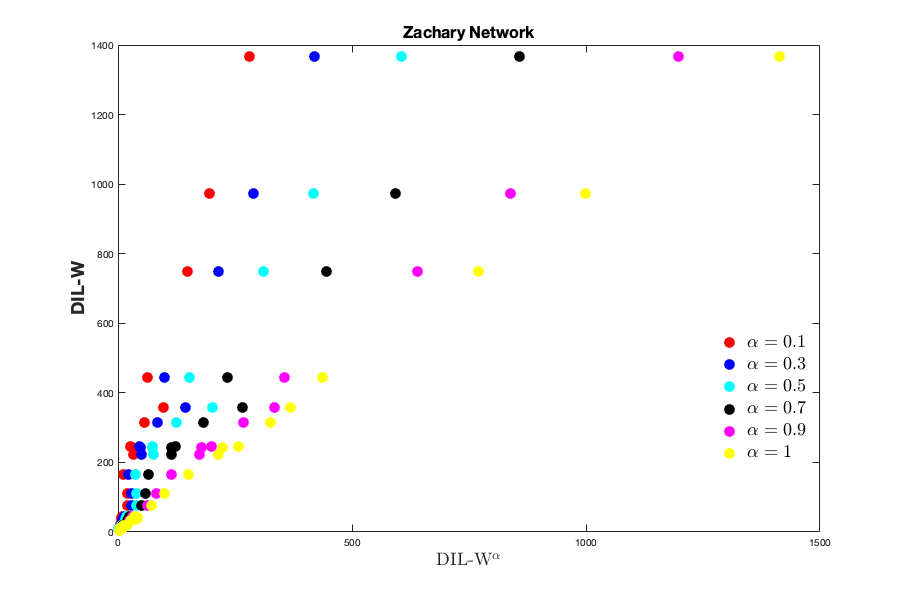}}
   \subfigure
    {\includegraphics[width=0.45\textwidth]{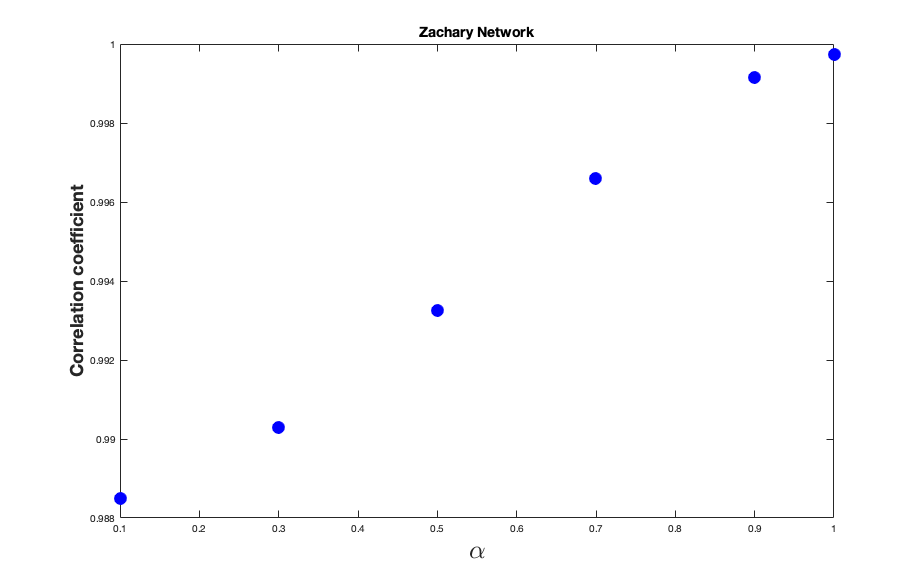}}
  \subfigure
    {\includegraphics[width=0.45\textwidth]{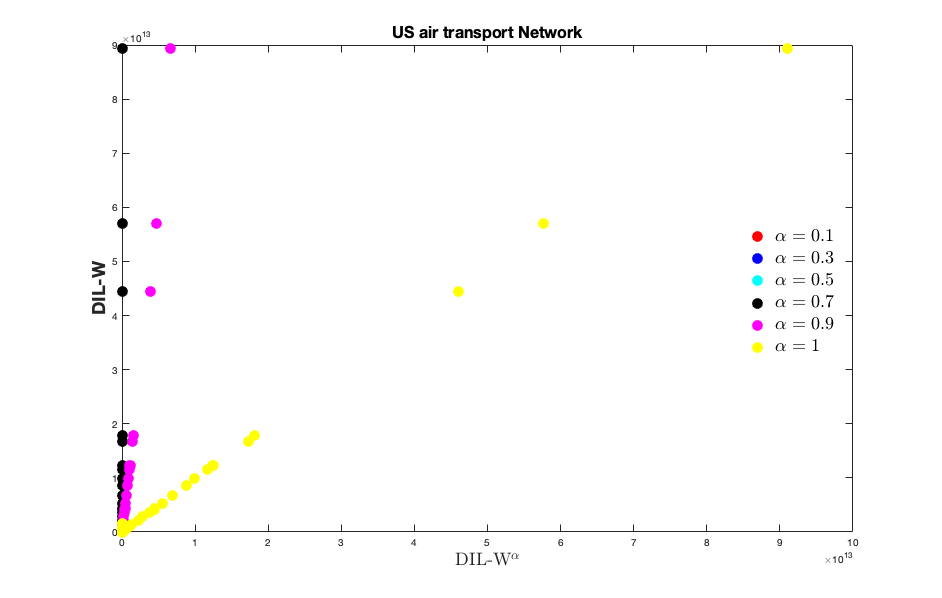}}
    \subfigure
    {\includegraphics[width=0.45\textwidth]{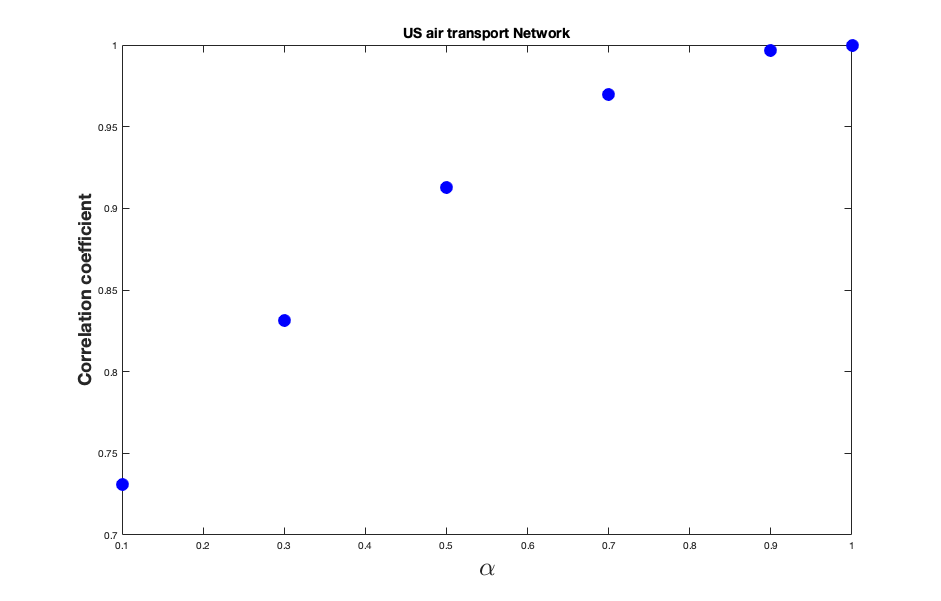}}
    \subfigure
    {\includegraphics[width=0.45\textwidth]{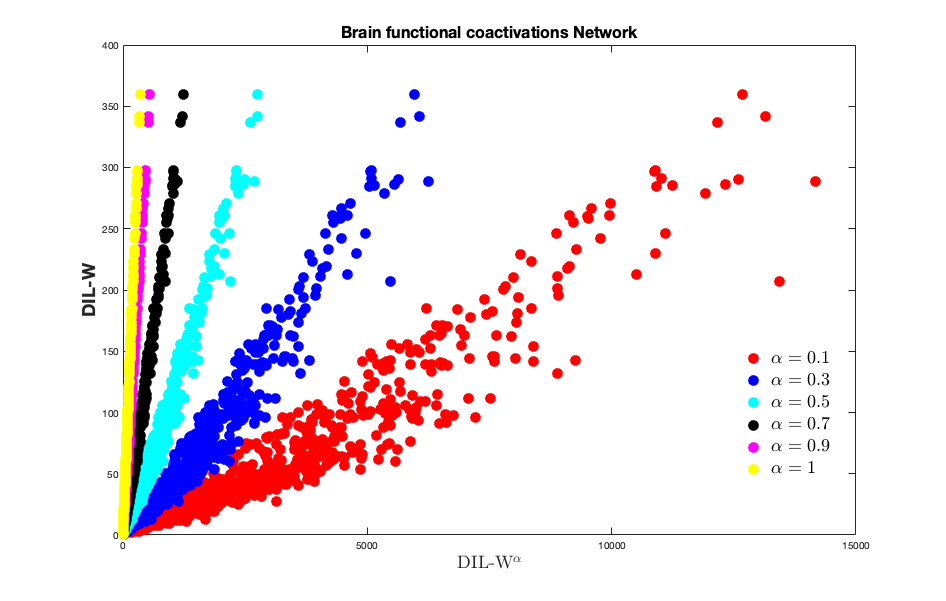}}
      \subfigure
    {\includegraphics[width=0.45\textwidth]{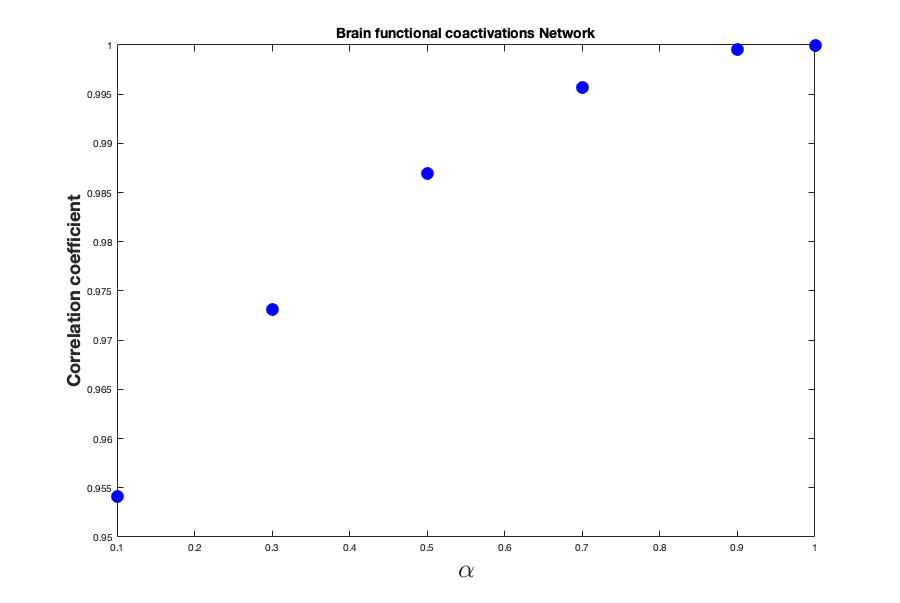}}
   \subfigure
    {\includegraphics[width=0.45\textwidth]{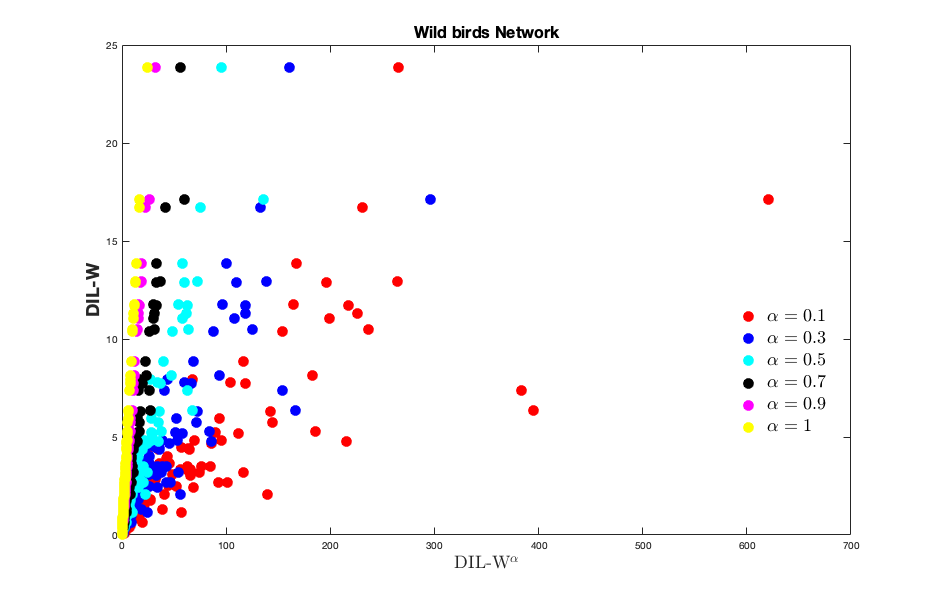}}
  \subfigure
    {\includegraphics[width=0.45\textwidth]{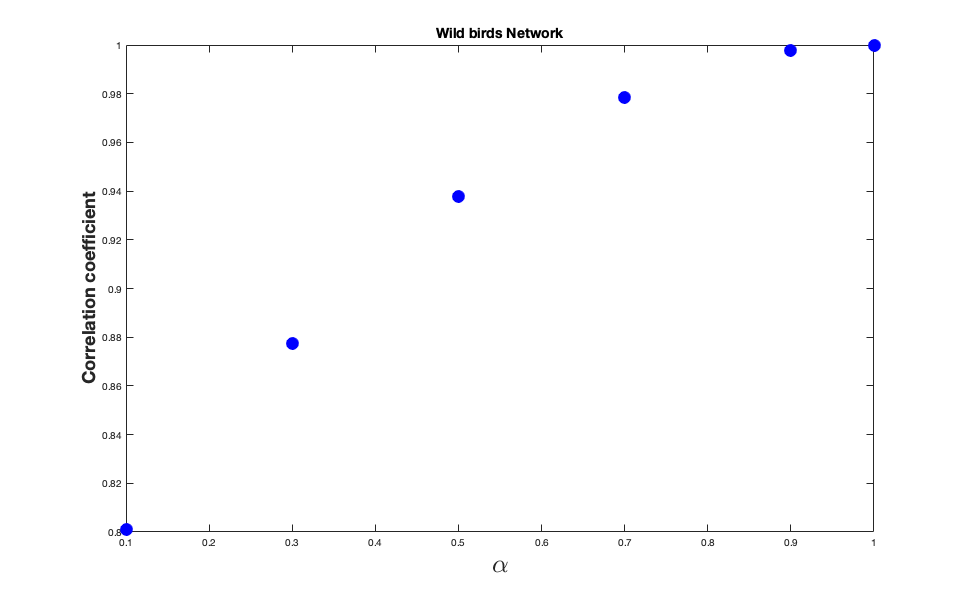}}
    \subfigure
    {\includegraphics[width=0.45\textwidth]{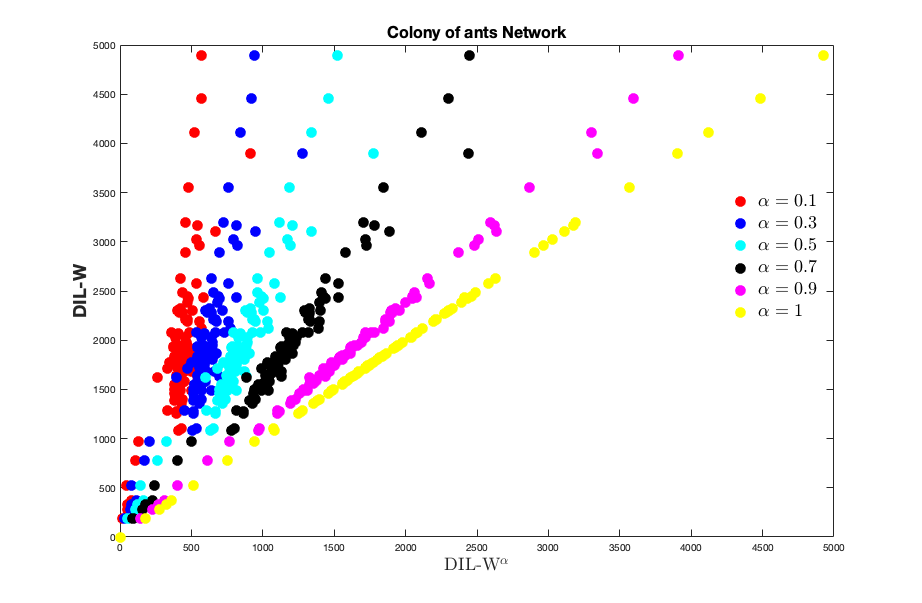}}
    \subfigure
    {\includegraphics[width=0.45\textwidth]{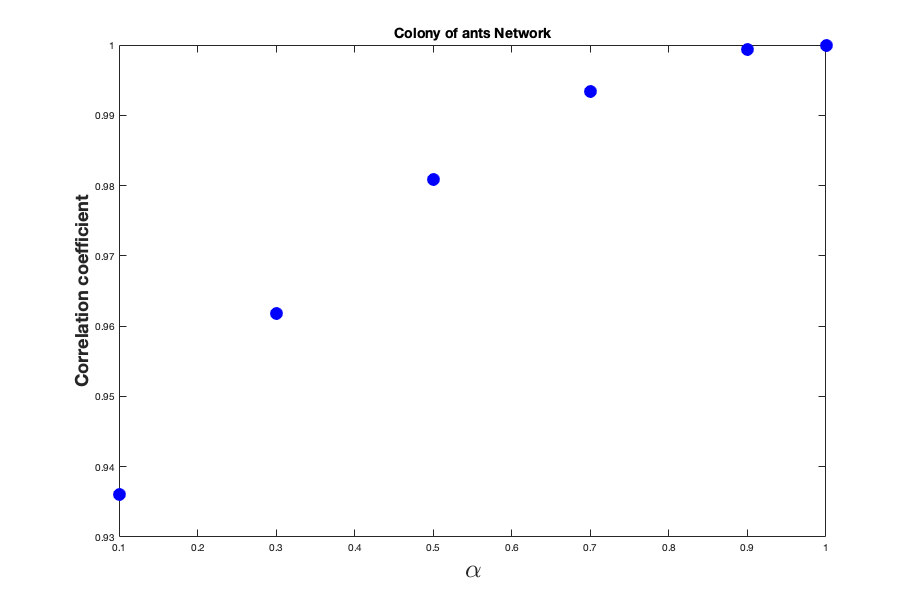}}
   \caption{On the left the relationship between the DIL-W and DIL-W $^{\alpha}$ rankings for different values of $\alpha $. On the right the correlation coefficient for the different values of $\alpha $.}
    \label{coef}
\end{center}
\end{figure}


\subsection{Computational complexity}
In general the real complex networks have many vertices, then the computational complexity plays an important role. DIL-W$^{\alpha}$ has a computational complexity of $O(n\cdot <k>)$, where $n$ is the total numbers of vertices of the graph and $<k>$ is the average degree of vertices in a network. Moreover, DIL-W$^{\alpha}$ uses local information of a vertex for its importance evaluation. By the above, we believe that DIL-W$^{\alpha}$ can be efficient to compute in large networks.

\section{Conclusion}\label{conclusionchap2}
In this chapter, we propose a ranking of importance of nodes which is applicable to undirected and edge-weighted networks. We can see that DIL-W $^{\alpha}$ is equal to or better in efficiency than the DIL-W proposal given in \cite{almasi2019measuring} in the five networks that we have considered.

We assess effectiveness by comparing the rate of decrease of network efficiency caused by node removal. The higher the rate of decrease of network efficiency caused by node removal, the more important will be the node. In Figure \ref{figura4}, the results show that DIL-W$^{\alpha}$ can evaluate the importance of nodes better in the US air transport, Zachary karate club and wild birds networks, when $\alpha $ is less than one. Conversely, in colony of ants and brain coactivations networks the best performance for efficiency is with $\alpha=1 $ or $\alpha $ close to $1$. In consequence, when modifying the parameter $\alpha$ it is possible to obtain better rankings for the network in the sense of effectiveness. 

From the above, we are presented with an important task or challenge, which is to establish what value of $\alpha $ makes the ranking perform more optimally. It is difficult to determine the exact value of $\alpha$ to use. We think that there may be a dependency on the weights at the edges. For instance, is it better to have many edges with weights between 0 and 1 or few edges with weights greater than 1?. Such answers would allow a better understanding of the optimal $\alpha $ in certain networks.

From table \ref {table1} we can see that only in the US air transport network the rankings made by DIL-W and DIL-W$^1$ do not match completely, since there is a small difference in positions 6 and 7 of the rankings. 

With respect to the computational complexity, we believe that DIL-W$^{\alpha}$ can be efficient to compute in large networks because it uses local information of a vertex for its importance evaluation.

In summary, our research provides a way to generalize the classification method for edge- weighted and undirected graphs. This turns out to be an element for the discussion around the usefulness of the concept of degree of centrality and its standard use has limitations in complex networks, as pointed out by Sciarra~et~al.~\cite{sciarra}.

\chapter[Protection Strategy for Edge-Weighted Graphs]{Protection Strategy for Edge-Weighted Graphs in Disease Spread}\label{pro}
\minitoc

\section{Introduction}
Fake news, viruses on computer systems or infectious diseases on communities are some of the problems that are addressed by researchers dedicated to study complex networks. The immunization process is the solution to these challenges and hence the importance of obtaining immunization strategies that control these spreads. 

For more than 2 decades, the~complex networks have been the focus of many researchers due to their multiple applications in economy, traffic problems, spread of diseases, electrical networks, biological systems and the famous social networks (see for \mbox{instance~\cite{almasi2019measuring,an2014synchronization,Crossley11583, guangzeng2009novel,demongeot2013archimedean,montenegro2019linear,pastor-satorras_epidemic_2001}}). Furthermore, in~discrete mathematics the most important object is the~graph.

Modeling the spread of disease over networks has many applications in real life. For~example, rumors or spam that spread on large social networks such as Facebook or Twitter~\cite{guess2019less}, viruses that spread through computer systems~\cite{yang2012towards}, or~diseases that spread over a population such as it does the actual SARS-CoV-2 (see for instance~\cite{ronald1}). This leads to solve the problem of totally or partially controlling these spreads.

It is well known that a set of nodes can spread disease over the entire network. The~nodes that have this characteristic are called influencing nodes or spreader nodes. Consequently, determining this type of nodes is crucial for the objective of controlling spreads over networks. That is, we can prevent the spread of contagious diseases over a network by immunizing the influencing nodes. In~that direction, efforts have been directed in many investigations in the last time. Among~them, Pastor-Satorras and Vespignani in~\cite{pastor2002immunization} concluded that uniform random immunization of nodes does not lead to the eradication of infections in all networks, while targeted immunization drastically reduces the vulnerability of the network to epidemic attacks.
In~\cite{cohen2003efficient} the authors provided a purely local strategy, which requires minimal information about the randomly selected nodes.
Tong~et~al.~\cite{tong2010} proposed NetShield: an effective immunization strategy which use the properties of matrix perturbation. A~few years later (2016) its NetShield$+$ variant appears to balance the optimization quality and speed~\cite{chen2015node}.
The authors in~\cite{hebert2013global} use different features of the organization of the network to identify influential diffusers.
In 2014 and 2015 Zhang and Prakash developed DAVA and DAVA-fast, two methods in which all infected nodes merge into a supernode by constructing a weighted dominator tree of the input network~\cite{zhang2014dava,zhang2015data}. 
Song~et~al. in~\cite{song2015node} provided NIIP, which selects $ k $ nodes to immunize over a period of time. For~each time point, the~NIIP algorithm will take decisions about which nodes to immunize given the estimated value of $ k $ for that time point. 
 In~\cite{gupta}, Gupta~et~al. also agree that community structure is important for understanding the spread of an epidemic. Wang~et~al. propose a dynamic minimization model of the influence of a rumor with the user experience (DRIMUX), considering both the global popularity and the individual attraction of the rumor~\cite{wang2017drimux}. 
GraphShield method is developed by Wijayanto and Murata in~\cite{wijayanto2017flow} taking into account the function of infection flow, graphics connectivity and top-grade centrality.
Saxena~et~al. designed a method to identify the best ranked nodes in order to control an epidemic through immunizations~\cite{saxena}. Ghalmane~et~al. analyzed the nonoverlapping community structure network based on an immunization strategy, recognizing the relevance of the node given the characteristics of the communities~\cite{ghalmane2019immunization}, while the same author asserts that modular centrality must include the influences of the nodes both in their own communities and in others~\cite{ghalmane2019centrality}.
In~\cite{wijayanto2019effective} the authors introduced ReProtect and ReProtect-p methods, which divide the size of protection budget into several turns and protect nodes according to the currently observed temporal snapshot of dynamic networks. 
Tang~et~al.~\cite{tang2020research} proposed the weighted $K$-order propagation number algorithm to extract the disease propagation based on the network topology to evaluate the node~importance.

In~\cite{ronaldranking}, the~authors show that the DIL-W $^{\alpha}$ ranking provides good results, regarding the rate of decline in network efficiency (more details in~\cite{ren2013node}). Furthermore, one of the good qualities of the DIL-W$^\alpha$ ranking is that it recognizes the importance of bridge nodes, these nodes are those that connect the peripheral nodes and the peripheral groups with the rest of the network (see more details in~\cite{musial}). This attribute is inherited from the version of the DIL ranking for graphs not weighted at the edges (see~\cite{liu2016evaluating}). For~this reason, we have chosen this ranking to evaluate its effectiveness in the immunization of nodes that are attacked by an infectious disease that spreads on an edge-weighted graph. The~immunization is done according to the importance ranking list produced by DIL-W$^{\alpha}$ ranking. The~effectiveness is measured with the ratio of vertices that remain uninfected at the end of the disease over the total numbers of vertices subject to budget protection. The~experimentation was done on real {and scale-free} networks.
In this chapter, we evaluate the effectiveness of the DIL-W$^{\alpha}$ ranking in the immunization of nodes that are attacked by an infectious disease that spreads on an edge-weighted graph using a graph-based SIR model. The experimentation was done on real {and scale-free} networks and the results illustrate the benefits of this ranking.\\

\section{Strategy~Protection}\label{strategy}
In this section, we provide definitions of the protection of a graph when a disease spreads on it. Moreover, we state the protection strategy.
\begin{Definition}\label{def_proc}
Protecting a vertex means removing all of its corresponding edge (See Figure~\ref{protectionexample1}).
\end{Definition}
It is also possible to find in the literature that protecting a vertex means removing the vertex from the graph. See for instance~\cite{tong2010}. 
\begin{figure}[h]
\centering
  \subfigure
    {\includegraphics[width=0.3\textwidth]{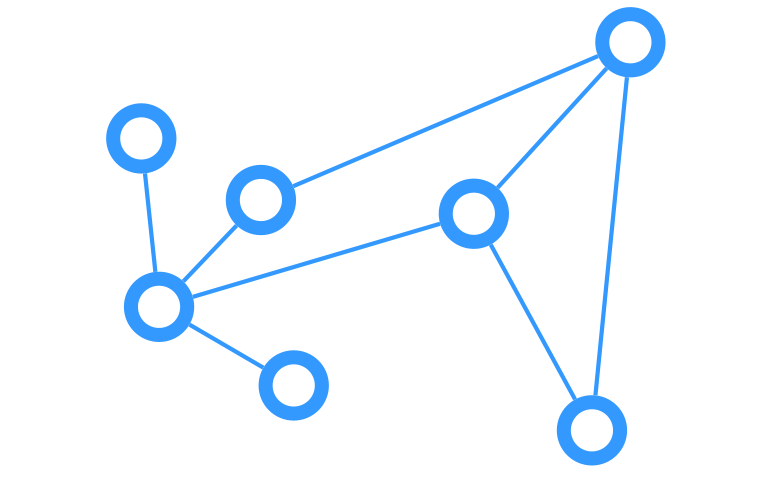}}
  \subfigure
    {\includegraphics[width=0.34\textwidth]{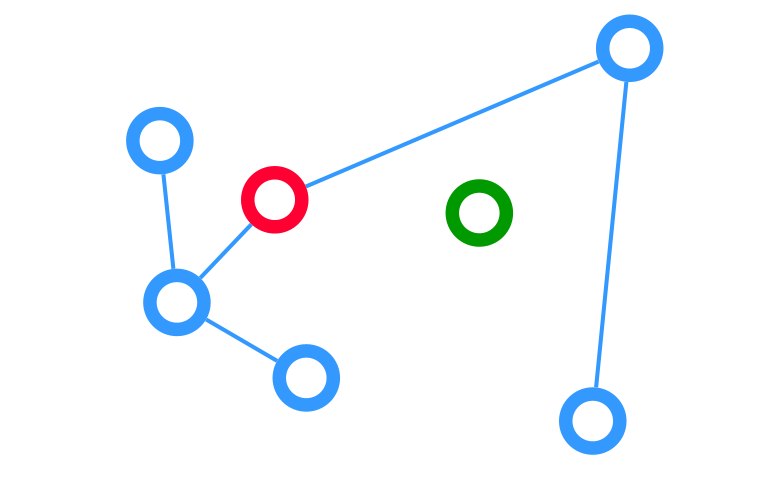}}
  \caption{The left side shows a graph without protection or infection. On~the right side, we see an infected node (red) and a protected node (green) in the same~graph.}
  \label{protectionexample1}
\end{figure}

\begin{Definition}
The numbers of vertices allowed to protect is called protection budget, denoted by $k$.
\end{Definition}

It is clear that $k\in\mathbb{Z}^+$.
\begin{Definition}
We will say that the survival rate, denoted by $\sigma$, is the ratio of vertices that remain uninfected at the end of disease over the total numbers of vertices.
\end{Definition}

Therefore, our problem is: given a graph $G=(E,V)$, SIR model, and~a protection budget $k$, the~goal is to find a set of vertices $\mathcal{S}\subseteq V$ such that 
\begin{align}\label{problem}
\displaystyle\theta^{*}&=\arg\max_{\mathcal{S}\subseteq V}\sigma,
\end{align}
with $|\mathcal{S}|=k$.
However, the~problem (\ref{problem}) is NP-Hard (see~\cite{1318579}).

It is well known that an index to measure the connection of a graph is the efficiency of the networks (see~\cite{PhysRevLett}). High connectivity of the graph indicates high efficiency. In~\cite{ronaldranking} the authors show that the DIL-W $^{\alpha}$ ranking provides good results, regarding the rate of decline in network efficiency (more details in~\cite{ren2013node}), when it comes to eliminating the nodes best positioned by this ranking. One of the good qualities of the DIL-W$^\alpha$ ranking is that it recognizes the importance of bridge nodes (see more in~\cite{musial}). This quality is inherited from the version of the DIL ranking for graphs not weighted at the edges (see~\cite{liu2016evaluating}).  For~this reason, we have chosen this ranking to evaluate its effectiveness in the immunization of nodes that are attacked by an infectious disease that spreads on an edge-weighted~graph.

From the results in~\cite{lai2004}, we have that 5\% to 10\% of important nodes can cause the entire network to fail. According to the latter, our protection budget $k$ will be 10\% of the network~nodes. 

Finally, in~order to illustrate in a simple way why the DIL-W$^{\alpha}$ ranking has been chosen, let us consider the graph of Figure \ref{free_toydil} with 16 edges, 15 nodes, and the respective weights on the edges. When we apply the DIL-W$^1$ ranking, the~first $3$ places are occupied by nodes 3, 5 and 1, respectively. These nodes are precise bridge nodes and, when protecting them, according to Definition \ref{def_proc}, the~graph loses connectivity (see Figure \ref{free_toy2}). If~we apply the Strength ranking, the~first three places are occupied by nodes $3$, $1$ and $4$, respectively. Note that the order in which it positions the nodes and the importance it gives to node $4$ makes the loss of network connectivity lower than the loss when applying DIL-W$^1$ (see Figure \ref {free_toy2}).

\begin{figure}[h]
{\includegraphics[width=0.55\textwidth]{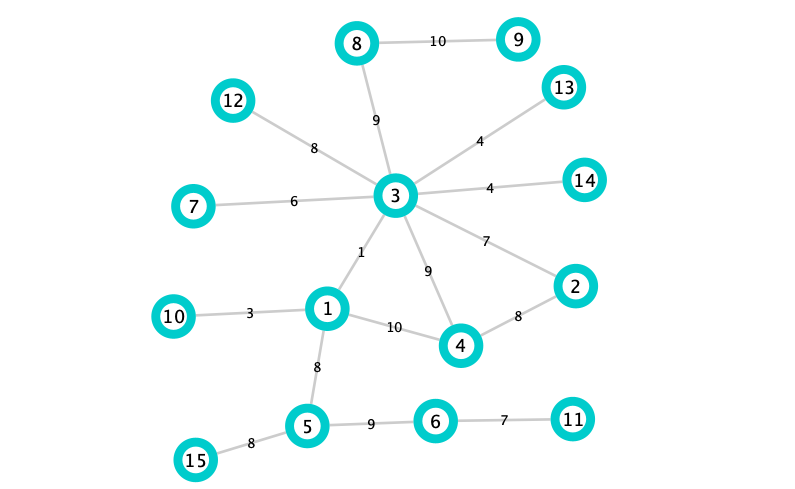}}
\caption{Graph with 15 nodes, 16 edges, and the respective weights.} \label{free_toydil}
\end{figure}
\begin{figure}[h]
\subfigure
{\includegraphics[width=0.49\textwidth]{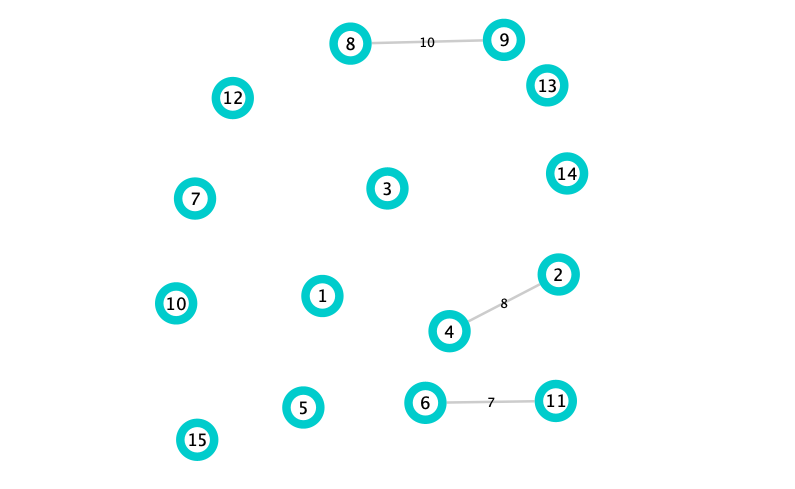}}
\subfigure
{\includegraphics[width=0.49\textwidth]{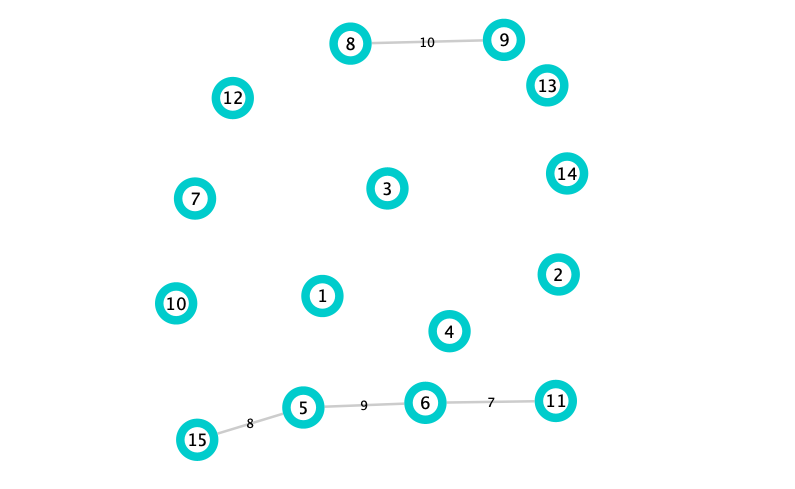}}
\caption{In the left, the first 3 protected places generated by the DIL-W$^1$ ranking. On the right, the first 3 protected places generated by Strength ranking.} \label{free_toy2}
\end{figure}
\clearpage

In summary, we protect the ten percent of the network nodes according to the importance ranking list produced by DIL-W$^{\alpha}$.


\section{Simulation of Disease Spread with Protection and Analysis~Results}\label{simulation}
\subsection{Data and~Methodology}
The protection tests are carried out on 4 networks, these are: Zacharys karate club network~\cite{zachary} (with 34 nodes and 78 edges), wild bird network~\cite{firth2015experimental} (with 131 nodes and 1444~edges), Sandy authors network~\cite{imrich2000product} (with 86 nodes and 124 edges), and~CAG-mat72 (with 72 nodes and 678 edges)
where each network will have the immunization corresponding to the 10\% of its highest ranked nodes according to the DIL-W$^{\alpha}$ ranking~\cite{ronaldranking}, strength of the node~\cite{barrat2004architecture}, weighted betweenness centrality~\cite{brandes2001faster}, weighted closeness centrality~\cite{newman2001scientific}, and~Laplacian centrality for undirected and edge-weighted graph~\cite{qi2012laplacian}. For~the ranking DIL-W$^{\alpha}$, we consider $3$ different values for $\alpha$, they are $0$, $0.5$ and $1$. In~this way, we compare how the protection performs according to the value of $\alpha $. Finally, the~databases of our test networks can be found in~\cite{nr-aaai15}.
\subsection{Simulation of~Disease}
In this section, we use a graph-based SIR model in the same way as in~\cite{ronald1}, that is, each individual is represented by a vertex in an edge-weighted graph. At~time $t$, each vertex $v_i$ is in a state $v_i^t$ belonging to $\mathcal{S}=\{0,1,-1\}$, where $0,1$ and $-1$  represent the three discrete states:  susceptible (S), infected (I) and recovered (R). 

The initial population contains one infected node and all the simulations considering $\displaystyle\delta=\frac{1}{15}$. For~each network, it was consider a different $\rho$, so in each case the percentage of population affected is similar. Table~\ref{tabla0} shows the different $\rho$ for each test~network. 

\begin{table}[h]
\caption{Different $\rho$ for each test~network.} 
\begin{tabular}{lcccc}
\toprule
                & \textbf{Zachary Karate Club} & \textbf{Wild Bird} & \textbf{Sandy Authors} & \textbf{CAG-mat72} \\ \midrule
\textbf{$\rho$} & $0.15$                       & $0.23$             & $0.12$                 & $0.0052$           \\ \bottomrule
\end{tabular}
\label{tabla0}
\end{table}

Finally,  2000 simulations on each network of disease were performed. 
Figure~\ref{simulacion1} shows the results.

In Sandy authors and CAG-mat72 networks the DIL-W$^{0}$ ranking performs the best with respect to the minimum peak of the infected curve, respectively. Furthermore, DIL-W$^{0.5}$ ranking comes second in both networks. The~Strength and Laplacian rankings perform the best with respect to the minimum peak of the infected curve in Zachary karate club network. The~DIL-W$^{0.5}$ and DIL-W$^{1}$ rankings are in fourth and fifth place respectively. However, the~computational complexity of Laplacian and Strength is $O(n\cdot \max_{v\in V}(deg(v))$ and $O(m)$ respectively, while DIL-W$^{\alpha}$ has a computational complexity of $O(n\cdot< k>)$ (see~\cite{liu2016evaluating} or~\cite{ronaldranking}), where $n$ is the total number of vertices of the graph, $m$ is the total number of edges of the graph, and~$<k>$ is the average degree of vertices in the graph. In~Wild bird network,  DIL-W$^{1}$ performs better in the same~criterion. \\
On the other hand, regarding the decrease of the infected curve, DIL-W$^{\alpha}$ performs the best in Wild bird, Sandy authors, and~CAG-mat72 networks. The~Laplacian is better on the Zachary network. 
The above is summarized in the survival rate. In~Table~\ref{tabla1}, we can see the survival rates obtained according to the protections applied to each~network.

\begin{table}[h]
\caption{Survival rate in each network according to the protection ranking with a protection budget of 10\% of the network~nodes.}
\label{tabla1}
\begin{tabular}{lcccc}
\toprule
\multicolumn{1}{c}{} & \textbf{Zachary} & \textbf{Wild Birds} & \textbf{Sandy Authors} & \textbf{CAG-mat72} \\ \midrule
\textbf{DIL-W
$^0$}     & 0.8946 & 0.7871 & 0.9671 & 0.8947 \\
\textbf{DIL-W$^{0.5}$} & 0.8912 & 0.7943 & 0.9637 & 0.8887 \\
\textbf{DIL-W$^1$}     & 0.8941 & 0.7375 & 0.9613 & 0.8792 \\
\textbf{Strength}      & 0.9051 & 0.7041 & 0.9536 & 0.8685 \\
\textbf{Closeness}     & 0.8211 & 0.7327 & 0.9435 & 0.8451 \\
\textbf{Betweenness}   & 0.8890 & 0.6952 & 0.9616 & 0.8078 \\
\textbf{Laplacian}     & 0.9016 & 0.7023 & 0.9382 & 0.8707 \\ \bottomrule
\end{tabular}
\end{table}
\unskip

\begin{figure}[h]
\centering
\subfigure
{\includegraphics[width=0.49\textwidth]{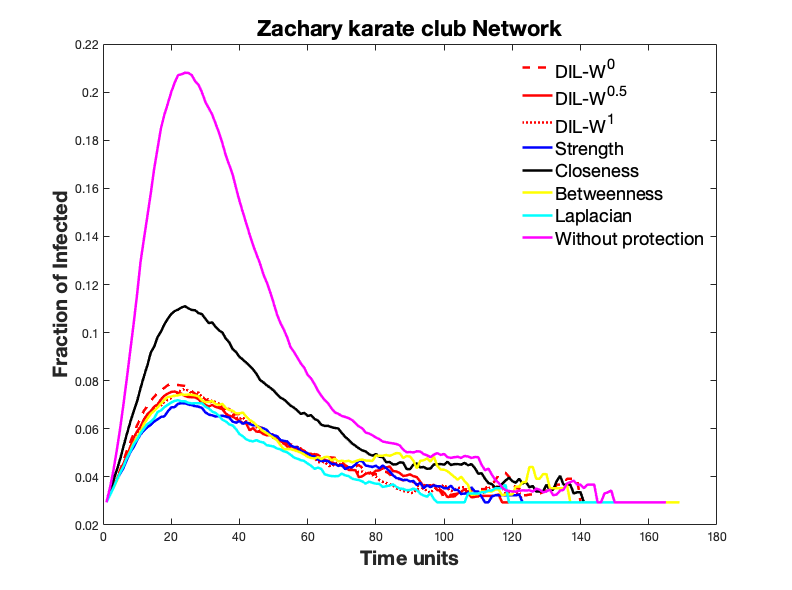}}
\subfigure
{\includegraphics[width=0.49\textwidth]{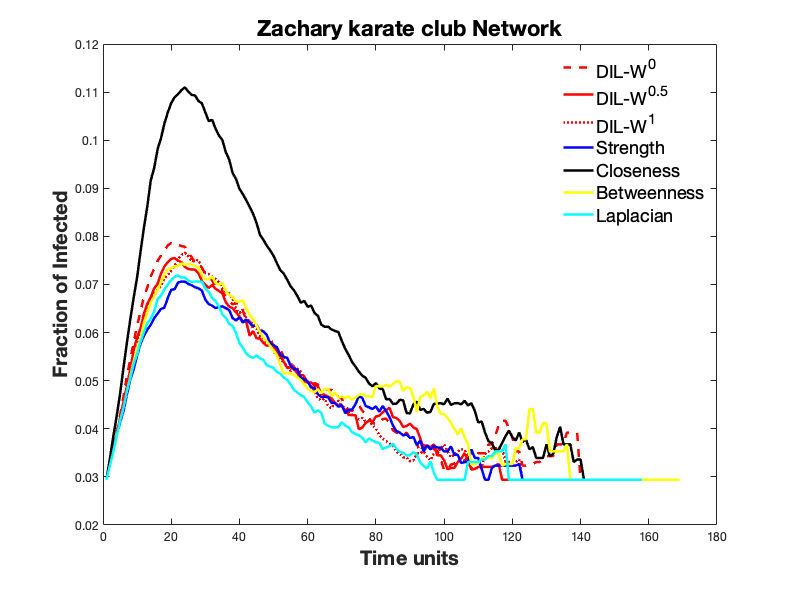}}
\subfigure
{\includegraphics[width=0.49\textwidth]{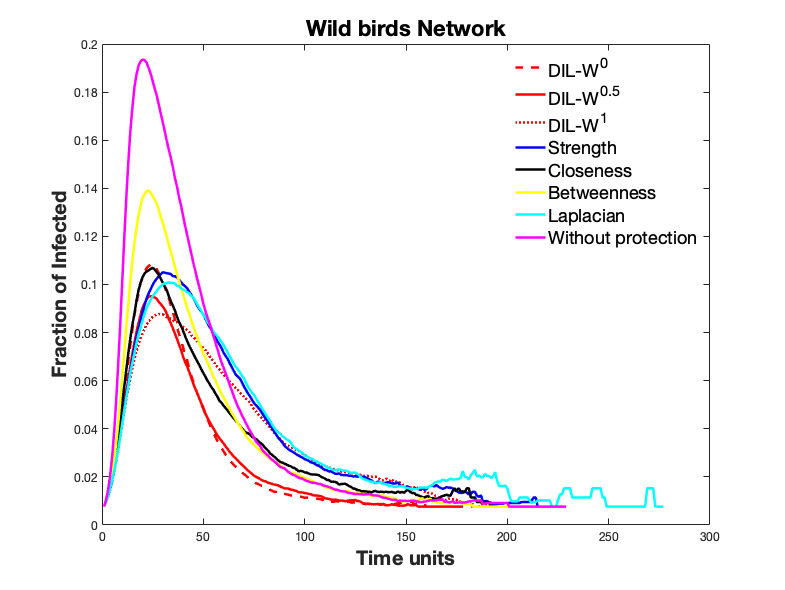}}
\subfigure
{\includegraphics[width=0.49\textwidth]{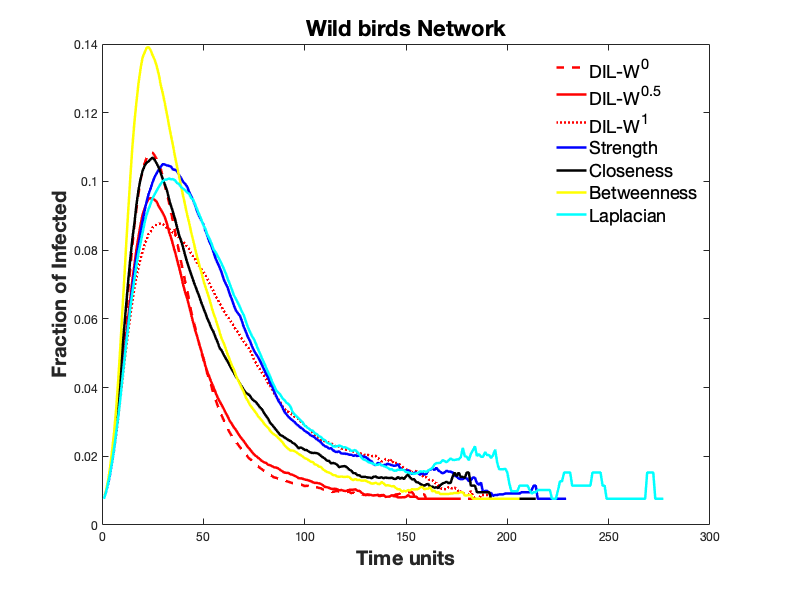}}
\subfigure
{\includegraphics[width=0.49\textwidth]{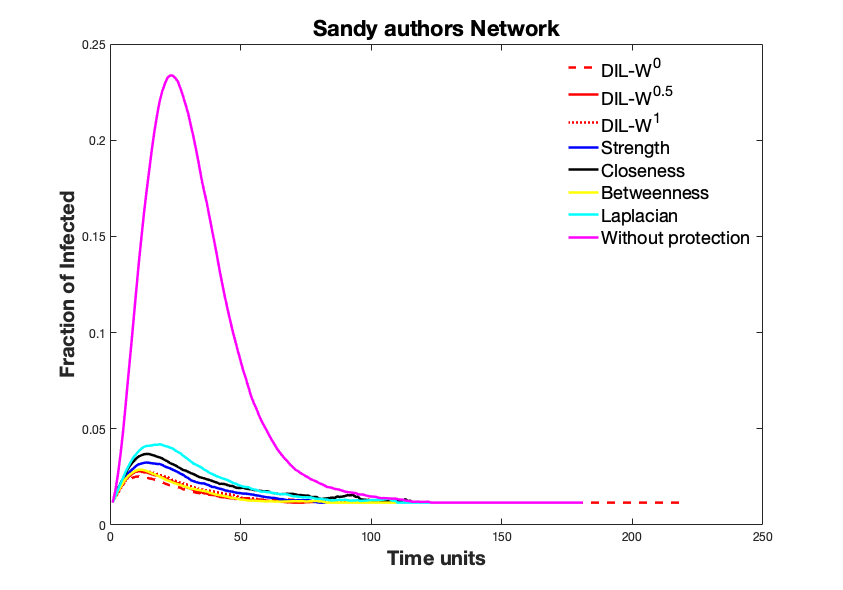}}
\subfigure
{\includegraphics[width=0.49\textwidth]{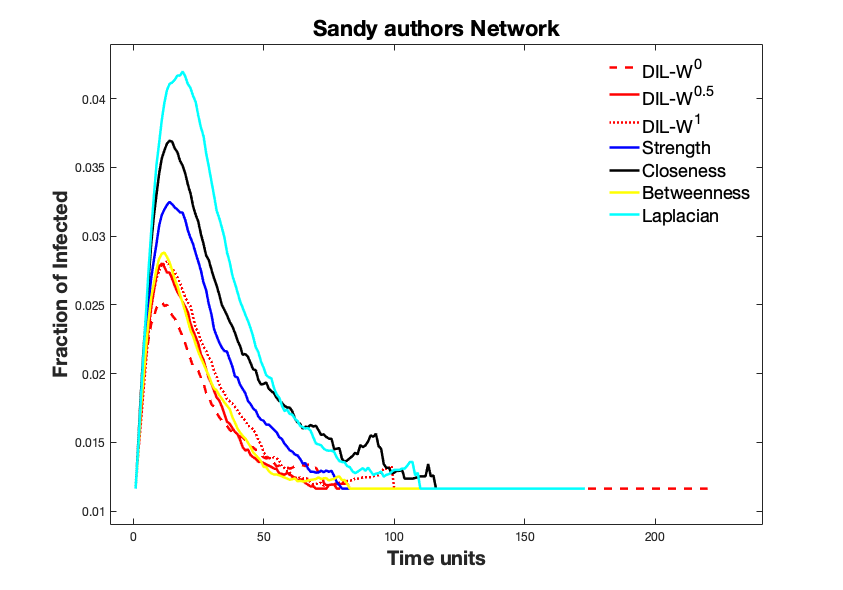}}
\subfigure
{\includegraphics[width=0.49\textwidth]{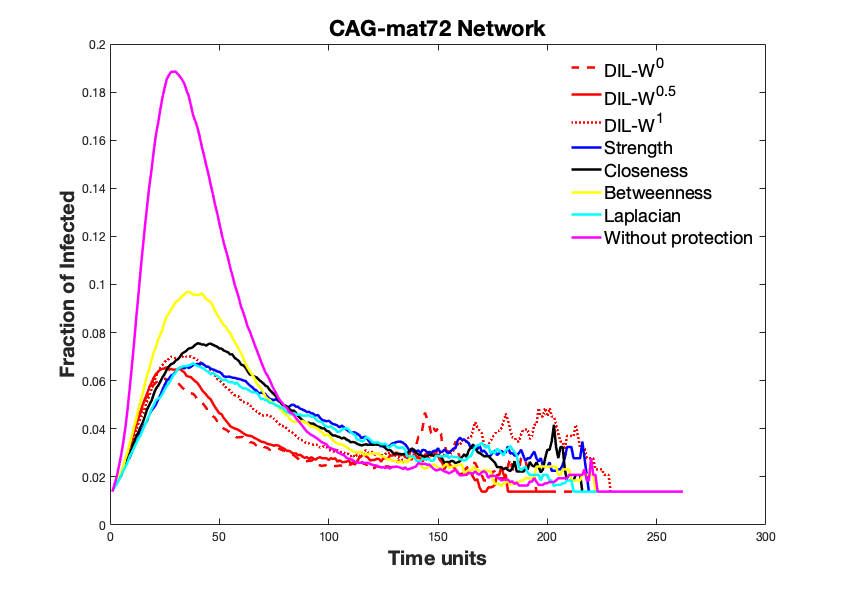}}
\subfigure
{\includegraphics[width=0.49\textwidth]{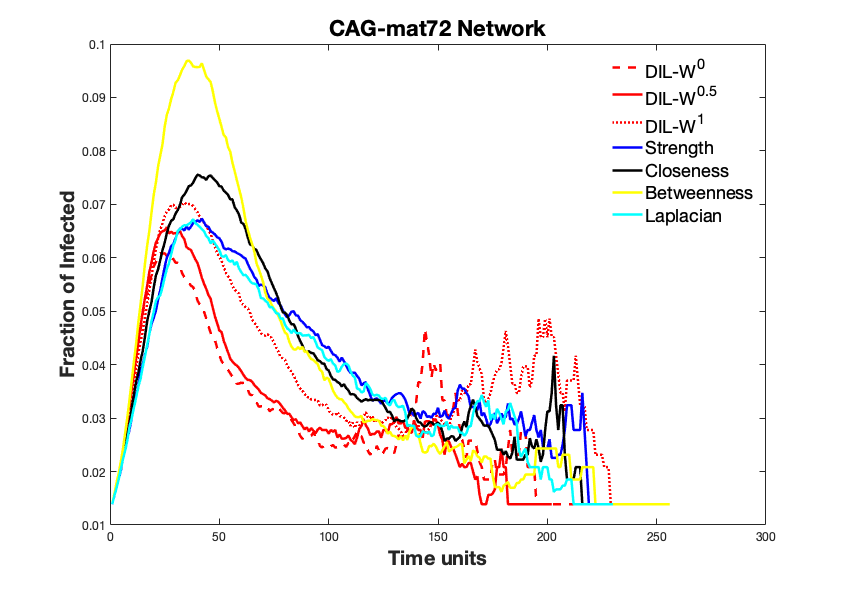}}
\caption{On the \textbf{right}, curve of infected with and without protection on each network. On~the \textbf{left}, only curves of infected with protection according to the different~rankings.} \label{simulacion1}
\end{figure}
\clearpage
\subsection{Survival~Rate}
In this Section, we address the variation in the survival rate ($\sigma$) when the protection budget ($k$) is modified. As~we have seen in the above Section, $k$ is the 10\% of the network nodes. The~variation of $k$ is between $5$ and $50$ percent of the network nodes. On~each test network and each protection budget 2000 simulations were done. Figure~\ref{survivalnetwork} shows the~results. 

\begin{figure}[h!]
\centering
\subfigure
{\includegraphics[width=0.49\textwidth]{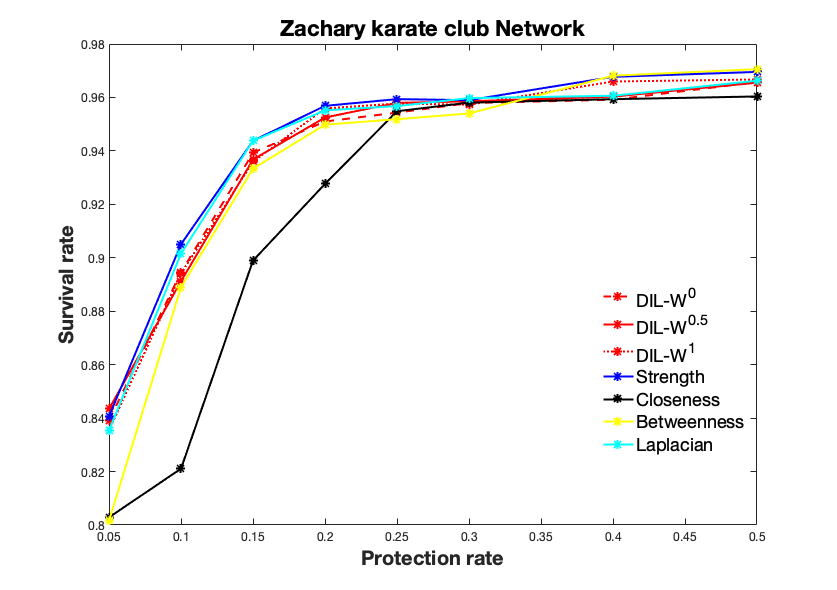}}
\subfigure
{\includegraphics[width=0.49\textwidth]{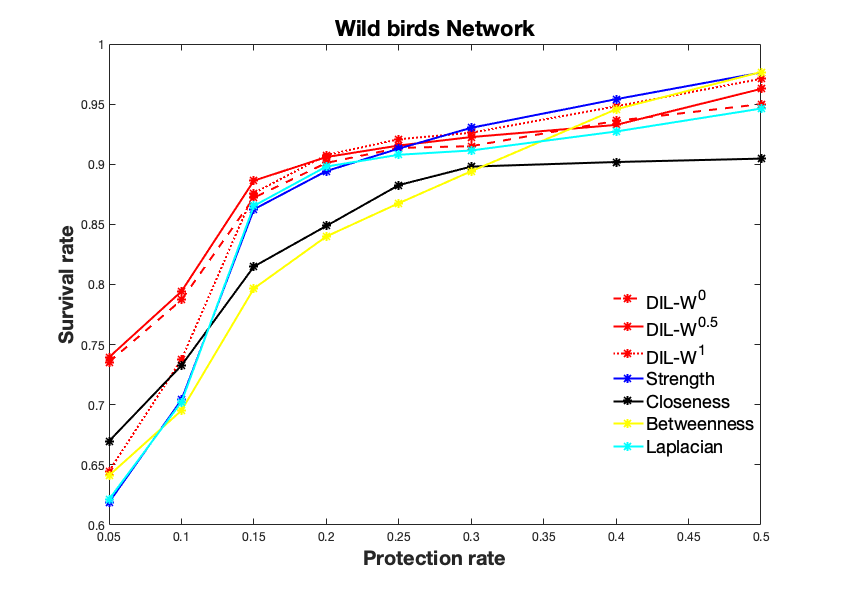}}
\subfigure
{\includegraphics[width=0.49\textwidth]{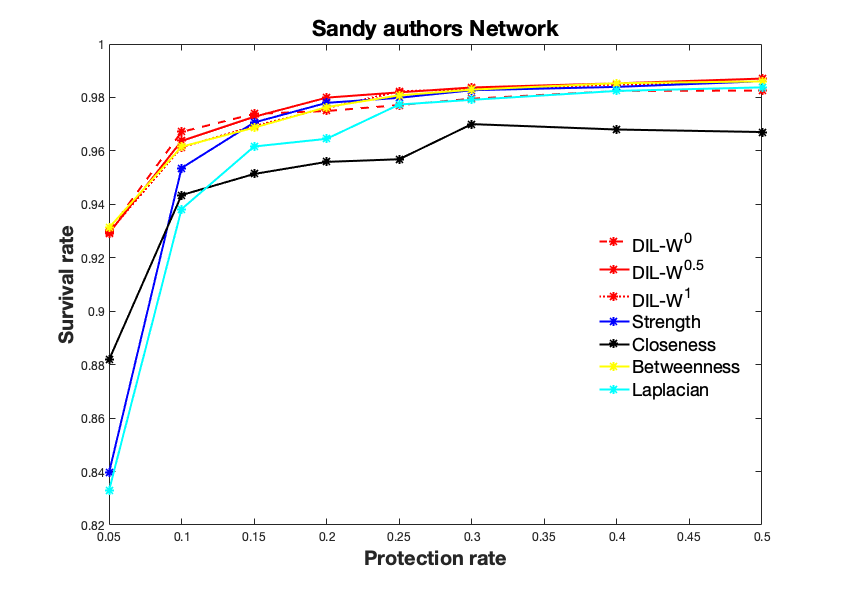}}
\subfigure
{\includegraphics[width=0.49\textwidth]{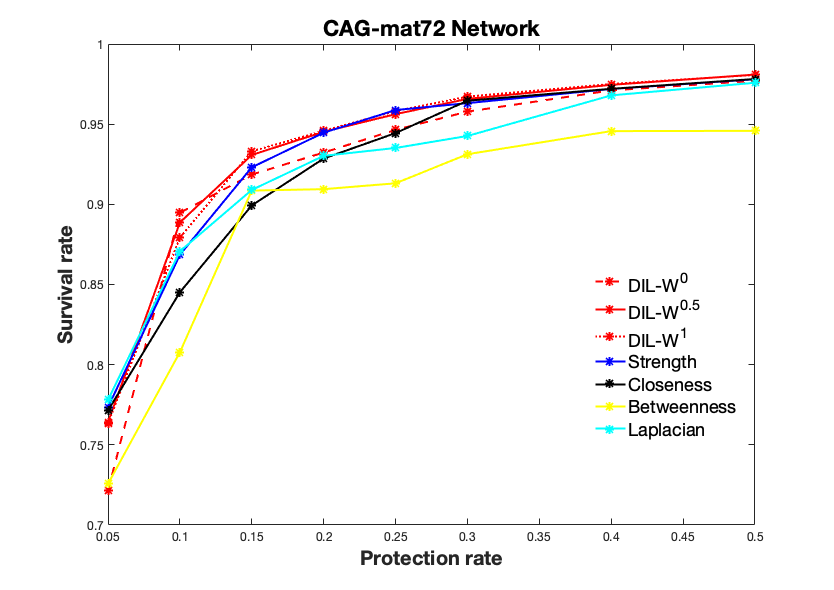}}
\caption{$\sigma$ as a function of $k$ on each~network.} \label{survivalnetwork}
\end{figure}



We can see that, at~the beginning (that is considering $k=5\%$), DIL-W$^{\alpha}$ is always among the first 3 places with the highest survival rate in Wild bird, Zachary karate club, and~Sandy authors networks. Even in the CAG-mat72 network it ranks fourth, very close to the~tops.

With the 10\% protection budget, DIL-W$^0$, DIL-W$^{0.5}$, and~DIL-W$^1$ perform best with the top three spots on Wild birds and  CAG-mat72 networks. It similarly happens in Sandy-authors network since DIL-W$^1$ has a very close rate to that generated by the betweenness ranking ($0.9613$ and $0.9616$ respectively). In~Zachary karate club network  the strength ranking performs better with this protection budget. The~DIL-W$^0$ ranking occupies the third position behind Laplacian ranking who performs second. However, DIL-W$^0$ has minor computational complexity than~Laplacian.   

Finally, with~50\% of the network nodes as the protection budget, in~the CAG-mat72 network DIL-W$^0$ and DIL-W$^1$ rankings are better performed. In~the Wild-birds network, the~betweenness and strength  rankings finish with the best performance respectively. The~DIL-W$^1$ ranking follows. In~the Sandy authors network, DIL-W$^{0.5}$ performs better while in the Zachary network, betweenness and strength displace DIL-W$^1$ to third place. Tables~\ref{table_zachary}--\ref{table_cag} show in detail the survival rates according to the protection budget in each test network (see Appendix \ref{apendi}).
{\subsection{Scale-Free~Network}
The scale-free networks are networks whose distribution degree follows a power law distribution with an exponent between 2 and 3. The~study of epidemics and disease dynamics on scale-free networks is a relevant theoretical issue~\cite{liu_2011}, because~this networks are a model for the spread of sexually transmitted diseases (see for instance~\cite{Mao_Xing_2009}) or a model to explain the early spread of COVID-19 in China (see~\cite{song2020massive}). Many works that address this type of network in spread of disease can be found in the literature. See for instance~\cite{schneeberger2004scale, moreno2003disease,yang2019dynamics,ke2006immunization}}.

Following the previous structure, we analyze the protection on a Scale-free network using the model proposed by Albert-L\'aszl\'o Barab\'asi and R\'eka Albert in~\cite{barabasi_emergence_1999}. The~variables size of the graph ($N$) and the average degree of the vertices ($d$) were studied with respect to the survival rate. The~weights are uniformly distributed between $1$ and $10$. Finally,  5000~simulations on each network were done, considering $\rho=0.011$, $\displaystyle\delta=\frac{1}{14}$, and~the initial population containing one infected vertex.

In the same direction as before, we analyze the variation of $\sigma$ according to the change of $k$. For~that, we fix $d=5$ and $N=100$.
The results obtained are the same as those obtained in the previous Section, this is to say, the~ranking DIL-W$^{\alpha}$ performs the best. Indeed, with~$k=5\%$ the best is DIL-W$^{1}$. For~the following protection values, DIL-W$^{0}$ performs the best very close to Strength ranking (see Figure~\ref{free0}). In~Appendix \ref{apendi}, Table~\ref{tab_free} shows in detail $\sigma$ according to $k$.

\begin{figure}[h!]
\centering
{\includegraphics[width=0.6\textwidth]{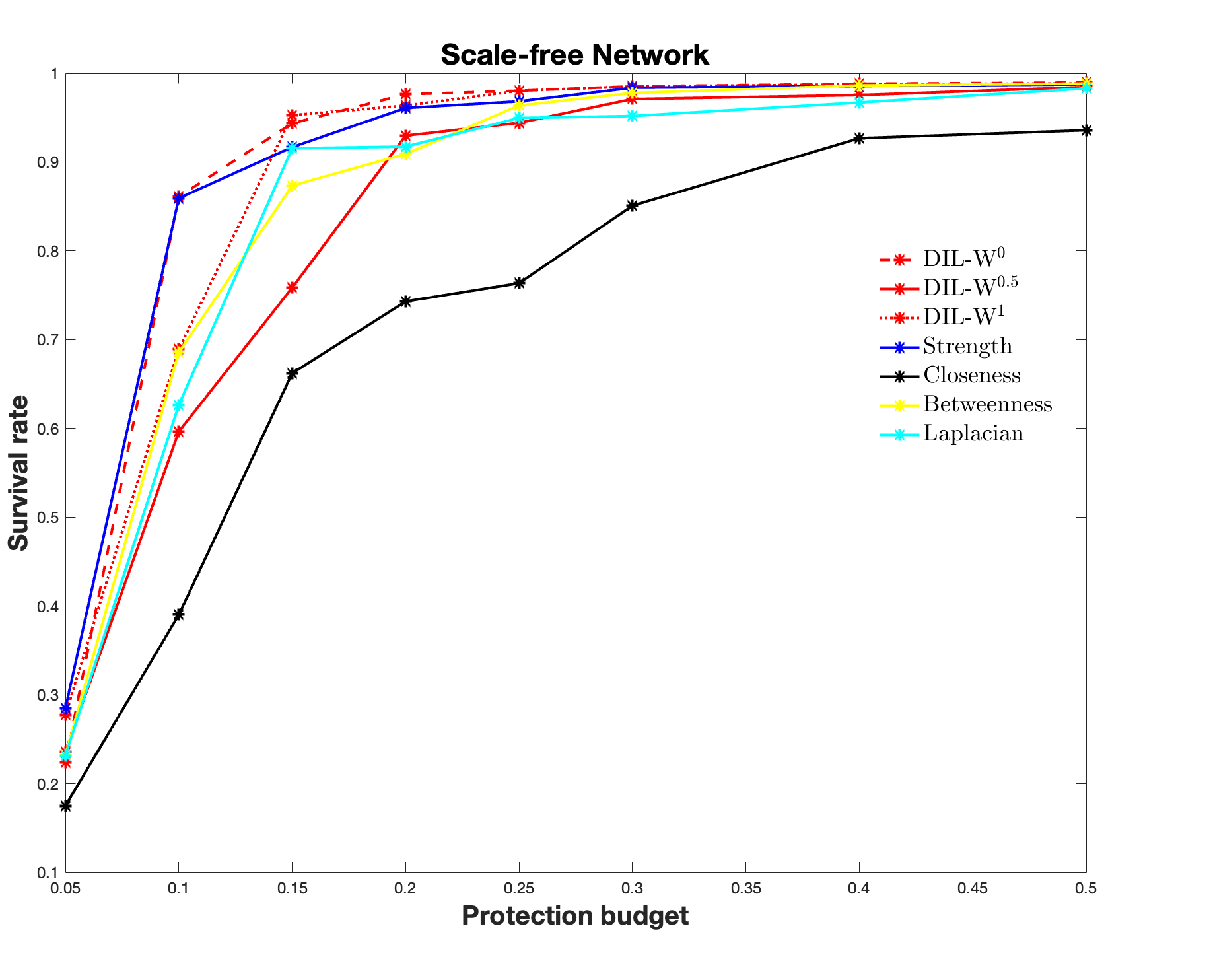}}
\caption{$\sigma$ as a function of $k$ on a Scale-free network with $N=100$ and $d=5$.} \label{free0}
\end{figure}

In the variation of the network size, we fix $d=5$ and $k=10\%$ of the network nodes. The~DIL-W$^{\alpha}$ ranking performs the best as well as the Strength ranking when $\alpha=1$ or $\alpha=0$. It is clear that the survival rate is increasing along with the size of the network, when we consider DIL-W$^{\alpha}$. Figure~\ref{free1} shows $\sigma$ as a function of the size of the network.

\begin{figure}[h!]
\centering
{\includegraphics[width=0.6\textwidth]{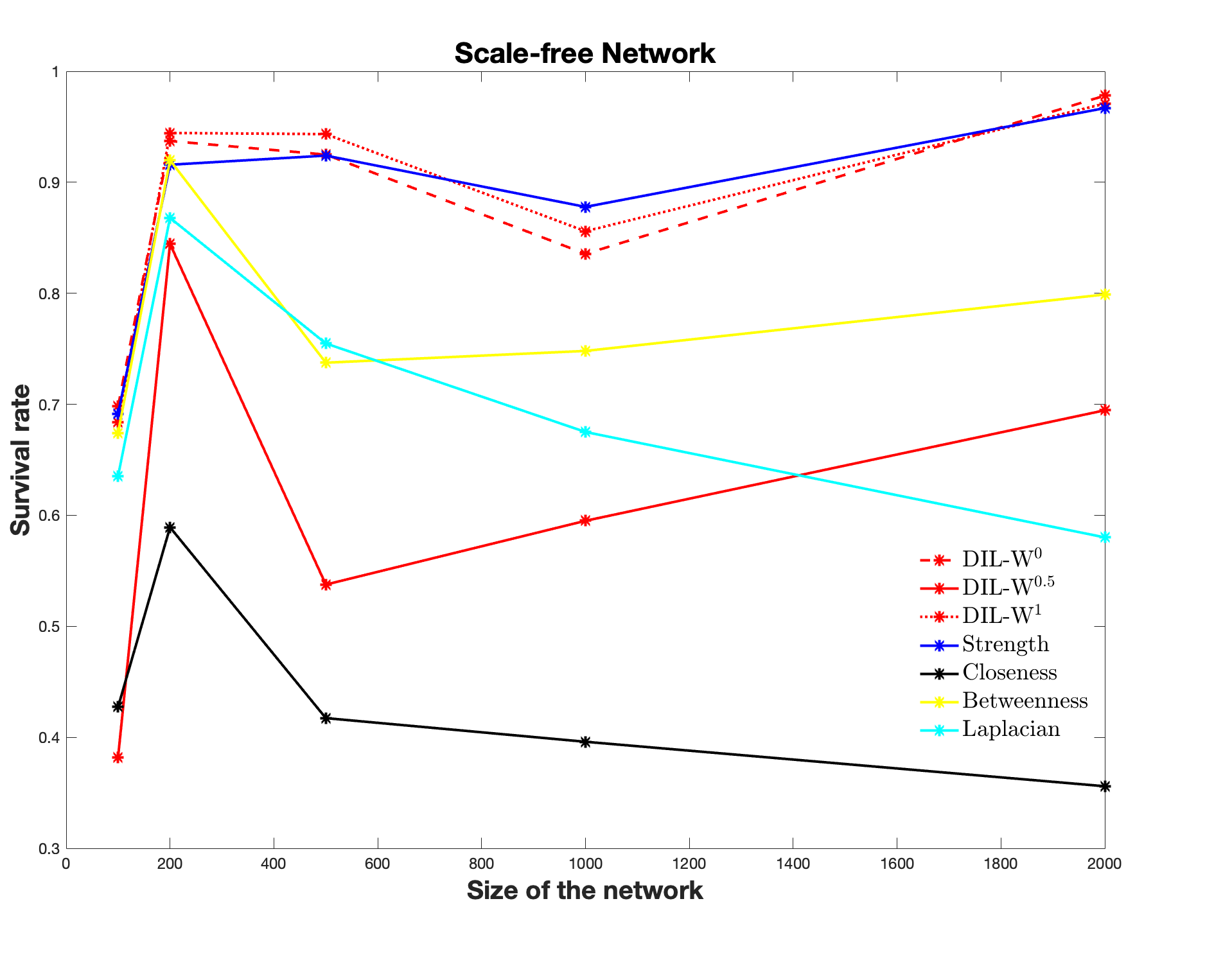}}
\caption{$\sigma$ as a function of $N$ on a Scale-free network with $d=5$.} \label{free1}
\end{figure}

{When $\sigma$ is a function of the average contact, we fix $N=100$ and $k=10\%$ of the network nodes. The~best performance is for the DIL-W$^{1}$ ranking up to the average of 30~contacts, that is, 30\% of the total network nodes. The~survival rate is decreasing because the higher the average number of contacts, the~greater the probability of contagion. \mbox{Figure~\ref{free2}} shows the results. (See Appendix \ref{apendi}, Table~\ref{tab_free2},  to~see in detail $\sigma$ according to $d$.)}

\begin{figure}[h!]
\centering
{\includegraphics[width=0.6\textwidth]{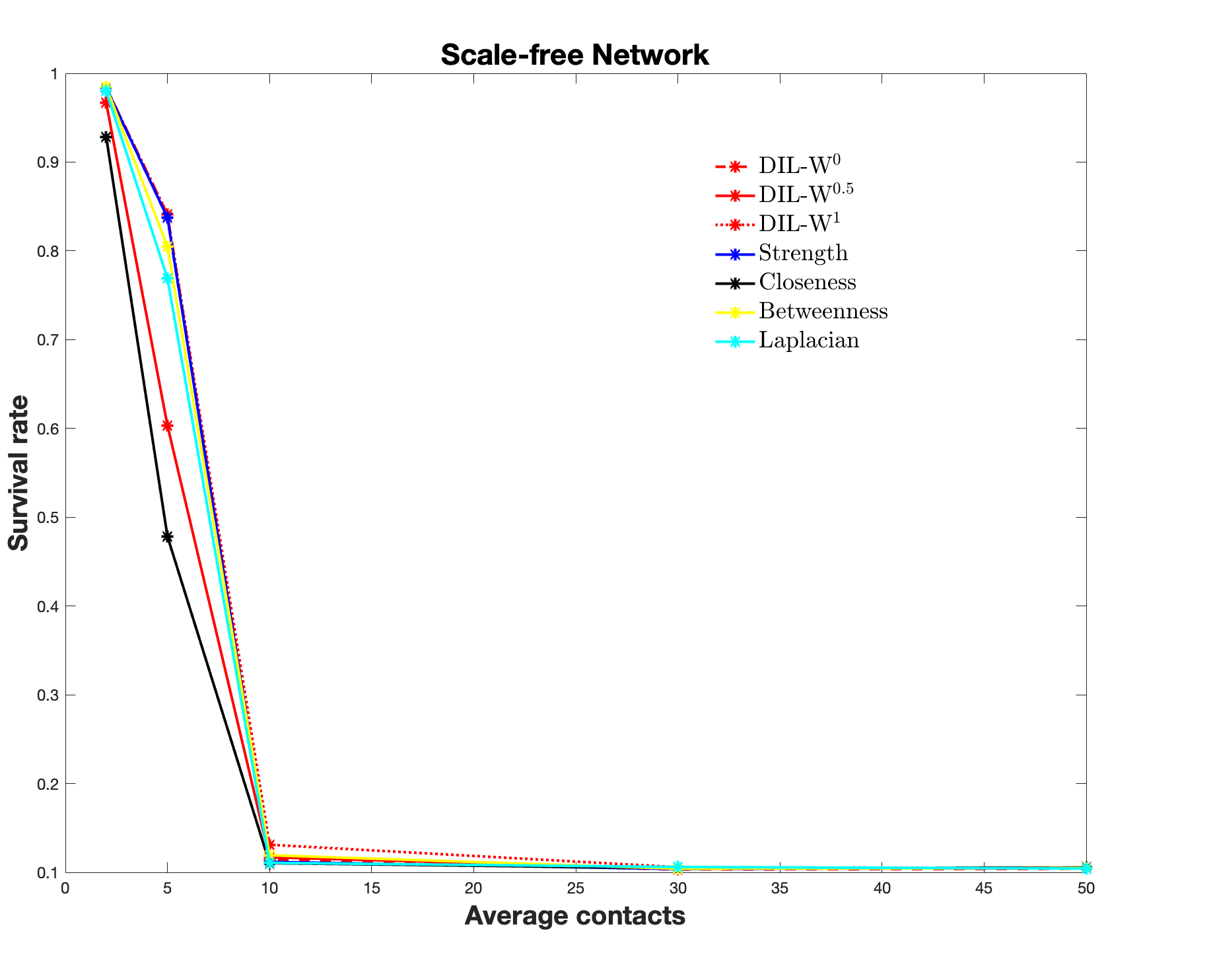}}
\caption{$\sigma$ as a function of $d$ on a Scale-free~network.} \label{free2}
\end{figure}

\begin{Remark}
If the weights are uniformly distributed between $\epsilon$ and $1$, with~$\epsilon$ close to $0$, then the results are equals to the above.
\end{Remark}

\section{Conclusions}\label{conclusionpro}

In this chapter, we evaluate the effectiveness of the DIL-W$^{\alpha}$ ranking in the immunization of nodes that are attacked by an infectious disease that spreads on an edge-weighted graph using a graph-based SIR model. This protection strategy shows better survival rates in 3 of the four networks that we have considered, conditioned to the protection budget equal to 10\% of the network nodes. In~the case where DIL-W$^{\alpha}$ does not perform better than the other strategies, it is still a good alternative due to its lower computational~complexity.

When we modify the protection budget, the~DIL-W$^{\alpha} $ ranking continues to show itself as one of the 3 main immunization strategies with the best survival rate. Even for more than a value of $\alpha$. Again, when it is not in the top places, with~respect to performance, its computational complexity makes DIL-W$^{\alpha} $ a good alternative to choose over the~others.

{In the case of scale-free networks, the~DIL-W$^{\alpha}$ ranking maintains the trend of obtaining better survival rates than the other compared rankings when the protection budget is modified. The~size of the network does not affect this result, because~DIL-W$^{\alpha}$ reaches the first place with $\alpha=0$ or $\alpha=1$. Something similar happens when modifying the average number of contacts, DIL-W$^{1}$ has the best performance up to 30\% of total nodes.
}

An interesting and complex task to solve is to determine which value of $\alpha $ will be chosen in certain networks so that the ranking generated would be the optimal one. Not always the same value makes the best performance. However, when considering this method, there are as many rankings as there are numbers between 0 and~1.

\chapter[Protection Strategy Applied to a Covid-19 Case]{Protection Strategy against an Epidemic Disease on Edge-Weighted Graphs Applied to a Covid-19 Case}\label{covid}
\minitoc

\section{Introduction}

Infectious diseases have been the focus of multiple fields of research. In public health and epidemiology, efforts are directed at establishing transmission dynamics, the characteristics of infectious agents, and the populations most affected by pathogens, among others, which are of high importance for science \cite{nature}. In recent decades, research on infectious diseases has involved the application of complex theories from mathematics and engineering. In particular, the use of network models has allowed explanations of the spread of diseases from infected people (nodes) and their links with others (edges) \cite{lloyd2007network}.

Network models establish the connection between population groups, which is useful not only in the field of public health or epidemiology, but also in engineering and social sciences \cite{haslbeck2018well} (see \cite{an2014synchronization,montenegro2019linear,guangzeng2009novel,mersch2013tracking,firth2015experimental}). In the health field, being a theoretical approach, the importance of recognizing the complexities of community structures has been discussed in order to understand social dynamics in the spread of infectious diseases \cite{wu2008community}.  For example, Magelinski et al. developed a model to estimate the role played by certain nodes in community structures \cite{magelinski2021measuring}, while Ghalmane et al. included the dimension of centrality in complex networks \cite{ghalmane2019centrality}.

Among the diverse and important applications that networks currently have is the modeling of infectious diseases. Immunization, or the process of protecting nodes in the network, plays a key role in stopping diseases from spreading. Hence the importance of having tools or strategies that allow the solving of this challenge. 

This study aims to analyze the protection effect against COVID-19 using the DIL-W$^{\alpha}$ ranking with real data from a city in Chile (Olmu\'e-City), obtained from the Epidemiological Surveillance System of the Ministry of Health of Chile, from which we obtain an edge-weighted graph, denoted by $G_{\mathcal{E}}$, according to the method proposed in \cite{ronald1}. We apply the protection to the $G_{\mathcal{E}}$ network, according to the importance ranking list produced by DIL-W$^{\alpha}$, considering different protection budgets. For the ranking DIL-W$^{\alpha}$, we consider three different values for $\alpha$; they are $0$, $0.5$ and $1$. In this way, we compare how the protection performs according to the value of $\alpha$. We use a graph-based SIR model, namely, each individual is represented by a vertex in $G_{\mathcal{E}}$. At time $t$, each vertex $v_i$ is in a state $v_i^t$ belonging to $\mathcal{S}=\{0,1,-1\}$, where $0,1$ and $-1$ represent the three discrete states:  Susceptible (S), Infected (I) and Recovered or Removed (R). Five hundred simulations were performed on $G_{\mathcal{E}}$; the initial population contains one infected node and all the simulations, considering $\displaystyle\delta=\frac{1}{15}$ (recovered rate).

\section{Method}\label{meth}
The data that are modeled correspond to the city of Olmu\'e (Valpara\'iso {region}, Chile) and were obtained from the database of the Epidemiological Surveillance System of the Ministry of Health of Chile, which included the notified cases (positive or negative) and their contacts from 3 March  2020 to 15 January  2021 with a total of 3866 registered persons. 

We denote by $\mathcal{E}_{pi}$ the database of the Epidemiological Surveillance System of the Ministry of Health of Chile. From the total of variables included in $\mathcal{E}_{pi}$ ($K=279$) 7 of them are relationship variables ($K_1=7$). They are: full address ($\mathcal{X}_1$); the street where the people live ($\mathcal{X}_2$); town ($\mathcal{X}_3$); place of work ($\mathcal{X}_4$); workplace section ($\mathcal{X}_5$); health facility where they were treated ($\mathcal{X}_6$) and the region of the country where the test was taken to confirm, or not, the contagion ($\mathcal{X}_7$).

In our criteria, the hierarchical order of the seven variables in descending form is  $\mathcal{X}_1, \mathcal{X}_2, \mathcal{X}_3, \mathcal{X}_4, \mathcal{X}_5, \mathcal{X}_6, \mathcal{X}_7$. Moreover, we consider that the variables $\mathcal{X}_1, \mathcal{X}_2 $ and  $\mathcal{X}_3$ have the same weight. In the same way, we also consider the variables $\mathcal{X}_4$ and $\mathcal{X}_5$ with equal weight. Hence,
$A_1=\{\mathcal{X}_1, \mathcal{X}_2, \mathcal{X}_3\}$, $A_2=\{\mathcal{X}_4, \mathcal{X}_5\}$, $A_3=\{\mathcal{X}_6\}$ and $A_4=\{\mathcal{X}_7\}$  are the different classes that are defined by the different weights. Hence, by Definition \ref{weightvariable}.
$$p_1=\frac{3}{7},\quad p_2=\frac{2}{7},\quad p_3=\frac{1}{7},\quad p_4=\frac{1}{7}.$$

Figure \ref{grapholmue} shows the obtained graph. 

\begin{figure}[h]
\centering
\includegraphics[width=0.9\textwidth]{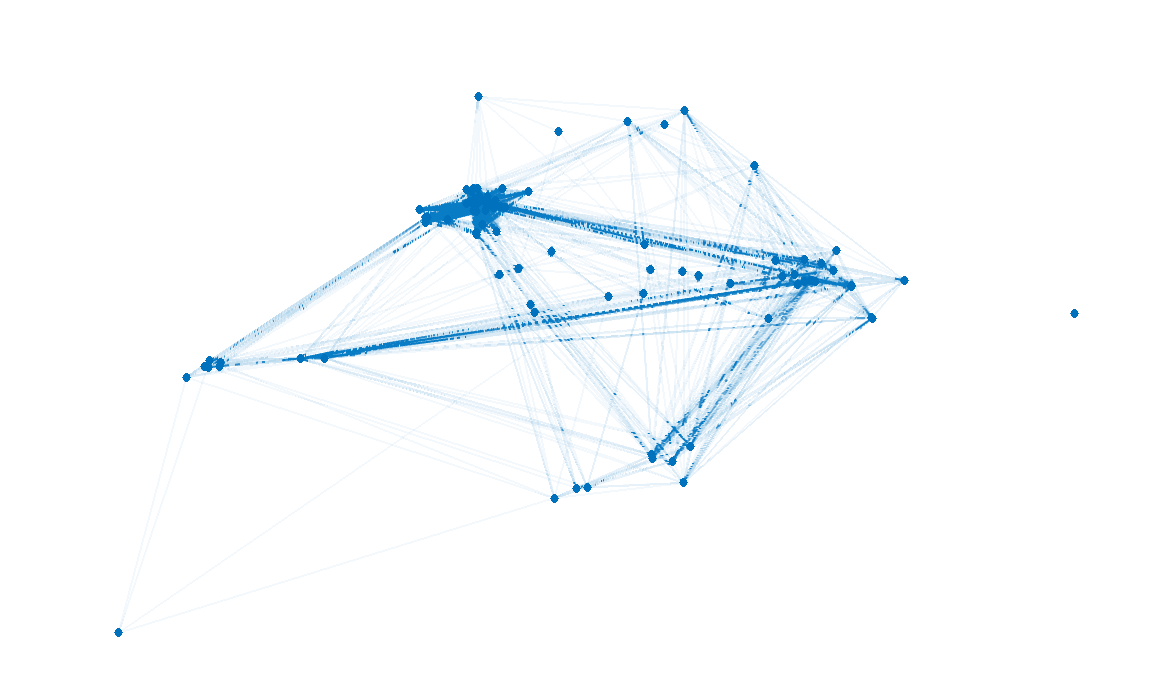}
\caption{Graph obtained from database of Olmu\'e city, Chile, with  $3866$ vertices and 6,841,470 edges.} \label{grapholmue}
\end{figure}

Let us denote by $G_{\mathcal{E}}$ the graph obtained from database $\mathcal{E}_{pi}$.
\newpage
\section{Protection strategy}
Our chosen protection strategy corresponds to the DIL-W$^{\alpha}$ ranking (\cite{ronaldranking}) described in Chapter \ref{pro}.
Furthermore, \cite{ronald2} evaluated the effectiveness of the DIL-W$^{\alpha}$ ranking in the immunization of nodes that are attacked by an infectious disease that spreads on an edge-weighted graph using a graph-based SIR model.
In summary, we apply the protection to the $G_{\mathcal{E}}$ network, according to the importance ranking list produced by DIL-W$^{\alpha}$, considering different protection budgets. For the ranking DIL-W$^{\alpha}$, we consider three different values for $\alpha$; they are $0$, $0.5$ and $1$. In this way, we compare how the protection performs according to the value of $\alpha$.

\section{Results}\label{results}
In this section, we use a graph-based SIR model in the same way as in \cite{ronald1,ronald2}, namely, each individual is represented by a vertex in $G_{\mathcal{E}}$. At time $t$, each vertex $v_i$ is in a state $v_i^t$ belonging to $\mathcal{S}=\{0,1,-1\}$, where $0,1$ and $-1$  represent the three discrete states:  Susceptible (S), Infected (I) and Recovered or Removed (R).  

Moreover, we assume that the disease is present for a certain period of time and that, when individuals recover, they are immune, {that is, reinfection is not considered.}

The initial population contains one infected node and all the simulations that consider $\displaystyle\delta=\frac{1}{15}$. 
Five hundred simulations were performed on $G_{\mathcal{E}}$ with $\rho=0.00121$. Figure \ref{ajus1} shows the average infected curve and the real infected data in $\mathcal{E}_{pi}$. Moreover, it shows a curve fitted to the data following the SIR model; for this, we used the classic method of least squares to compare with our proposal.

\begin{figure}[h]
\begin{center}
    {\includegraphics[width=0.9\textwidth]{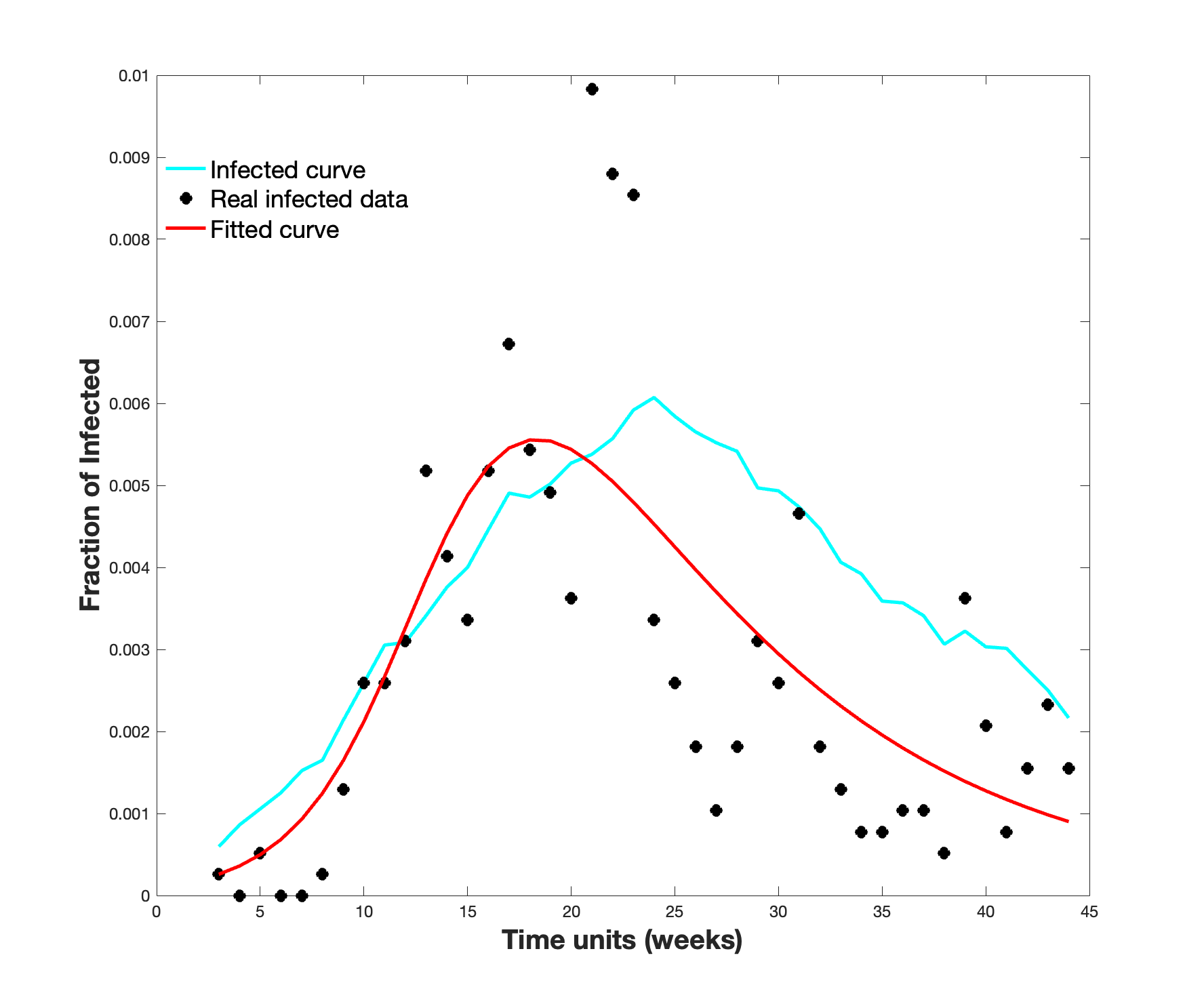}}
  \caption{Real infected data (black), fitted curve (red), and infected curve obtained in the spread on $G_{\mathcal{E}}$ (cyan).} 
  \label{ajus1}
  \end{center}
\end{figure}
\newpage
The graph $G_{\mathcal{E}} $ was protected with different protection budgets according to the importance of the DIL-W$^0$, DIL-W$^{0.5}$ and DIL-W$^1$ rankings. Protection is carried out in week 1, this is to say, at the beginning of the spread of the disease. Figure \ref{varia_protec} shows the results.

\begin{figure}[h]
\begin{center}
  \subfigure
    {\includegraphics[width=0.4\textwidth]{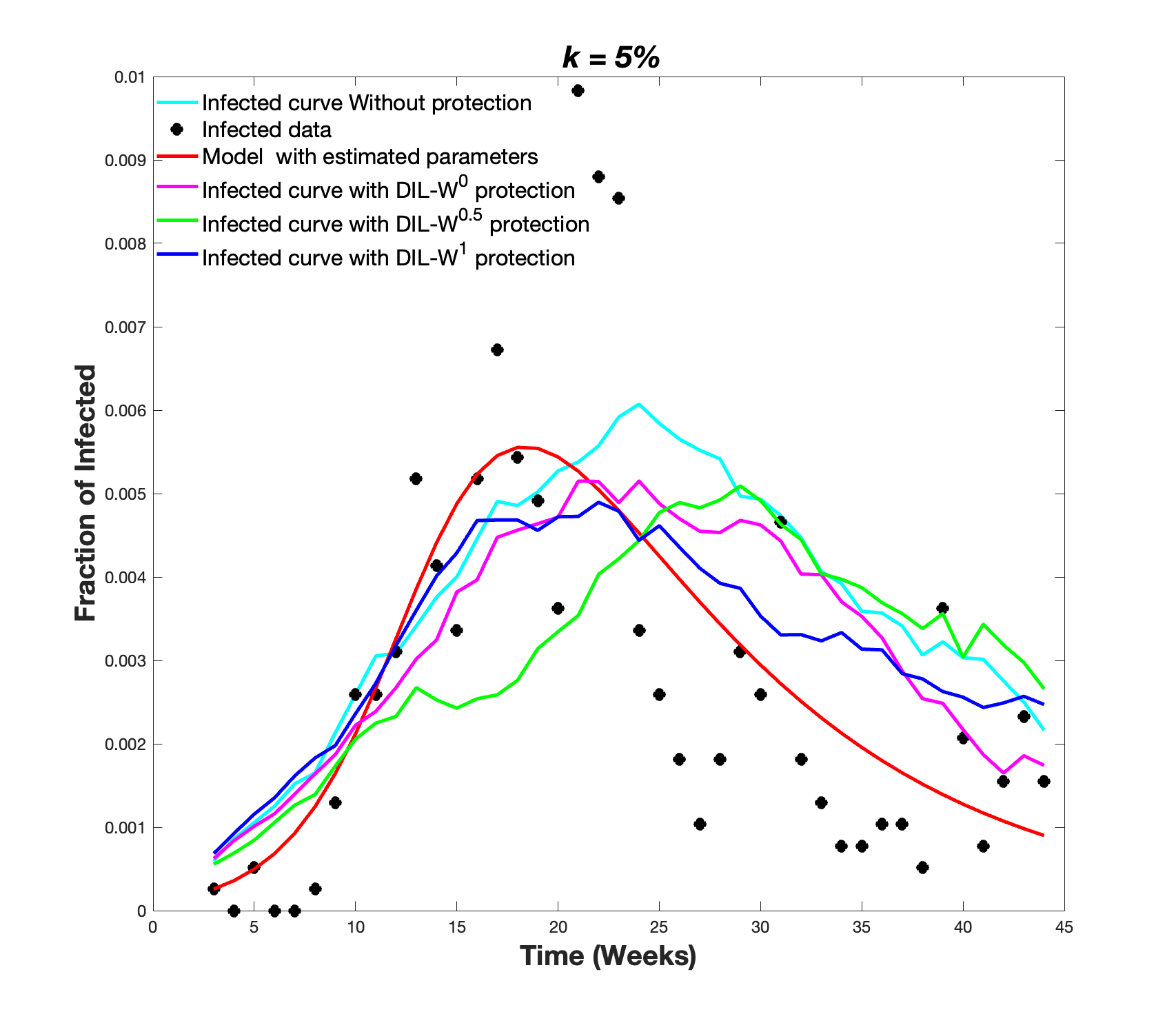}}
  \subfigure
    {\includegraphics[width=0.4\textwidth]{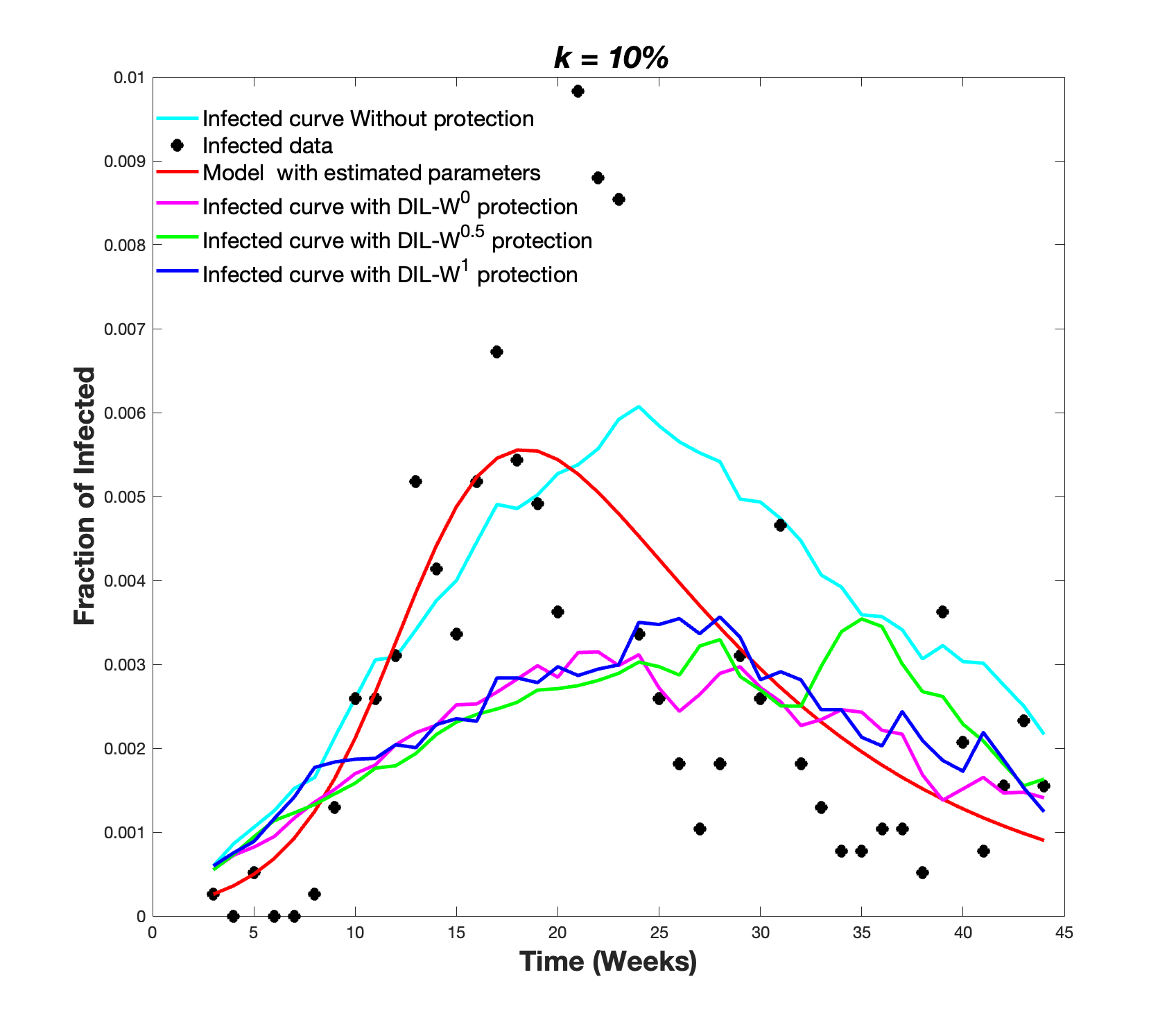}}
    \subfigure
    {\includegraphics[width=0.4\textwidth]{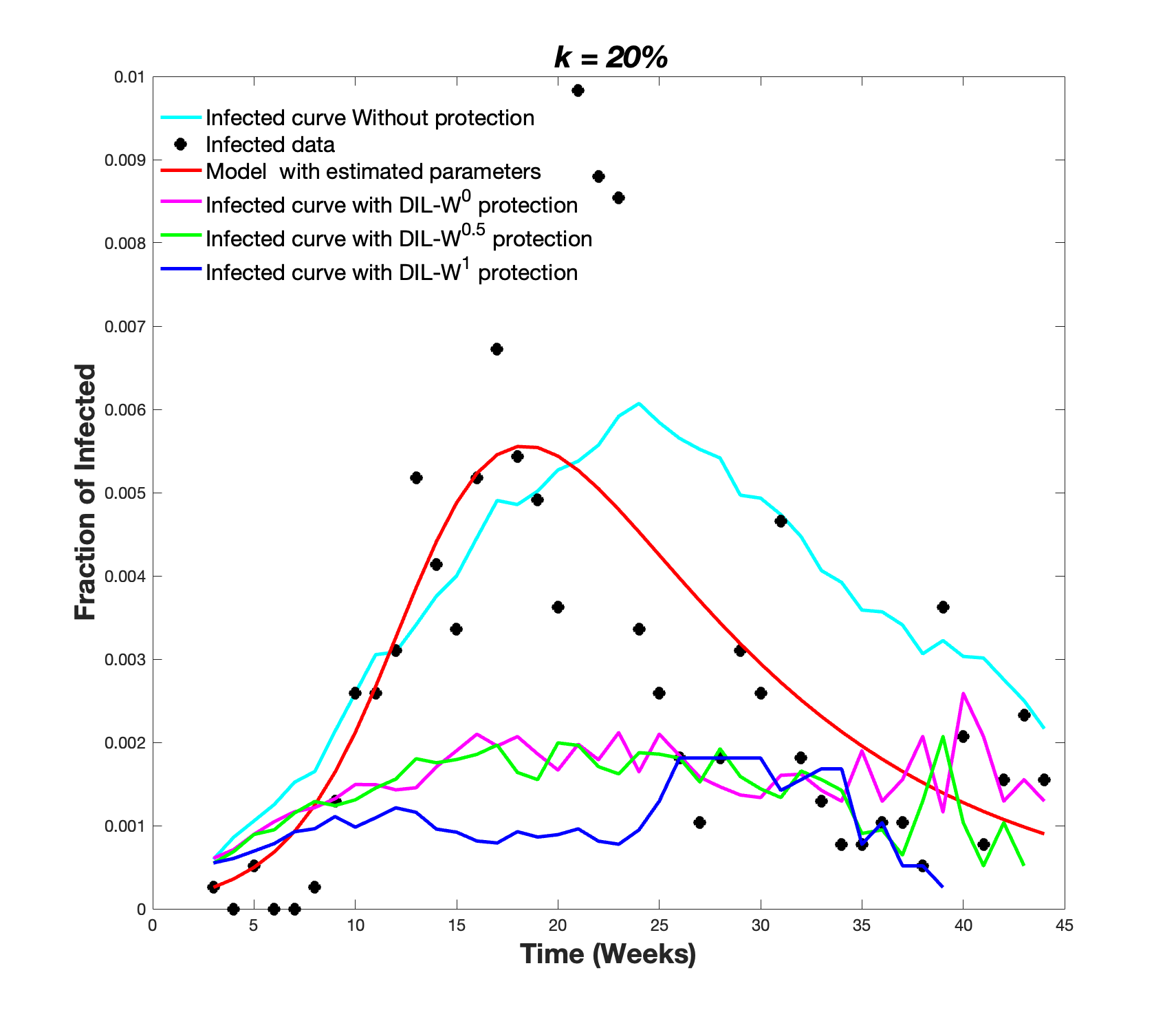}}
  \subfigure
    {\includegraphics[width=0.4\textwidth]{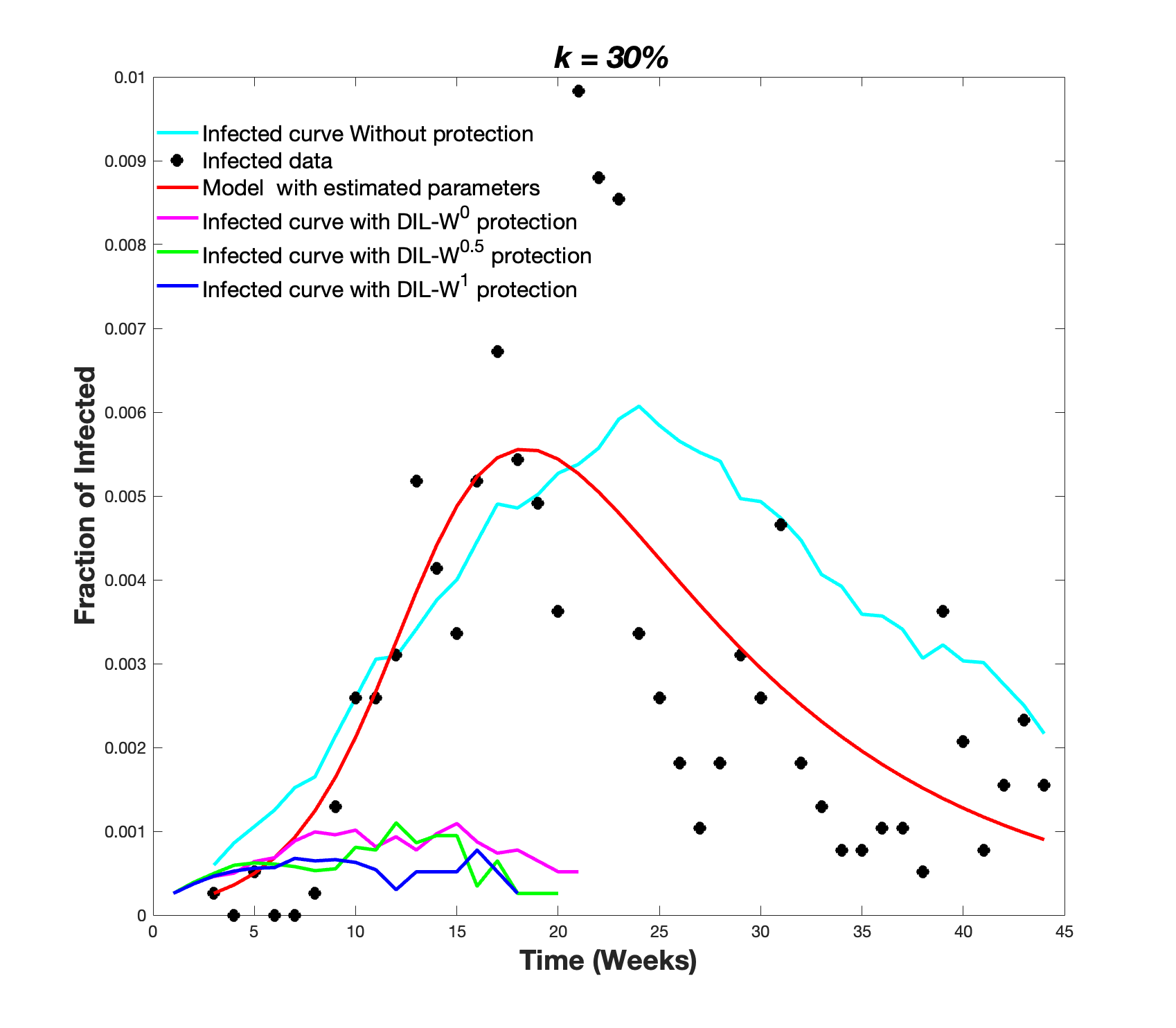}}
\subfigure
    {\includegraphics[width=0.4\textwidth]{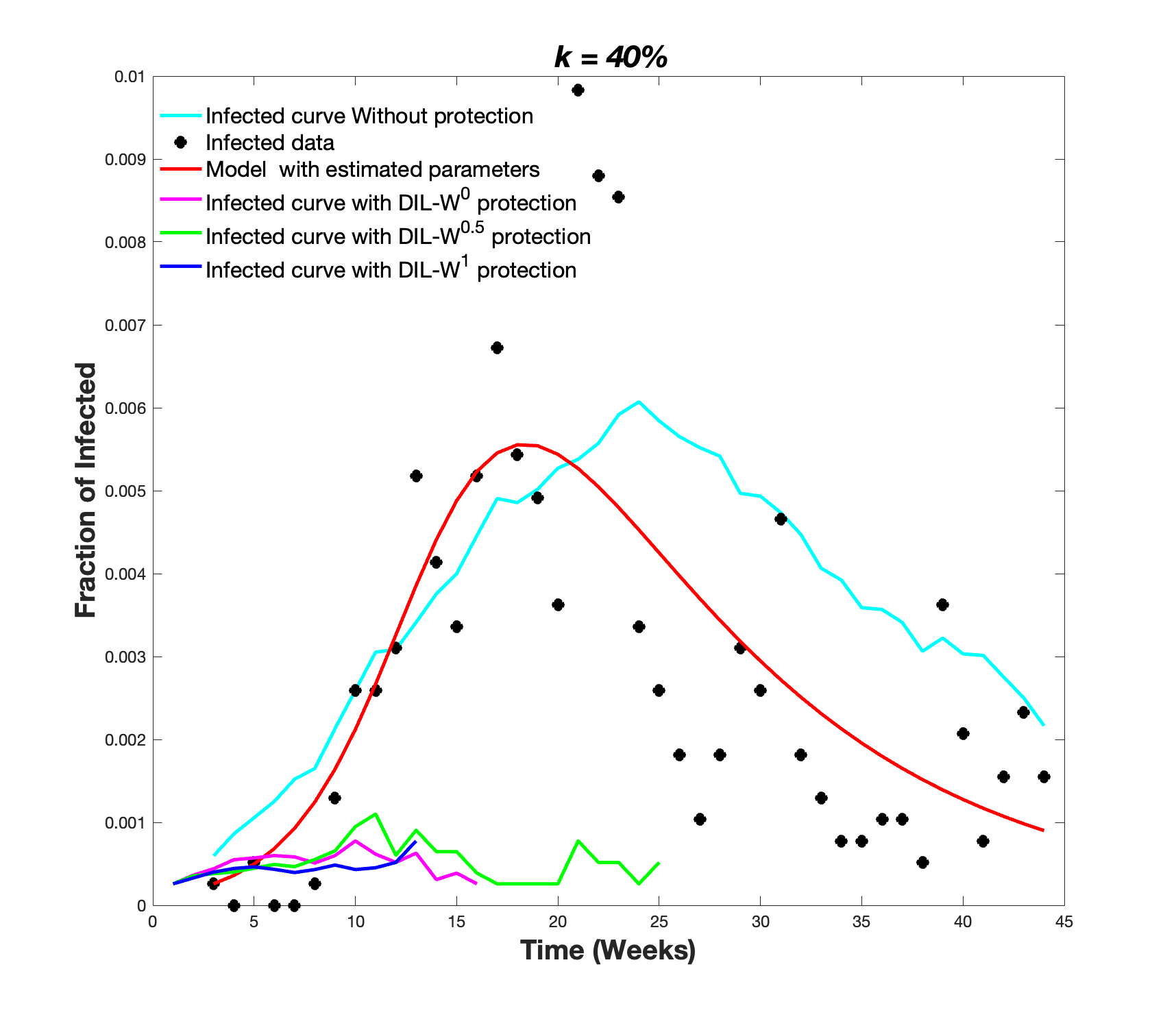}}
  \subfigure
    {\includegraphics[width=0.4\textwidth]{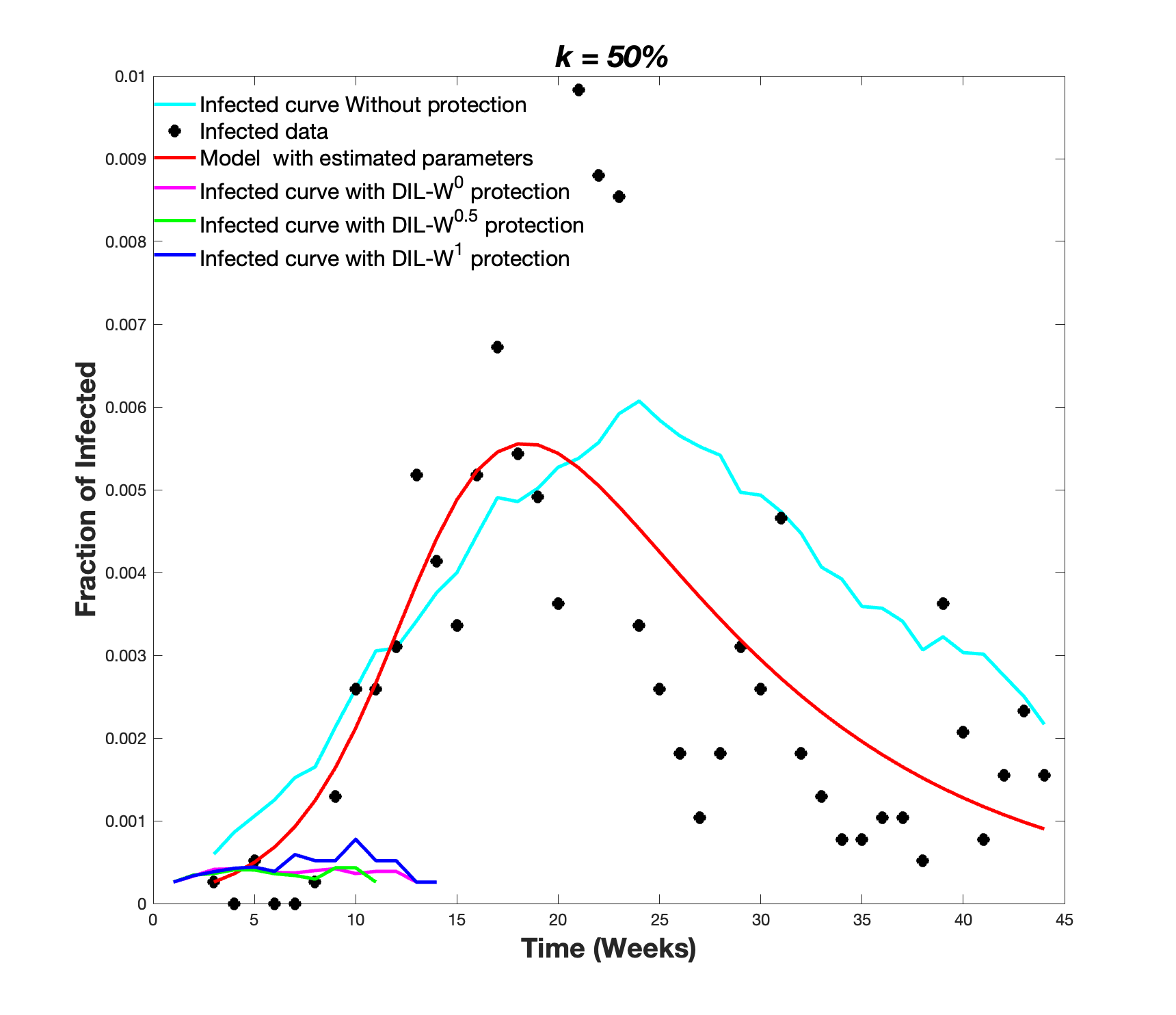}}
    \subfigure
    {\includegraphics[width=0.4\textwidth]{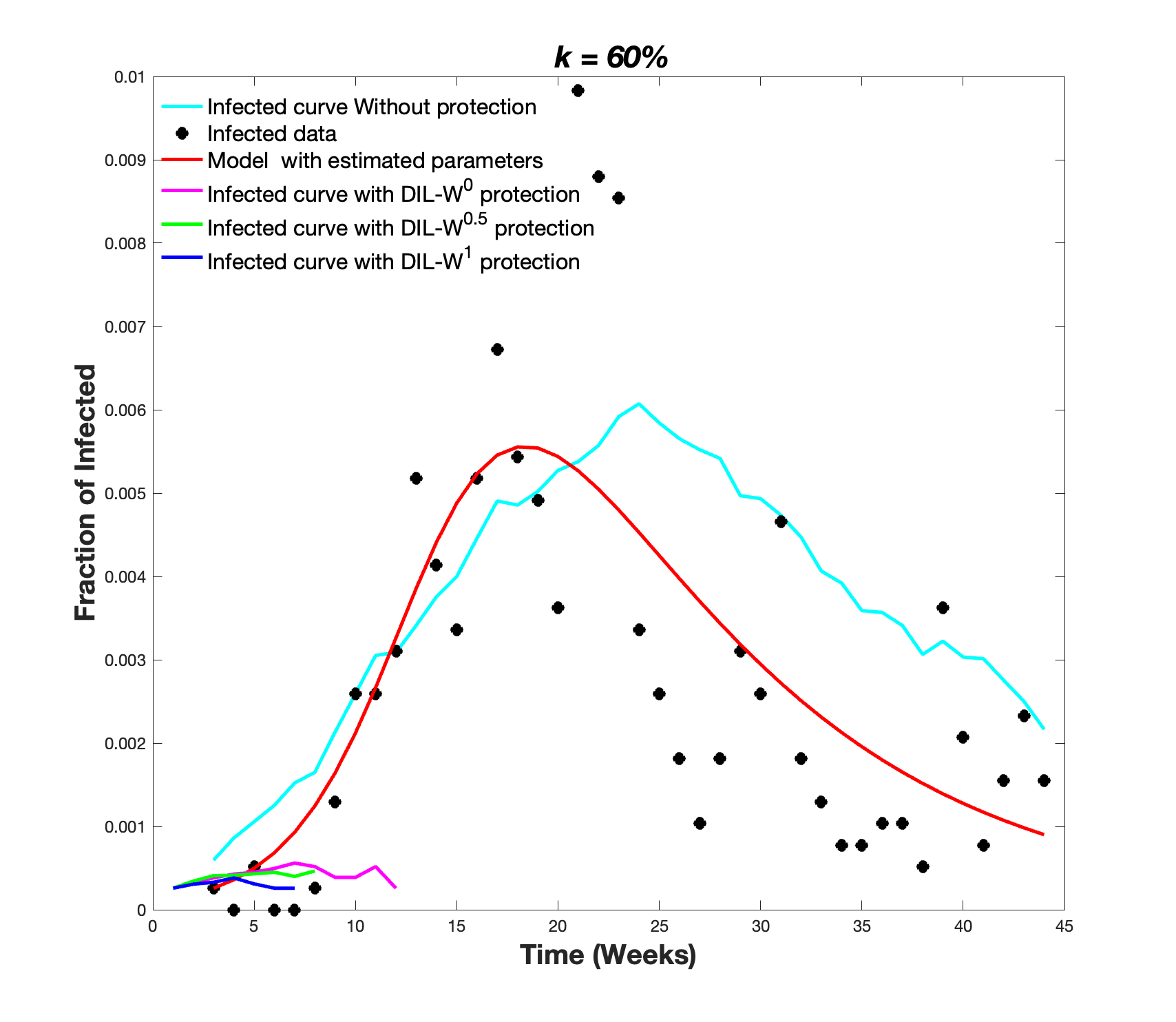}}
  \subfigure
    {\includegraphics[width=0.4\textwidth]{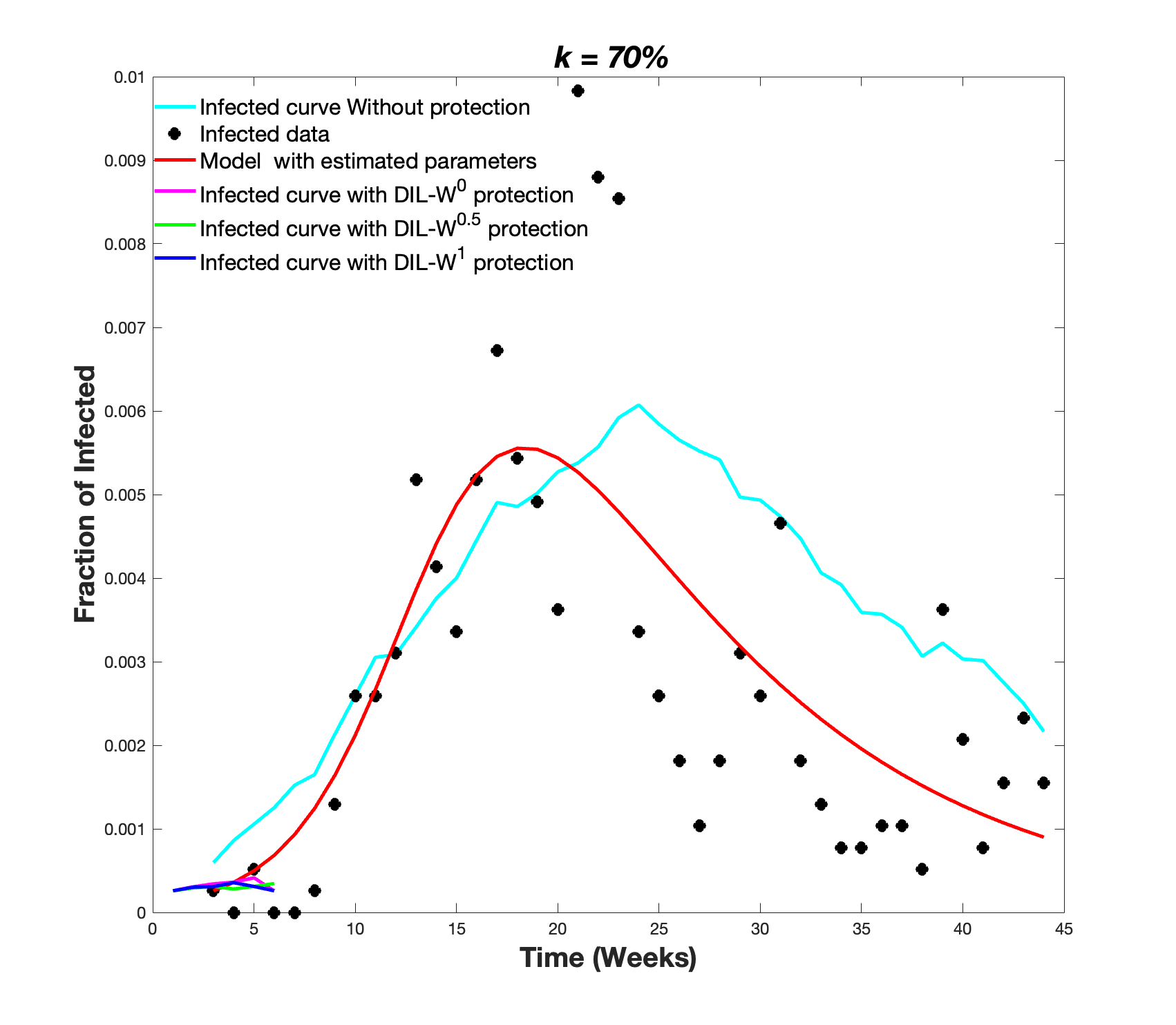}}
  \caption{Infected curves obtained according to different values of $k$.} 
  \label{varia_protec}
  \end{center}
\end{figure}
\clearpage

We can see the survival rate in Figure \ref{sup_protec}.
\begin{figure}[h]
\begin{center}
  \subfigure
    {\includegraphics[width=0.49\textwidth]{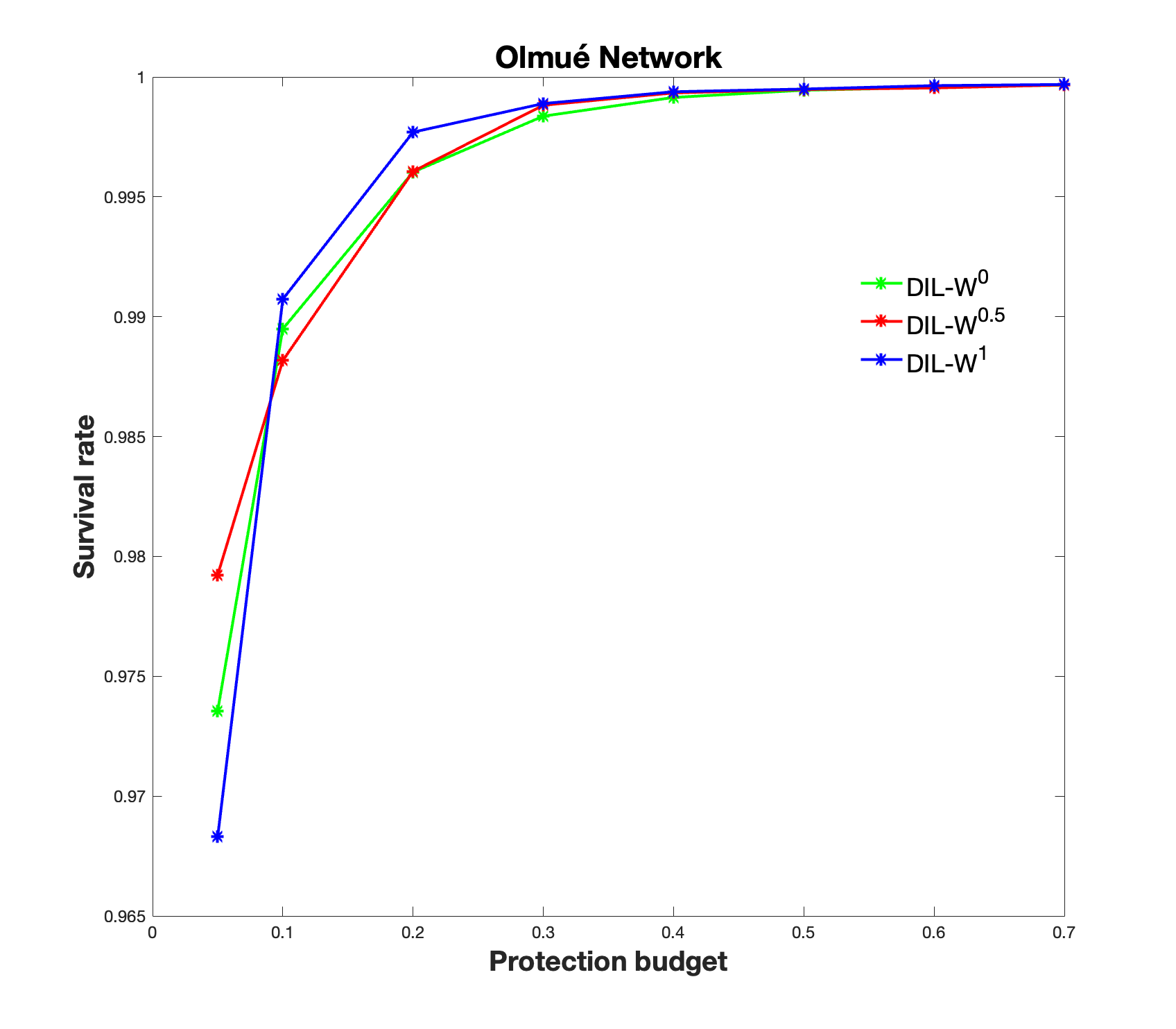}}
  \subfigure
    {\includegraphics[width=0.49\textwidth]{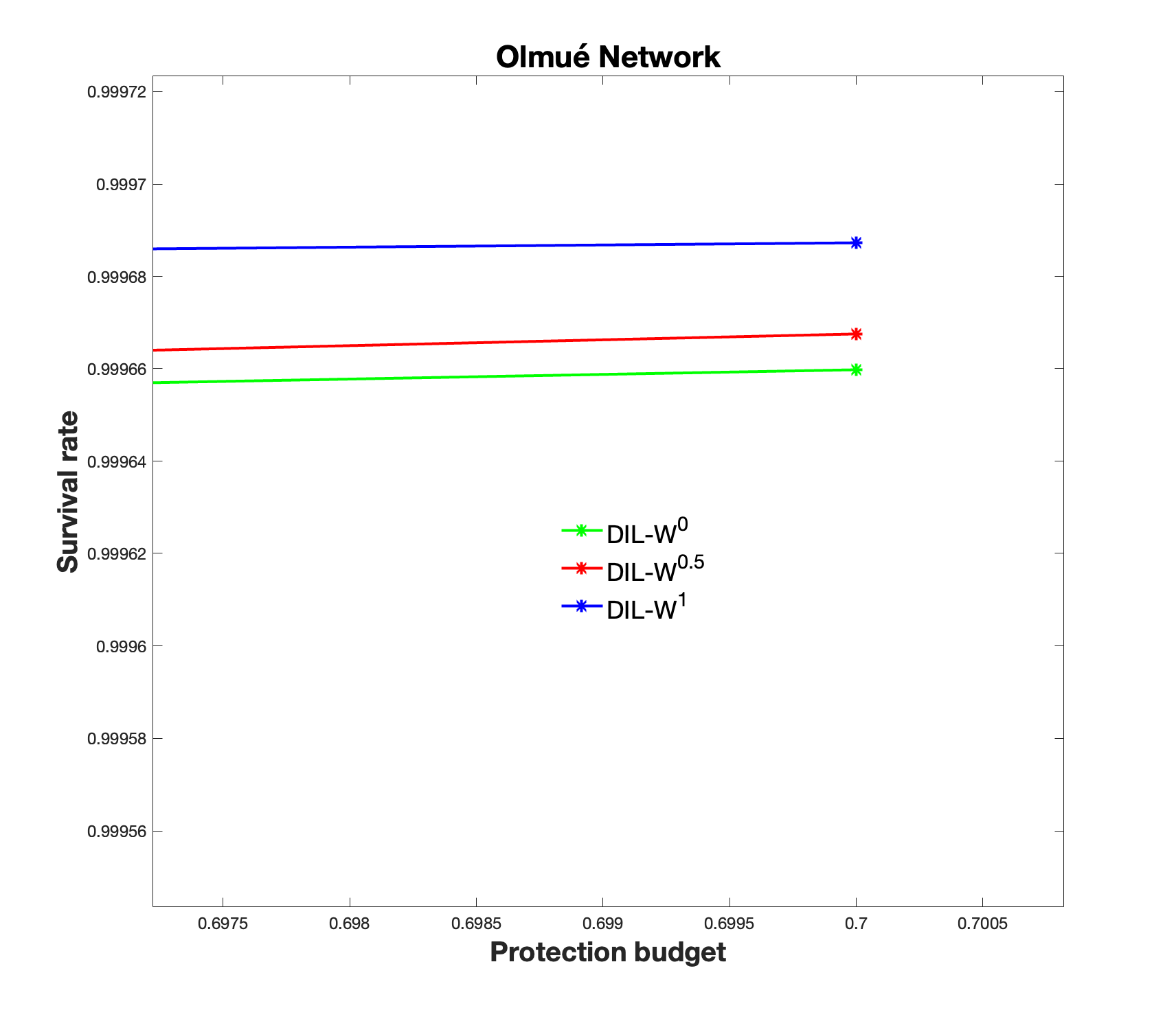}}
  \caption{Survival rate of each ranking.} 
  \label{sup_protec}
  \end{center}
\end{figure}

Figure \ref{real_top} shows the relationship between the real infected (450 people) and those immunized according to our proposal.
\begin{figure}[h]
\begin{center}
    {\includegraphics[width=0.6\textwidth]{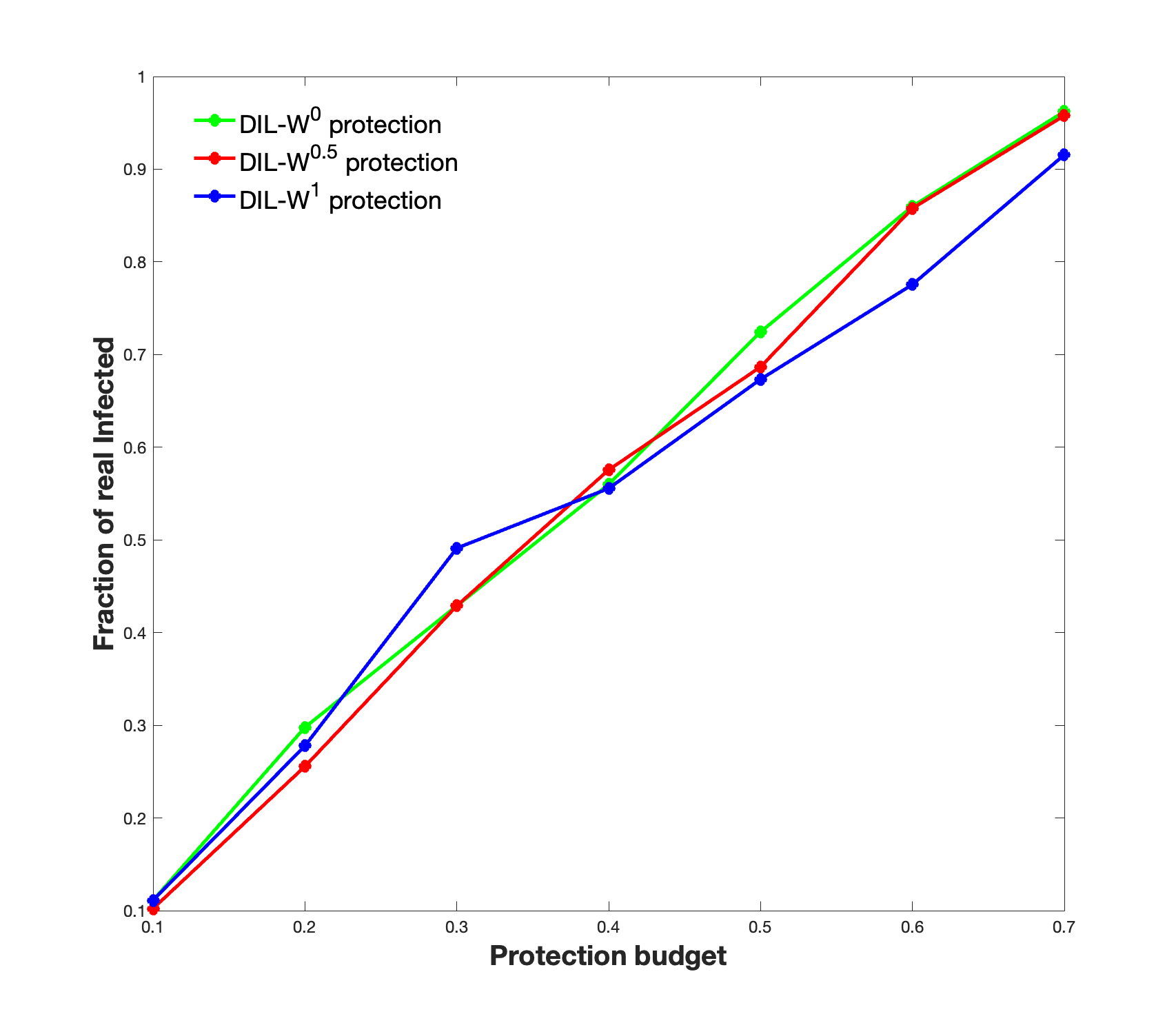}}
  \caption{Relationship between the real infected and those immunized according to the different DIL-W$^{\alpha}$ ranking.} 
  \label{real_top}
  \end{center}
\end{figure}

We can see that 80\% of the real infected, is located in the 60\% of the top ranked according to DIL-W$^{\alpha}$.
We think that it is a way to recognize those who will get sick. But it is not the solution.

Another element that we have considered investigating is the time in which the protection takes place. We modified the protection in the graph as the weeks advanced. In Figure \ref{varia_week}, we can see the different infected curves, considering the 10\% protection according to the DIL-W$^{\alpha}$ ranking, 
\begin{figure}[h]
\begin{center}
  \subfigure
    {\includegraphics[width=0.4\textwidth]{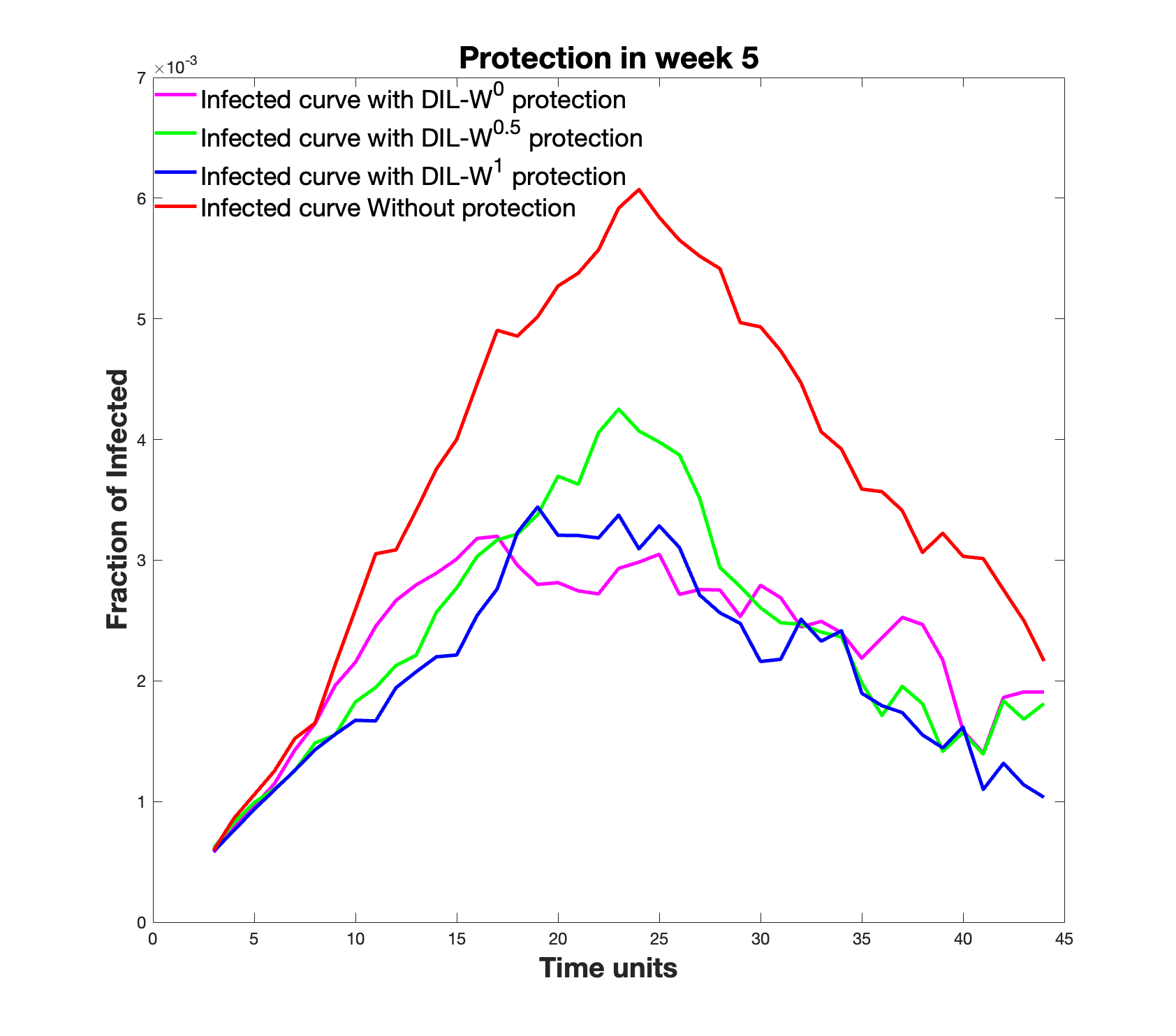}}
  \subfigure
    {\includegraphics[width=0.4\textwidth]{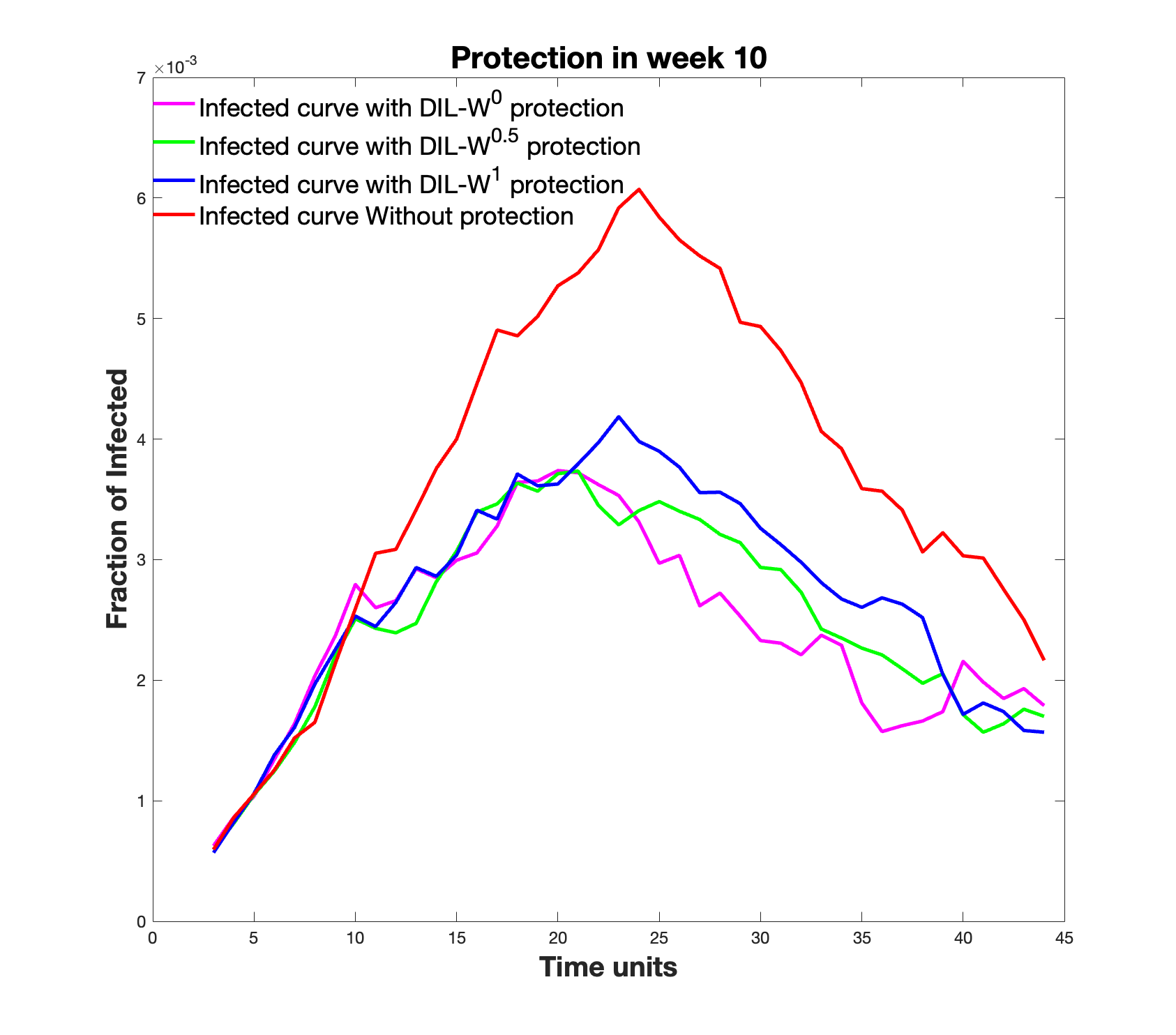}}
    \subfigure
    {\includegraphics[width=0.4\textwidth]{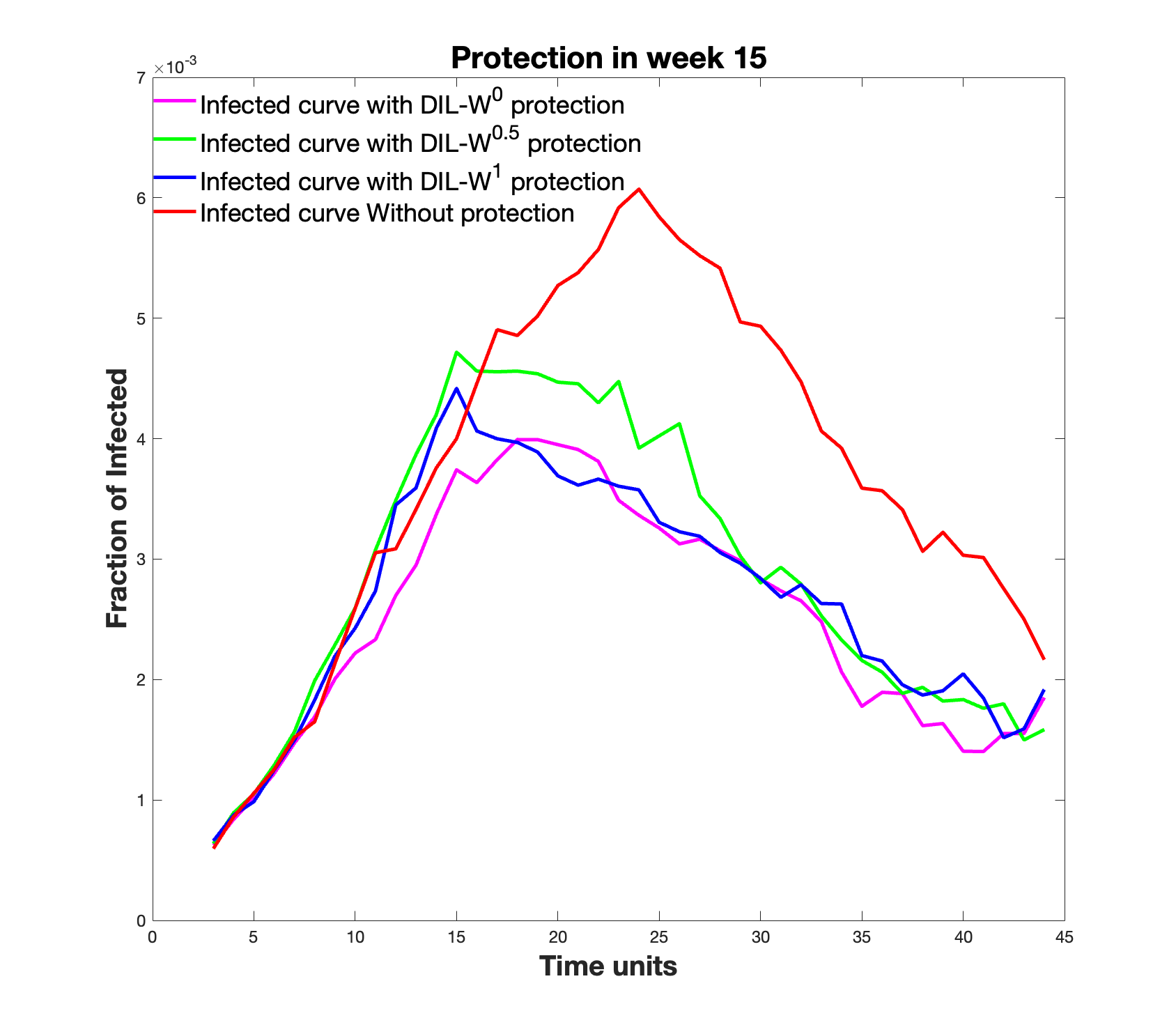}}
  \subfigure
    {\includegraphics[width=0.4\textwidth]{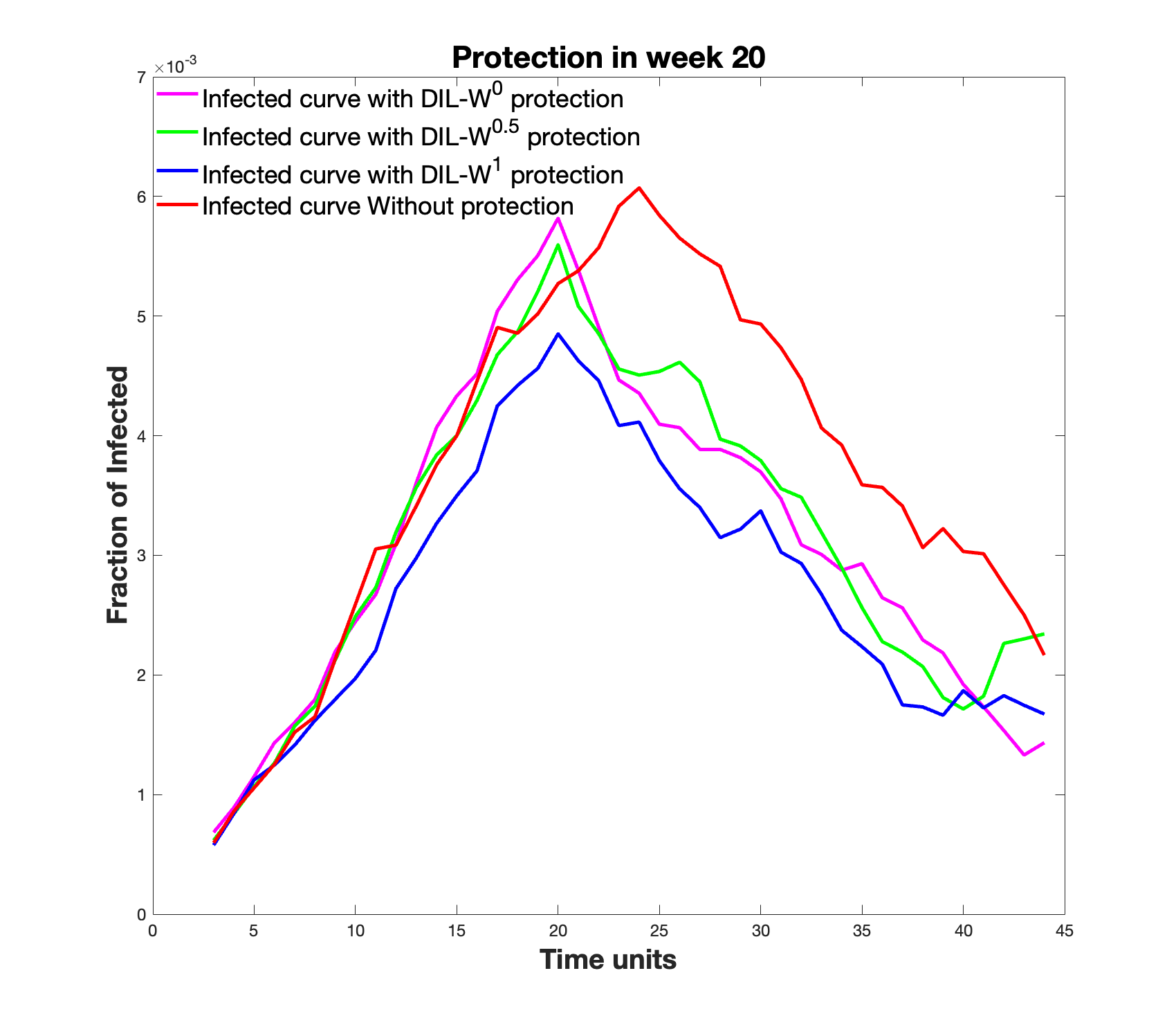}}
\subfigure
    {\includegraphics[width=0.4\textwidth]{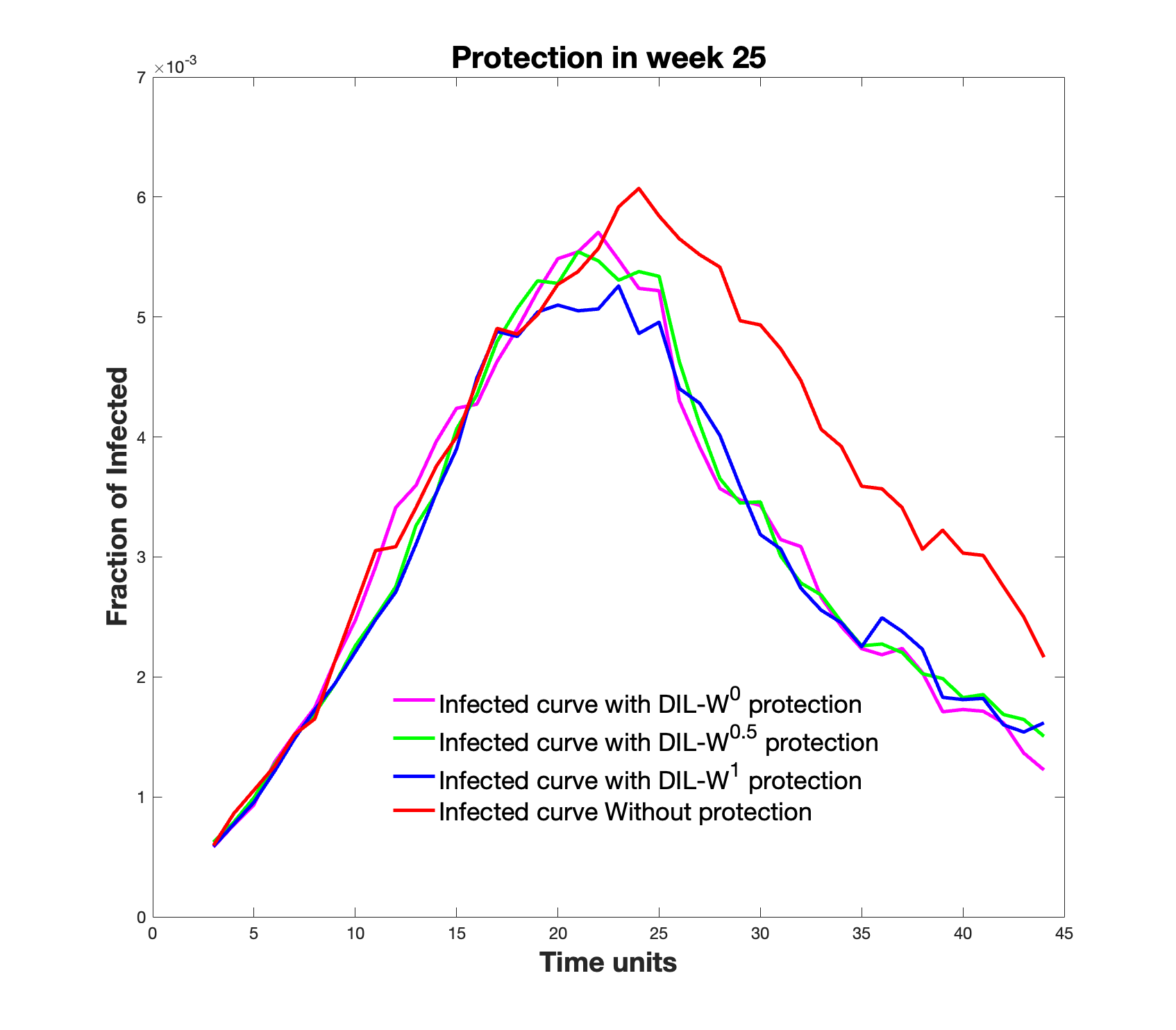}}
  \subfigure
    {\includegraphics[width=0.4\textwidth]{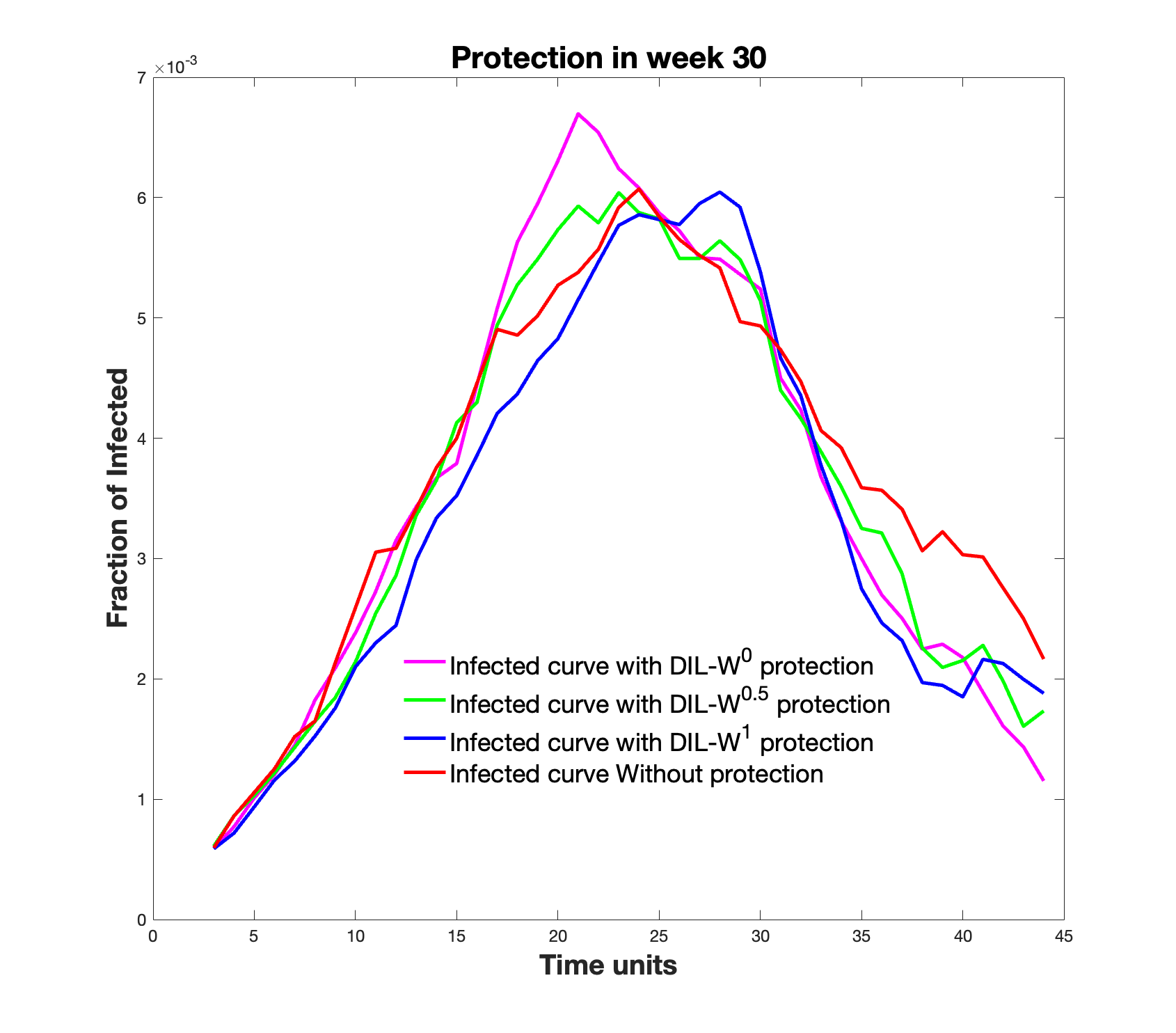}}
    \subfigure
    {\includegraphics[width=0.4\textwidth]{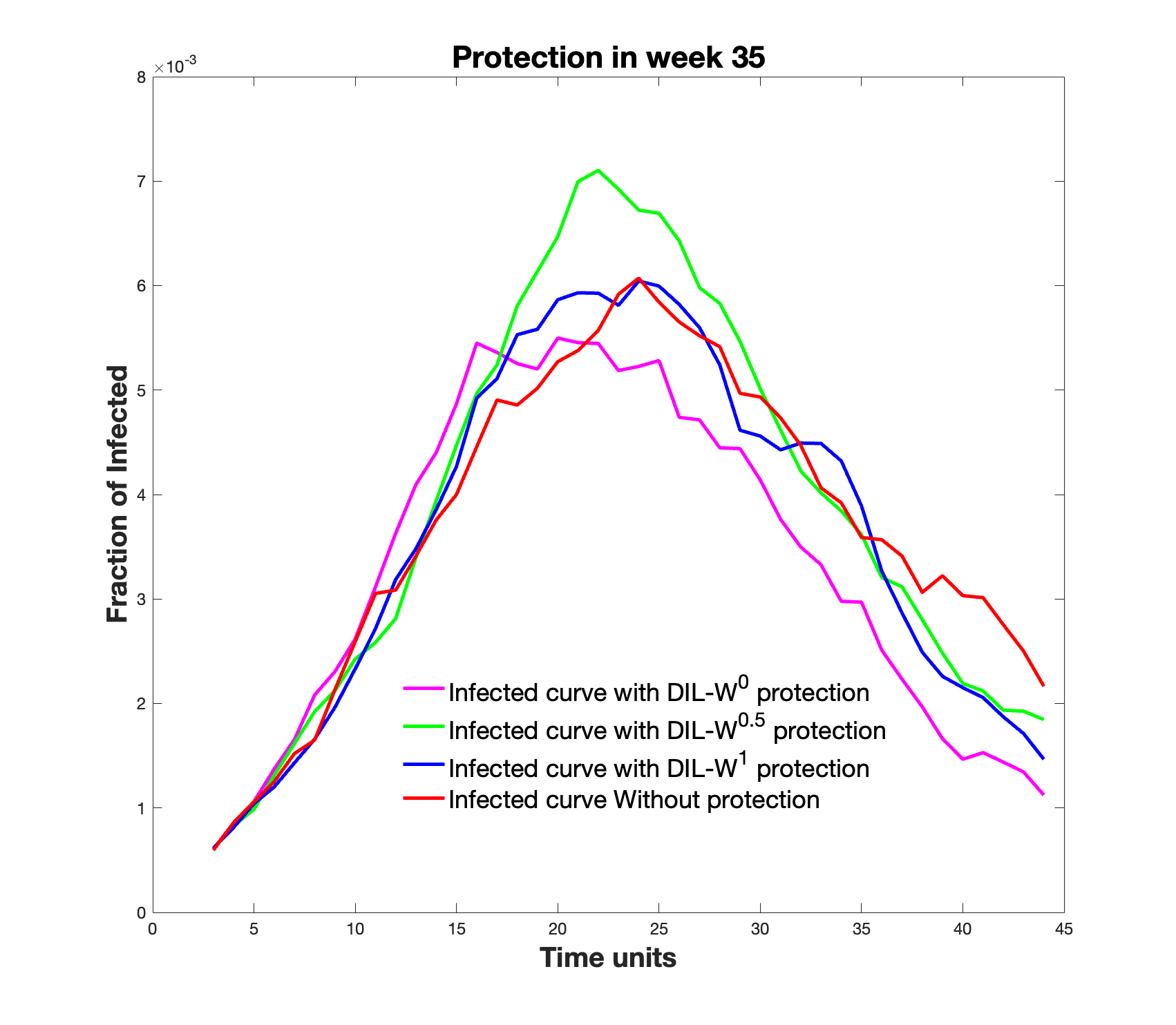}}
  \subfigure
    {\includegraphics[width=0.4\textwidth]{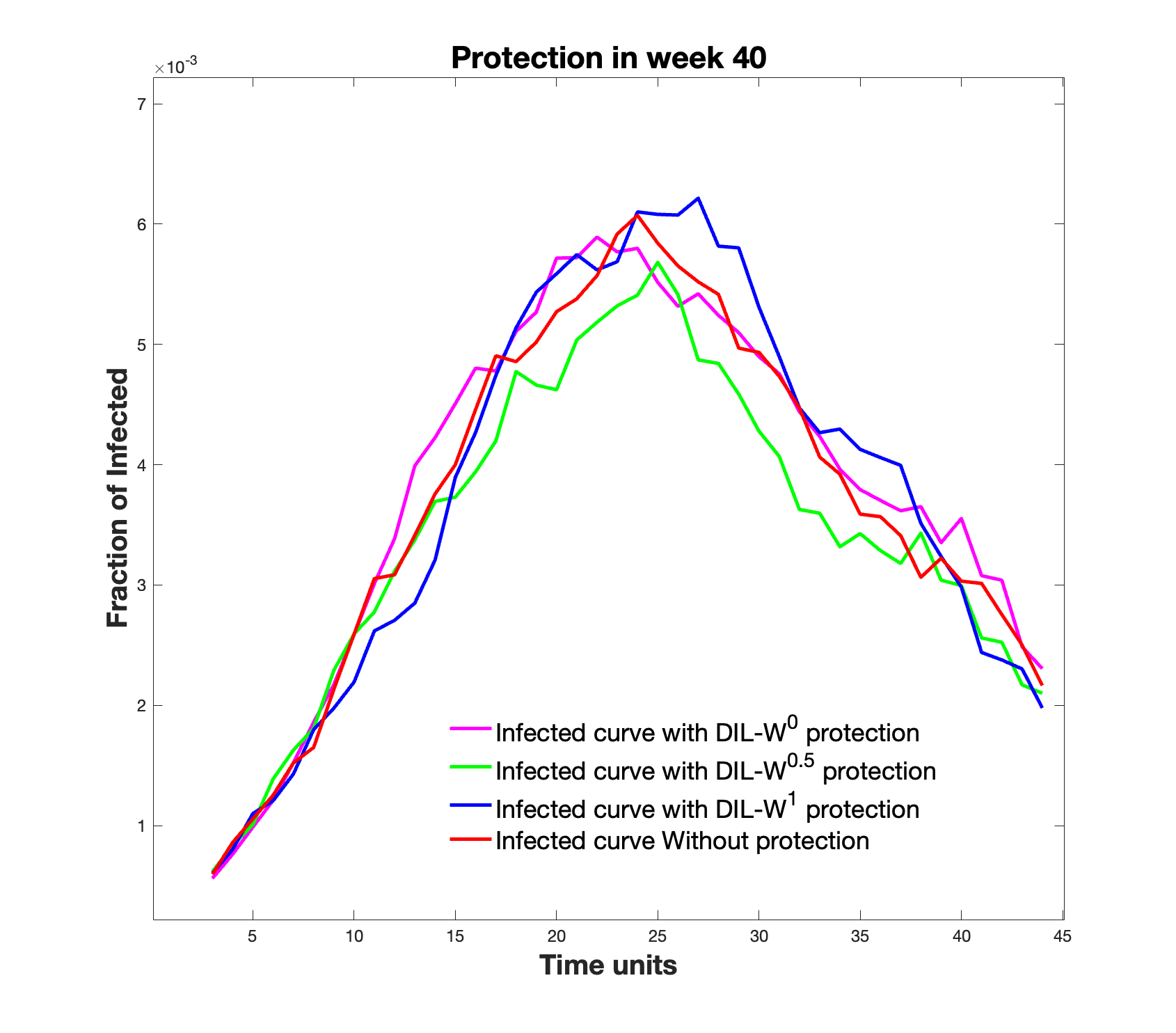}}
  \caption{The different infected curves, considering the 10\% protection according to the DIL-W$^{\alpha}$ ranking and done in different weeks.} 
  \label{varia_week}
  \end{center}
\end{figure}
\clearpage

Figure \ref{sup_week} shows the relationship between the survival rate and week in which the protection is carried out with our proposal. 

\begin{figure}[ht]
\begin{center}
    {\includegraphics[width=0.6\textwidth]{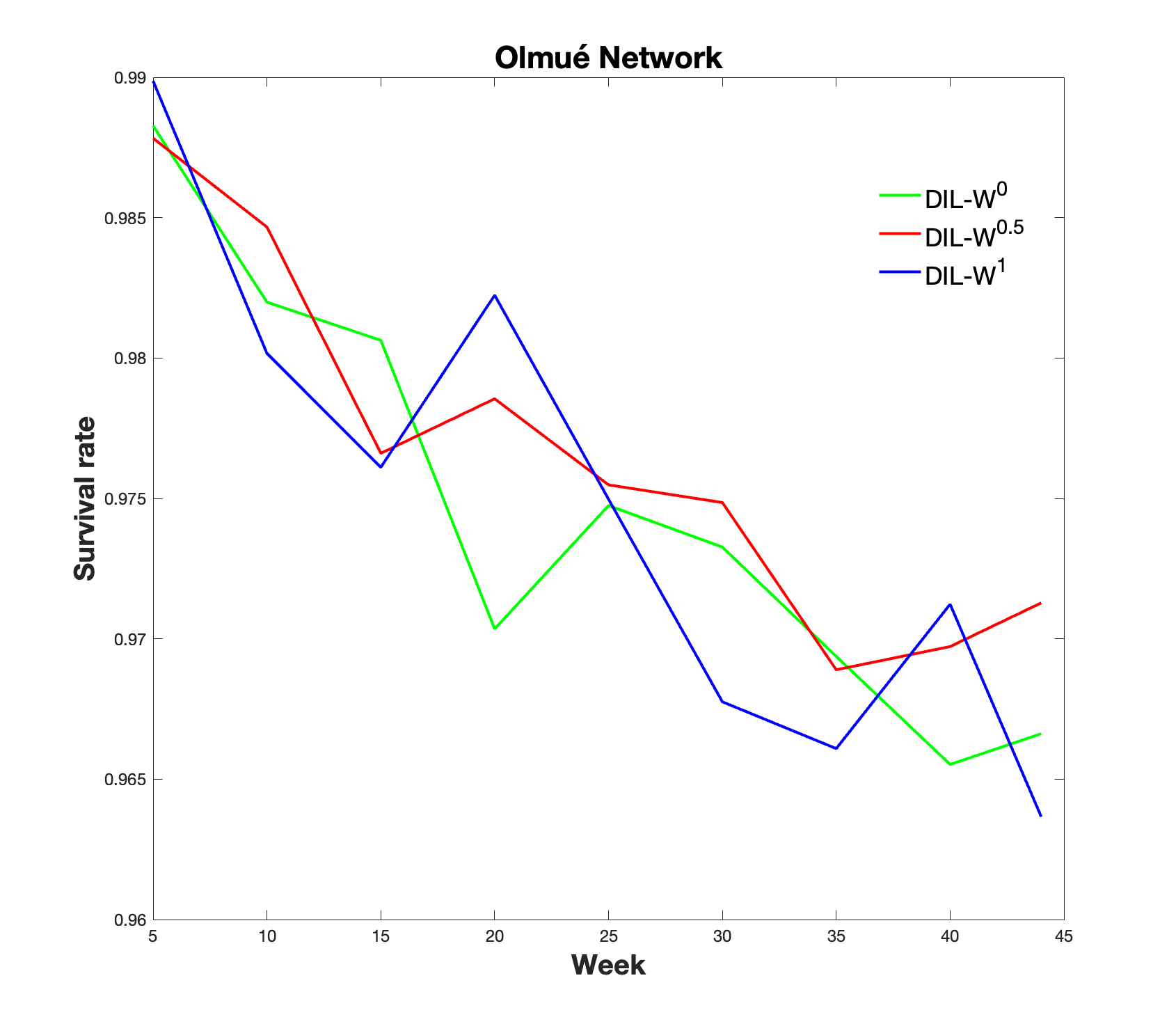}}
  \caption{Survival rate of each ranking.} 
  \label{sup_week}
  \end{center}
\end{figure}
\clearpage
\section{Discussion}\label{discussion}
The results of the present investigation are directed towards the analysis of the effectiveness of immunization using the DIL-W$^{\alpha}$ ranking with real COVID-19 data from the city of Olmu\'e-in Chile. Depending on the importance of the rankings, the immunization results were similar, despite the percentage of protection proposed in the simulations. Our method, therefore, goes in the direction of finding new optimization algorithms in network protection strategies  \cite{shams2014using}.

At the level of protection, it is evident that when the percentage of initial coverage is higher, the epidemic ends with a smaller fraction of people affected by COVID-19. This event is related not only to the random increase in immunization, but also to the possibility, in the model, of recognizing the bridge nodes to increase the effectiveness of vaccination. This is consistent with other investigations that indicate that the recognition of central nodes or high-risk individuals improves the efficiency of immunization strategies in real networks, a situation that favors the protection of the network and the best use of vaccine doses \cite{chen2008finding,nian2010efficient}. The best use of doses is a challenge for the current scenario of vaccine shortages worldwide, mainly in poor nations \cite{wouters2021challenges}.

On the other hand, the level of effectiveness of the DIL-W$^{\alpha}$ ranking, given the percentage of protection, is established in the recognition of the bridging nodes in a regular vaccination process. {The results using the $\alpha$ parameters of DIL-W$^{\alpha}$ indicate a high survival rate. DIL-W$^1$ achieves better results with 70\% protection and is positioned with the best survival rate, but DIL-W$^0$ and DIL-W$^{0.5}$ show good results. The difference between the different values of $\alpha$ is marginal and can be explained by the adequate representation that the DIL-W$^{\alpha}$ model has and by the values of $\alpha$, which do not generate excessive differences in the ranking. This is similar to the results of the research by Ophsal et al. who, through Freeman's EIES network, mention that the centrality degree (Definition \ref{dcent}) is relatively stable among the different $\alpha$ parameters \cite{opsahl2010node}.}

In a real and regular immunization strategy situation, such as the administration of vaccines, determining the population that infects most frequently is relevant since it allows optimization of these processes. Among our findings, it stands out that 80\% of the real infected in the Olmu\'e-Chile commune were located in 60\% of the top of the DIL-W$^{\alpha}$ ranking. Consequently, our proposal recognizes the heterogeneity of the network, approaching the reality of human interactions and achieving similar results in complex homogeneous networks \cite{nian2010efficient}.

Regarding immunization with 10\% protection, a decrease in the survival rate is established by 4\% from weeks 5 to 45 of protection. Likewise, with the same percentage of protection, the effectiveness of the immunity strategy tends to be important until week  20. After 20 weeks, the fraction of infected is similar with or without protection. Consequently, our model is strongly effective as a measure of rapid recognition of the epidemic outbreak in a given territory.

Therefore, according to the findings of our research, there are two important variables for the success and effectiveness of immunity strategies against COVID-19: (1) Recognition of bridging nodes (people with the highest probability of contagion) to apply measures of protection; and (2) the development time of this strategy.

Regarding the recognition of bridging nodes, there is evidence to support that targeted immunization schemes significantly reduce epidemic outbreaks \cite{liu2003propagation}. This opens the possibility of changing the traditional perspective of immunization by protecting a small proportion of the population over a long period of time \cite{pastor2002immunization}. It is important, therefore, not only to direct COVID-19 immunity efforts towards the population most affected by mortality, but also in those population groups that tend to infect with greater force.

The time of development of the immunization strategy continues to be a variable under discussion in the scientific community regarding the slowness worldwide of the vaccination process, which risks not achieving herd immunity \cite{aschwanden2021five}. In summary, both at a theoretical and empirical level, the execution time of immunization strategies is important in overcoming the COVID-19 pandemic.

Finally, our model helps to establish a ranking of bridge nodes in a non-homogeneous network, so it is highly replicable with real COVID-19 dissemination data and it is useful to establish more focused strategies given the reduced number of vaccines available.
\section{Conclusions}\label{conclusioncovid}
In this chapter, we evaluate the effectiveness of the DIL-W$^{\alpha}$ ranking in the immunization of nodes that are attacked by an infectious disease (COVID-19) that spreads on an edge-weighted graph obtained from the database of the Epidemiological Surveillance System of the Chilean Ministry of Health, using a graph-based SIR model.

Considering survival rates, the DIL-W$^{1}$ ranking performs better (by a small margin) than DIL-W$^{0.5}$ and DIL-W$^{0}$ rankings, subject to the protection budget being equal to 10\% of the network nodes. 

The period in which immunization or protection is given plays a key role in stopping the spread of the disease (see Figure \ref{varia_week}) since around week 25 immunization does not generate a great impact and as time progresses the survival rate decreases almost linearly.

An interesting and complex task to solve is to determine which value of $\alpha $ to choose in the network so that the ranking generated is the optimal one. The same value does not always make the performance the best. {One way to explore this is to continue with the ideas proposed in \cite{wei}, where the selection standard of the optimal turning parameters is proposed for the centrality degree, but is not for DIL-W $^{\alpha}$ ranking}. 

\chapter{General Information}
 \subsection*{Funding}
This work was supported and funded by Agencia Nacional de Investigación y Desarrollo (ANID), COVID0739 project.

\subsection*{Institutional Review Board Statement} 
Not applicable.
\subsection*{Informed Consent Statement} 
Not applicable.
\subsection*{Data Availability Statement} 
The data presented in this study were available after being requested by research project COVID-ANID to the Chilean Ministry of Health. The data are not publicly available due to legal restrictions.


\subsection*{Conflicts of Interest} 
The authors declare that they have no conflict of interest. The funders had no role in the design of the study; in the collection, analysis, or interpretation of data; in the writing of the manuscript, or in the decision to publish the results.


\appendix 
\chapter{Appendix}\label{apendi}

Survival rates ($\sigma$), according to different protection budgets in each test network.

\begin{table}[h]
\centering
\tabcolsep=0.11cm
\begin{tabular}{lccccccc}
\cline{2-8}
              & \multicolumn{7}{c}{\textbf{Zachary karate club Network}} \\ \cline{2-8} 
 &
  \multicolumn{1}{c}{\textbf{DIL-W$^0$}} &
  \multicolumn{1}{c}{\textbf{DIL-W$^{0.5}$}} &
  \multicolumn{1}{c}{\textbf{DIL-W$^1$}} &
  \multicolumn{1}{c}{\textbf{Strength}} &
  \multicolumn{1}{c}{\textbf{Closeness}} &
  \multicolumn{1}{c}{\textbf{Betweenness}} &
  \multicolumn{1}{c}{\textbf{Laplacian}} \\ \hline
\textbf{5\%}  &   0.8391& 0.8434 & 0.8354 & 0.8406	 & 0.8028 & 0.8016 & 0.8353    \\
\textbf{10\%} & 0.8945&	0.8912&	0.8941&	0.9050&	0.8210&	0.8889&	0.9015    \\
\textbf{15\%} & 0.9392&	0.9367&	0.9359&	0.9437&	0.8989&	0.9333&	0.9437    \\
\textbf{20\%} & 0.9510&	0.9525&	0.9558&	0.9568&	0.9276&	0.9497&	0.9551    \\
\textbf{25\%} &  0.9543&	0.9579&	0.9577&	0.9593&	0.9548&	0.9518&	0.9567   \\
\textbf{30\%} &  0.9578&	0.9585&	0.9572&	0.9589&	0.9580&	0.9540&	0.9597  \\
\textbf{40\%} &  0.9591&	0.9603&	0.9659&	0.9677&	0.9593&	0.9681&	0.9606   \\
\textbf{50\%} &  0.9657&	0.9655&	0.9667&	0.9695&	0.9603&	0.9705&	0.9662  \\ \hline
\end{tabular}
\caption{$\sigma$ according to different values of $k$ in Zachary karate club Network.}\label{table_zachary}
\end{table}

\begin{table}[h]
\centering
\tabcolsep=0.11cm
\begin{tabular}{lccccccc}
\cline{2-8}
              & \multicolumn{7}{c}{\textbf{Wild birds Network}} \\ \cline{2-8} 
 &
  \multicolumn{1}{c}{\textbf{DIL-W$^0$}} &
  \multicolumn{1}{c}{\textbf{DIL-W$^{0.5}$}} &
  \multicolumn{1}{c}{\textbf{DIL-W$^1$}} &
  \multicolumn{1}{c}{\textbf{Strength}} &
  \multicolumn{1}{c}{\textbf{Closeness}} &
  \multicolumn{1}{c}{\textbf{Betweenness}} &
  \multicolumn{1}{c}{\textbf{Laplacian}} \\ \hline
\textbf{5\%}  &  0.7355&	0.7394&	0.6447&	0.6188&	0.6696&	0.6411&	0.6212   \\
\textbf{10\%} &   0.7871&	0.7942&	0.7375&	0.7040&	0.7326&	0.6952&	0.7022  \\
\textbf{15\%} &  0.8721&	0.8863&	0.8760&	0.8625&	0.8149&	0.7966&	0.8654   \\
\textbf{20\%} & 0.9011&	0.9059&	0.9076&	0.8943&	0.8486&	0.8399&	0.8980   \\
\textbf{25\%} & 0.9135&	0.9155&	0.9208&	0.9130&	0.8826&	0.8676&	0.9078  \\
\textbf{30\%} &  0.9150&	0.9226&	0.9261&	0.9303&	0.8980&	0.8941&	0.9114  \\
\textbf{40\%} & 0.9361&	0.9327&	0.9481&	0.9540&	0.9018&	0.9457&	0.9272\\
\textbf{50\%} & 0.9500&	0.9626&	0.9710&	0.9762&	0.9046&	0.9766&	0.9462  \\ \hline
\end{tabular}
\caption{$\sigma$ according to different values of $k$ in Wild birds Network.}\label{table_wild}
\end{table}

\begin{table}[h]
\centering
\tabcolsep=0.11cm
\begin{tabular}{lccccccc}
\cline{2-8}
              & \multicolumn{7}{c}{\textbf{Sandy authors Network}} \\ \cline{2-8} 
 &
  \multicolumn{1}{c}{\textbf{DIL-W$^0$}} &
  \multicolumn{1}{c}{\textbf{DIL-W$^{0.5}$}} &
  \multicolumn{1}{c}{\textbf{DIL-W$^1$}} &
  \multicolumn{1}{c}{\textbf{Strength}} &
  \multicolumn{1}{c}{\textbf{Closeness}} &
  \multicolumn{1}{c}{\textbf{Betweenness}} &
  \multicolumn{1}{c}{\textbf{Laplacian}} \\ \hline
\textbf{5\%}  & 0.9301&	0.9292&	0.9296&	0.8395&	0.8817&	0.9312&	0.8329   \\
\textbf{10\%} & 0.9671&	0.9637&	0.9613&	0.9535&	0.9434&	0.9615&	0.9382  \\
\textbf{15\%} & 0.9738&	0.9727&	0.9692&	0.9706&	0.9513&	0.9686&	0.9616   \\
\textbf{20\%} &  0.9749&	0.9799&	0.9760&	0.9779&	0.9558&	0.9764&	0.9645    \\
\textbf{25\%} & 0.9771&	0.9819&	0.9822&	0.9798&	0.9568&	0.9808&	0.9773\\
\textbf{30\%} & 0.9795&	0.9837&	0.9832&	0.9827&	0.9700&	0.9829&	0.9791 \\
\textbf{40\%}& 0.9824&	0.9852&	0.9848&	0.9839&	0.9680&	0.9853&	0.9824  \\
\textbf{50\%} &0.9825&	0.9870&	0.9860&	0.9860&	0.9670&	0.9860&	0.9837   \\ \hline
\end{tabular}
\caption{$\sigma$ according to different values of $k$ in Sandy authors Network.}\label{table_sandy}
\end{table}

\begin{table}[h]
\centering
\tabcolsep=0.11cm
\begin{tabular}{lccccccc}
\cline{2-8}
              & \multicolumn{7}{c}{\textbf{CAG-mat72 Network}} \\ \cline{2-8} 
 &
  \multicolumn{1}{c}{\textbf{DIL-W$^0$}} &
  \multicolumn{1}{c}{\textbf{DIL-W$^{0.5}$}} &
  \multicolumn{1}{c}{\textbf{DIL-W$^1$}} &
  \multicolumn{1}{c}{\textbf{Strength}} &
  \multicolumn{1}{c}{\textbf{Closeness}} &
  \multicolumn{1}{c}{\textbf{Betweenness}} &
  \multicolumn{1}{c}{\textbf{Laplacian}} \\ \hline
\textbf{5\%}  &  0.7214&	0.7641&	0.7632&	0.7734& 0.7714&	0.7256&	0.7779   \\
\textbf{10\%} & 0.8947&	0.8886&	0.8792&	0.8685&	0.8451&	0.8077&	0.8706  \\
\textbf{15\%} & 0.9185&	0.9308&	0.9330&	0.9230&	0.8993&	0.9085&	0.9090 \\
\textbf{20\%} & 0.9323&	0.9452&	0.9458&	0.9446&	0.9285&	0.9094&	0.9302  \\
\textbf{25\%} & 0.9464&	0.9562&	0.9581&	0.9588&	0.9444&	0.9131&	0.9351   \\
\textbf{30\%} & 0.9578&	0.9657&	0.9672&	0.9631&	0.9646&	0.9312&	0.9426  \\
\textbf{40\%} & 0.9710&	0.9745&	0.9750&	0.9719&	0.9720&	0.9456&	0.9679\\
\textbf{50\%} & 0.9770&	0.9809&	0.9807&	0.9779&	0.9781&	0.9458&	0.9758 \\ \hline
\end{tabular}
\caption{$\sigma$ according to different values of $k$ in CAG-mat72 Network.}\label{table_cag}
\end{table}

\begin{table}[h]
\setlength{\tabcolsep}{0.11 cm}
\begin{tabular}{lccccccc}
\toprule
              & \multicolumn{7}{c}{\textbf{Scale-Free Network}} \\ \cmidrule{2-8} 
 &
  \multicolumn{1}{c}{\textbf{DIL-W\boldmath{$^0$}}} &
  \multicolumn{1}{c}{\textbf{DIL-W\boldmath{$^{0.5}$}}} &
  \multicolumn{1}{c}{\textbf{DIL-W\boldmath{$^1$}}} &
  \multicolumn{1}{c}{\textbf{Strength}} &
  \multicolumn{1}{c}{\textbf{Closeness}} &
  \multicolumn{1}{c}{\textbf{Betweenness}} &
  \multicolumn{1}{c}{\textbf{Laplacian}} \\ \midrule
\textbf{5\%}  &     0.2235  &   0.2359   &  0.2773 &    0.2845  &   0.1748 &    0.2320  &   0.2312\\
\textbf{10\%}  &     0.8610  &   0.5965   &  0.6893  &   0.8593   &  0.3901  &   0.6857  &   0.6264\\
\textbf{15\%}  &     0.9437   &  0.7585   &  0.9528   &  0.9168  &   0.6622   &  0.8734   &  0.9155\\
\textbf{20\%}  &     0.9763   &  0.9298   &  0.9636  &   0.9609  &   0.7431  &   0.9092   &  0.9174\\
\textbf{25\%}  &     0.9804   &  0.9442   &  0.9803   &  0.9684   &  0.7635   &  0.9637   &  0.9497\\
\textbf{30\%}  &     0.9853   &  0.9710   &  0.9852   &  0.9837  &   0.8509  &   0.9778  &   0.9519\\
\textbf{40\%}  &     0.9882   &  0.9755   &  0.9875   &  0.9857  &  0.9268  &   0.9863   &  0.9671\\
\textbf{50\%}  &     0.9897   &  0.9845 &    0.9887   &  0.9875  &   0.9359   &  0.9884  &   0.9830\\ \bottomrule
\end{tabular}
\caption{$\sigma$ according to different values of $k$ in a Scale-free~Network.}\label{tab_free}
\end{table}

\begin{table}[h]
\setlength{\tabcolsep}{0.11cm}
\begin{tabular}{lccccccc}
\toprule
              & \multicolumn{7}{c}{\textbf{Scale-Free Network}} \\ \cmidrule{2-8} 
 &
  \multicolumn{1}{c}{\textbf{DIL-W\boldmath{$^0$}}} &
  \multicolumn{1}{c}{\textbf{DIL-W\boldmath{$^{0.5}$}}} &
  \multicolumn{1}{c}{\textbf{DIL-W\boldmath{$^1$}}} &
  \multicolumn{1}{c}{\textbf{Strength}} &
  \multicolumn{1}{c}{\textbf{Closeness}} &
  \multicolumn{1}{c}{\textbf{Betweenness}} &
  \multicolumn{1}{c}{\textbf{Laplacian}} \\ \midrule
\textbf{2}  & 0.9841 &    0.9667  &  0.9842  &  0.9838  &  0.9284  &  0.9847   & 0.9800\\
\textbf{5}  &    0.8411 &   0.6032  &  0.8410  &  0.8373   & 0.4783  &  0.8046  &  0.7692\\
\textbf{10}  &    0.1143 &   0.1170  &  0.1311  &  0.1118  &  0.1104  &  0.1193  &  0.1107\\
\textbf{30}  &    0.1032   & 0.1043  &  0.1059  &  0.1037  &  0.1046  &  0.1041  &  0.1061\\
\textbf{50}  &    0.1041  &  0.1045 &   0.1046  &  0.1050  &  0.1053  &  0.1048  &  0.1043\\ \bottomrule
\end{tabular}
\caption{$\sigma$ according to different values of $d$ in a Scale-free~Network.}\label{tab_free2}
\end{table}


\newpage

\addcontentsline{toc}{chapter}{Blibliography}

\begin{thebibliography}{200}
\bibitem{nature}
Epidemiology is a science of high importance.
\newblock {\em Nat. Commun.} {\bf 2018}, {\em 9},~1703.
\newblock
  doi:{ \href{https://doi.org/10.1038/s41467-018-04243-3}{\detokenize{10.1038/s41467-018-04243-3}}}.

\bibitem{lloyd2007network}
Lloyd, A.L.; Valeika, S.
\newblock Network models in epidemiology: An overview.
\newblock In {\em Complex Population Dynamics: Nonlinear Modeling in Ecology,
  Epidemiology and Genetics}; World Scientific; London; UK;  {2007}; pp. 189--214. 
\bibitem{haslbeck2018well}
Haslbeck, J.M.; Waldorp, L.J.
\newblock How well do network models predict observations? On the importance of
  predictability in network models.
\newblock {\em Behav. Res. Methods} {\bf 2018}, {\em 50},~853--861.

\bibitem{an2014synchronization}
An, X.L.; Zhang, L.; Li, Y.Z.; Zhang, J.G.
\newblock Synchronization analysis of complex networks with multi-weights and
  its application in public traffic network.
\newblock {\em Phys. A Stat. Mech. Its Appl.} {\bf
  2014}, {\em 412},~149--156.

\bibitem{montenegro2019linear}
Montenegro, E.; Cabrera, E.; Gonz{\'a}lez, J.; Manr\'iquez, R.
\newblock Linear representation of a graph.
\newblock {\em Bol. Soc. Parana. Matem{\'A}Tica} {\bf 2019},
  {\em 37},~97--102.

\bibitem{guangzeng2009novel}
Wang, G.; Cao, Y.; Bao, Z.Y.; Han, Z.X.
\newblock A novel local-world evolving network model for power grid.
\newblock {\em Acta Phys. Sin.} {\bf 2009}, {\em 6},~58.

\bibitem{mersch2013tracking}
Mersch, D.P.; Crespi, A.; Keller, L.
\newblock Tracking individuals shows spatial fidelity is a key regulator of ant
  social organization.
\newblock {\em Science} {\bf 2013}, {\em 340},~1090--1093.

\bibitem{firth2015experimental}
Firth, J.A.; Sheldon, B.C.
\newblock Experimental manipulation of avian social structure reveals
  segregation is carried over across contexts.
\newblock {\em Proc. R. Soc. B Biol. Sci.} {\bf
  2015}, {\em 282},~20142350.

\bibitem{wu2008community}
Wu, X.; Liu, Z.
\newblock How community structure influences epidemic spread in social
  networks.
\newblock {\em Phys. A Stat. Mech. Its Appl.} {\bf
  2008}, {\em 387},~623--630.

\bibitem{magelinski2021measuring}
Magelinski, T.; Bartulovic, M.; Carley, K.M.
\newblock Measuring Node Contribution to Community Structure With Modularity
  Vitality.
\newblock {\em IEEE Trans. Netw. Sci. Eng.} {\bf
  2021}, {\em 8},~707--723.

\bibitem{ghalmane2019centrality}
Ghalmane, Z.; Cherifi, C.; Cherifi, H.; El~Hassouni, M.
\newblock Centrality in complex networks with overlapping community structure.
\newblock {\em Sci. Rep.} {\bf 2019}, {\em 9},~1--29.

\bibitem{guerrero-manriquez_proyeccion_2020}
Guerrero-Nancuante, C.; Manr\'iquez~P, R.
\newblock An epidemiological forecast of COVID-19 in Chile based on the
  generalized SEIR model and the concept of recovered.
\newblock {\em Medwave} {\bf 2020}, {\em 20}.

\bibitem{ronald1}
Manr\'iquez, R.; Guerrero-Nancuante, C.; Mart\'inez, F.; Taramasco, C.
\newblock Spread of Epidemic Disease on Edge-Weighted Graphs from a Database: A
  Case Study of COVID-19.
\newblock {\em Int. J. Environ. Res. Public Health} {\bf 2021}, {\em 18}.
\newblock
  doi:{ \href{https://doi.org/10.3390/ijerph18094432}{\detokenize{10.3390/ijerph18094432}}}.

\bibitem{ronaldranking}
Manr\'iquez, R.; Guerrero-Nancuante, C.; Mart\'inez, F.; Taramasco, C.
\newblock A generalization of the importance of vertices for an undirected
  weighted graph.
\newblock {\em Symmetry} {\bf 2021}, {\em 13}, 902.
\newblock
  doi:{ \href{https://doi.org/10.3390/sym13050902}{\detokenize{10.3390/sym13050902}}}.
  
\bibitem{tang2020research}
Tang, P.; Song, C.; Ding, W.; Ma, J.; Dong, J.; Huang, L.
\newblock Research on the node importance of a weighted network based on the
  k-order propagation number algorithm.
\newblock {\em Entropy} {\bf 2020}, {\em 22},~364.

\bibitem{wang2017drimux}
Wang, B.; Chen, G.; Fu, L.; Song, L.; Wang, X.
\newblock Drimux: Dynamic rumor influence minimization with user experience in
  social networks.
\newblock {\em IEEE Trans. Knowl. Data Eng.} {\bf 2017},
  {\em 29},~2168--2181.

\bibitem{zhang2014dava}
Zhang, Y.; Prakash, B.A.
\newblock Dava: Distributing vaccines over networks under prior information.
\newblock In \emph{Proceedings of the 2014 SIAM International Conference on Data
  Mining}; SIAM, Philadelphia, PA, USA, 2014; pp. 46--54. 
\bibitem{hebert2013global}
H{\'e}bert-Dufresne, L.; Allard, A.; Young, J.G.; Dub{\'e}, L.J.
\newblock Global efficiency of local immunization on complex networks.
\newblock {\em Sci. Rep.} {\bf 2013}, {\em 3},~1--8.

\bibitem{wijayanto2019effective}
Wijayanto, A.W.; Murata, T.
\newblock Effective and scalable methods for graph protection strategies
  against epidemics on dynamic networks.
\newblock {\em Appl. Netw. Sci.} {\bf 2019}, {\em 4},~1--31.

\bibitem{tong2010}
Tong, H.; Prakash, B.A.; Tsourakakis, C.; Eliassi-Rad, T.; Faloutsos, C.; Chau,
  D.
\newblock On the Vulnerability of Large Graphs.
\newblock  In {Proceedings of the }2010 IEEE International Conference on Data Mining, Sydney, NSW, Australia, 13--17 December 2010; pp.
  1091--1096.
\newblock
  doi:{ \href{https://doi.org/10.1109/ICDM.2010.54}{\detokenize{10.1109/ICDM.2010.54}}}.

\bibitem{nian2018immunization}
Nian, F.; Hu, C.; Yao, S.; Wang, L.; Wang, X.
\newblock An immunization based on node activity.
\newblock {\em Chaos Solitons Fractals} {\bf 2018}, {\em 107},~228--233.

\bibitem{gomez2006immunization}
G{\'o}mez-Gardenes, J.; Echenique, P.; Moreno, Y.
\newblock Immunization of real complex communication networks.
\newblock {\em  Eur. Phys. J. Condens. Matter Complex Syst.} {\bf 2006}, {\em 49},~259--264.

\bibitem{cohen2003efficient}
Cohen, R.; Havlin, S.and Ben-Avraham, D.
\newblock Efficient immunization strategies for computer networks and
  populations.
\newblock {\em Phys. Rev. Lett.} {\bf 2003}, {\em 91},~247901.

\bibitem{xia2018improved}
Xia, L.; Song, Y.R.; Li, C.C.; Jiang, G.P.
\newblock Improved targeted immunization strategies based on two rounds of
  selection.
\newblock {\em Phys. A  Stat. Mech. Its Appl.} {\bf
  2018}, {\em 496},~540--547.

\bibitem{wang2009imperfect}
Wang, Y.; Xiao, G.; Hu, J.; Cheng, T.; Wang, L.
\newblock Imperfect targeted immunization in scale-free networks.
\newblock {\em Phys. A Stat. Mech. Its Appl.} {\bf
  2009}, {\em 388},~2535--2546.

\bibitem{ghalmane2019immunization}
Ghalmane, Z.; El~Hassouni, M.; Cherifi, H.
\newblock Immunization of networks with non-overlapping community structure.
\newblock {\em Soc. Netw. Anal. Min.} {\bf 2019}, {\em 9},~1--22.

\bibitem{gupta}
Gupta, N.; Singh, A.; Cherifi, H.
\newblock Community-based immunization strategies for epidemic control.
\newblock In {Proceedings of the } 2015 7th international conference on communication systems and
  networks (COMSNETS),  Bangalore, India,  6--10 January 2015; pp. 1--6.

\bibitem{ronald2}
Manr\'iquez, R.; Guerrero-Nancuante, C.; Taramasco, C.
\newblock Protection Strategy for Edge-Weighted Graphs in Disease Spread.
\newblock {\em Appl. Sci.} {\bf 2021}, {\em 11}, 5115.
\newblock
  doi:{ \href{https://doi.org/10.3390/app11115115}{\detokenize{10.3390/app11115115}}}.


\bibitem{ronald3}
Manr\'iquez, R.; Guerrero-Nancuante, C.; Taramasco, C.
\newblock Protection Strategy against an Epidemic Disease on Edge-Weighted Graphs Applied to a Covid-19 Case.
\newblock {\em Biology.} {\bf 2021}, {\em 10}, 667.
\newblock
  doi:{ \href{https://doi.org/10.3390/biology10070667}{\detokenize{10.3390/biology10070667}}}.


\bibitem{zhao2014immunization}
Zhao, D.; Wang, L.; Li, S.; Wang, Z.; Wang, L.; Gao, B.
\newblock Immunization of epidemics in multiplex networks.
\newblock {\em PLoS One} {\bf 2014}, {\em 9},~e112018.

\bibitem{chartrand_graphs_1996}
Chartrand, G.; Lesniak, L.
\newblock {\em Graphs and Digraphs}, 1st ed.; CRC Press: London, UK,  1996.

\bibitem{west_introduction_2001}
West, D.B.
\newblock {\em Introduction to Graph Theory}, 2nd ed.; Prentice Hall: Englewood Cliffs, NJ, USA,  2001.

\bibitem{opsahl2010node}
Opsahl, T.; Agneessens, F.; Skvoretz, J.
\newblock Node centrality in weighted networks: Generalizing degree and
  shortest paths.
\newblock {\em Soc. Netw.} {\bf 2010}, {\em 32},~245--251.

\bibitem{liu2016evaluating}
Liu, J.; Xiong, Q.; Shi, W.; Shi, X.; Wang, K.
\newblock Evaluating the importance of nodes in complex networks.
\newblock {\em Phys. A  Stat. Mech. Its Appl.} {\bf
  2016}, {\em 452},~209--219.

\bibitem{1318579}
{Vinterbo}, S.A.
\newblock Privacy: A machine learning view.
\newblock {\em IEEE Trans. Knowl. Data Eng.} {\bf 2004},
  {\em 16},~939--948.
\newblock
  doi:{ \href{https://doi.org/10.1109/TKDE.2004.31}{\detokenize{10.1109/TKDE.2004.31}}}.

\bibitem{PhysRevLett}
Latora, V.; Marchiori, M.
\newblock Efficient Behavior of Small-World Networks.
\newblock {\em Phys. Rev. Lett.} {\bf 2001}, {\em 87},~198701. \linebreak 
\newblock
  doi:{\href{https://doi.org/10.1103/PhysRevLett.87.198701}{\detokenize{10.1103/PhysRevLett.87.198701}}}.

\bibitem{ren2013node}
Ren, Z.; Shao, F.; Liu, J.; Guo, Q.; Wang, B.H.
\newblock Node importance measurement based on the degree and clustering
  coefficient information.
\newblock {\em Acta Phys. Sin.} {\bf 2013}, {\em 62},~128901.

\bibitem{musial}
Musia{\l}, K.; Juszczyszyn, K.
\newblock Properties of bridge nodes in social networks.
\newblock In \emph{International Conference on Computational Collective Intelligence};
  Springer: Berlin/Heidelberg, Germany, 2009; pp. 357--364.

\bibitem{shams2014using}
Shams, B.
\newblock Using network properties to evaluate targeted immunization
  algorithms.
\newblock {\em Netw. Biol.} {\bf 2014}, {\em 4},~74.

\bibitem{chen2008finding}
Chen, Y.; Paul, G.; Havlin, S.; Liljeros, F.; Stanley, H.
\newblock Finding a better immunization strategy.
\newblock {\em Phys. Rev. Lett.} {\bf 2008}, {\em 101},~058701.

\bibitem{nian2010efficient}
Nian, F.; Wang, X.
\newblock Efficient immunization strategies on complex networks.
\newblock {\em J. Theor. Biol.} {\bf 2010}, {\em 264},~77--83.

\bibitem{wouters2021challenges}
Wouters, O.J.; Shadlen, K.; Salcher-Konrad, M.; Pollard, A.J.; Larson, H.;
  Teerawattananon, Y.; Jit, M.
\newblock Challenges in ensuring global access to COVID-19 vaccines:
  production, affordability, allocation, and deployment.
\newblock {\em  Lancet} {\bf 2021}.
\newblock
  doi:{ \href{https://doi.org/10.1016/S0140-6736(21)00306-8}{\detokenize{S0140-6736(21)00306-8}}}.

\bibitem{liu2003propagation}
Liu, Z.; Lai, Y.C.; Ye, N.
\newblock Propagation and immunization of infection on general networks with
  both homogeneous and heterogeneous components.
\newblock {\em Phys. Rev. E} {\bf 2003}, {\em 67},~031911.

\bibitem{pastor2002immunization}
Pastor-Satorras, R.; Vespignani, A.
\newblock Immunization of complex networks.
\newblock {\em Phys. Rev. E} {\bf 2002}, {\em 65},~036104.

\bibitem{aschwanden2021five}
Aschwanden, C.
\newblock Five reasons why COVID herd immunity is probably impossible.
\newblock {\em Nature} {\bf 2021}, {\em 591},~520--522.

\bibitem{wei}
Wei, D; Zhang, Y. L.; Deng, Y.  
\newblock Degree centrality based on the weighted network.
\newblock In {Proceedings of the }{24th Chinese Control and Decision Conference (CCDC)},  Taiyuan, China,  23--25 May {2012}; pp. 3976--3979.



\bibitem{1}
Hays, J.N. 
\newblock {\em Epidemics and Pandemics: their impacts on human history}, Ed. ABC-Clio; 2005. 513 pp.

\bibitem{2}
Anderson, R.M.; May, R.
\newblock {\em Infectious diseases of humans: Dynamics and control}, Ed OUP Oxford; 1992. 757 pp.

\bibitem{3}
Brauer, F.
\newblock {\em Compartmental models in Epidemiology. In: Notes in Mathematical Epidemiology}, Ed Springer Science \& Business Media; 2008. 414 pp.

\bibitem{4}
Pitlik, S.D
\newblock COVID-19 Compared to Other Pandemic Diseases.
\newblock {\em Rambam Maimonides Med J} {\bf 2020}, {\em 11},~e0027
\newblock
  doi:{ \href{https://doi.org/10.5041/RMMJ.10418}{\detokenize{10.5041/RMMJ.10418}}}.
  
\bibitem{bhapkar_virus_2020}
Bhapkar, H.R.; Mahalle, P.N.; Dhotre, P.S.
\newblock Virus {Graph} and {COVID}-19 {Pandemic}: {A} {Graph} {Theory}
  {Approach}. In {\em Big {Data} {Analytics} and {Artificial} {Intelligence}
  {Against} {COVID}-19: {Innovation} {Vision} and {Approach}}; Hassanien, A.E.;
  Dey, N.; Elghamrawy, S., Eds.; Studies in {Big} {Data}, Springer
  International Publishing,  2020; pp. 15--34.
\newblock
  doi:{ \href{https://doi.org/10.1007/978-3-030-55258-9_2}{\detokenize{10.1007/978-3-030-55258-9_2}}}.

\bibitem{CROCCOLO2020110077}
Croccolo, F.; Roman, E.
\newblock Spreading of infections on random graphs: A percolation-type model
  for COVID-19.
\newblock {\em Chaos, Solitons \& Fractals} {\bf 2020}, {\em 139},~110077.  

\bibitem{article}
Singhal, T.
\newblock A Review of Coronavirus Disease-2019 (COVID-19).
\newblock {\em The Indian Journal of Pediatrics} {\bf 2020}, {\em 87}.
\newblock
  doi:{ \href{https://doi.org/10.1007/s12098-020-03263-6}{\detokenize{10.1007/s12098-020-03263-6}}}.


\bibitem{8} 
Nicola, M.; Alsafi, Z.; Sohrabi C.; Kerwan, A.; Al-Jabir, A.; Iosifidis, C.; Agha, M.; Agha, R.
\newblock The socio-economic implications of the coronavirus pandemic (COVID-19): A review.
\newblock {\em International Journal of Surgery} {\bf 2020},
\newblock
doi:{ \href{https://doi.org/10.1016/j.ijsu.2020.04.018}{\detokenize{10.1016/j.ijsu.2020.04.018}}}.
 
\bibitem{Kermack_1927}
Kermack, W.; McKendrick, A.G.
\newblock A contribution to the mathematical theory of epidemics.
\newblock {\em Proc. of the Royal Society of London A,} {\bf 1927}, {\em
  115},~700--721,
\newblock
  doi:{ \href{https://doi.org/10.1098/rspa.1927.0118}{\detokenize{10.1098/rspa.1927.0118}}}.


\bibitem{11}
Montagnon, P. A
\newblock A Stochastic SIR model on a graph with epidemiological and population dynamics occurring over the same time scale.
\newblock {\em J. Math. Biol.} {\bf 2019}, {\em 79},~31--62.
\newblock
  doi:\href{https://doi.org/10.1007/s00285-019-01349-0}{\detokenize{https://doi.org/10.1007/s00285-019-01349-0}}.

\bibitem{12}
Enright, J.;  Kao, T. R.
\newblock Epidemics on dynamic networks,
\newblock {\em Epidemics} {\bf 2018}, {\em 24},~88--97.
\newblock
  doi:{ \href{https://doi.org/10.1016/j.epidem.2018.04.003}{\detokenize{https://doi.org/10.1016/j.epidem.2018.04.003}}}.


\bibitem{cardinal-fernandez_medicina_2014}
Cardinal-Fern\'andez, P.; Nin, N.; Ru\'iz-Cabello, J.; Lorente, J.
\newblock Medicina de sistemas: una nueva visi\'on de la pr\'actica cl\'inica.
\newblock {\em Archivos de Bronconeumolog\'ia} {\bf 2014}, {\em 50},~444--451.
\newblock
  doi:{ \href{https://doi.org/10.1016/j.arbres.2013.10.010}{\detokenize{10.1016/j.arbres.2013.10.010}}}.

\bibitem{RIZZO2016212}
Rizzo, A.; Pedalino, B.; Porfiri, M.
\newblock A network model for Ebola spreading.
\newblock {\em Journal of Theoretical Biology} {\bf 2016}, {\em 394},~212--222.
\newblock
  doi:{ \href{https://doi.org/10.1016/j.jtbi.2016.01.015}{\detokenize{10.1016/j.jtbi.2016.01.015}}}.

\bibitem{shafer}
Shafer, L.; Adegboye, O.; Elfaki, F.
\newblock Network Analysis of MERS Coronavirus within Households, Communities,
  and Hospitals to Identify Most Centralized and Super-Spreading in the Arabian
  Peninsula, 2012 to 2016.
\newblock {\em Canadian Journal of Infectious Diseases and Medical
  Microbiology} {\bf 2018}, {\em 2018},~6725284.
\newblock
  doi:{ \href{https://doi.org/10.1155/2018/6725284}{\detokenize{10.1155/2018/6725284}}}.
  


\bibitem{17}
Bi, K.; Chen, Y.; Zhao, S., Ben-Arieh, D.; Wu, C. H. J.
\newblock Modeling learning and forgetting processes with the corresponding impacts on human behaviors in infectious disease epidemics,
\newblock {\em Computers \& Industrial Engineering} {\bf 2019}, {\em 129},~563--577.
\newblock
  doi:{ \href{https://doi.org/10.1016/j.cie.2018.04.035}{\detokenize{https://doi.org/10.1016/j.cie.2018.04.035}}}.
  

\bibitem{margevicius_advancing_2014}
Margevicius, K.J.; Generous, N.; Taylor-McCabe, K.; Brown, M.; Daniel, W.B.;
  et~al., C.L.
\newblock Advancing a {Framework} to {Enable} {Characterization} and
  {Evaluation} of {Data} {Streams} {Useful} for {Biosurveillance}.
\newblock {\em PLOS ONE} {\bf 2014}, {\em 9},~e83730.
\newblock
  doi:{ \href{https://doi.org/10.1371/journal.pone.0083730}{\detokenize{10.1371/journal.pone.0083730}}}.


\bibitem{p_erdos_random_1959}
P., E.; A., R.
\newblock On {Random} {Graphs}.
\newblock {\em Publicationes Mathematicae (Debrecen)} {\bf 1959}, {\em
  6},~290--297.

\bibitem {barabasi_emergence_1999}
Barab\'asi, A.; Albert, R.
\newblock Emergence of {Scaling} in {Random} {Networks}.
\newblock {\em Science} {\bf 1999}, {\em 286},~509--512.
\newblock
  doi:{ \href{https://doi.org/10.1126/science.286.5439.509}{\detokenize{10.1126/science.286.5439.509}}}.

\bibitem{anderson_discussion_1991}
Anderson, R.
\newblock Discussion: {The} {Kermack}-{McKendrick} epidemic threshold theorem.
\newblock {\em Bulletin of Mathematical Biology} {\bf 1991}, {\em 53},~1--32.
\newblock
  doi:{\href{https://doi.org/10.1007/BF02464422}{\detokenize{10.1007/BF02464422}}}.

\bibitem{ross_sir_2013}
Ross, R.; Hamer, W.
\newblock The {SIR} model and the {Foundations} of {Public} {Health}.
\newblock  2013.

\bibitem{noauthor_theory_2005}
Kari, J.
\newblock Theory of cellular automata: {A} survey.
\newblock {\em Theoretical Computer Science} {\bf 2005}, {\em 334},~3--33.
\newblock
  doi:{ \href{https://doi.org/10.1016/j.tcs.2004.11.021}{\detokenize{10.1016/j.tcs.2004.11.021}}}.
  
\bibitem{Zhang_2015}
Zhang, Z.; Wang, H.; Wang, C.; Fang, H.
\newblock Modeling Epidemics Spreading on Social Contact Networks.
\newblock {\em IEEE transactions on emerging topics in computing} {\bf 2015}, {\em 3},~410-419.
\newblock
  doi:{ \href{https://doi.org/10.1109/TETC.2015.2398353}{\detokenize{10.1109/TETC.2015.2398353}}}.
  
\bibitem{pastor-satorras_epidemic_2015}
Pastor-Satorras, R., et al.
\newblock Epidemic {processes} in complex {networks}.
\newblock {\em Reviews of modern physics} {\bf 2015}, {\em 87},~925.

\bibitem{pastor-satorras_epidemic_2001}
Pastor-Satorras, R.; Vespignani, A.
\newblock Epidemic {Spreading} in {Scale}-{Free} {Networks}.
\newblock {\em Physical Review Letters} {\bf 2001}, {\em 86},~3200--3203.
\newblock
  doi:{ \href{https://doi.org/10.1103/PhysRevLett.86.3200}{\detokenize{10.1103/PhysRevLett.86.3200}}}.
  
\bibitem{2_2}
Keeling, M.J.; Eames, K.T.
\newblock Networks and epidemic models.
\newblock {\em Journal of the Royal Society Interface} {\bf 2005}, {\em 2},~295--307.  
 doi:{\href{https://doi.org/10.1098/rsif.2005.0051}{\detokenize{10.1098/rsif.2005.0051}}}. 
 
 \bibitem{4_4}
Tang, B.; Wang, X.; Li, Q.; Bragazzi, N. L.; Tang, S.; Xiao, Y.; Wu, J.
\newblock Estimation of the transmission risk of the 2019-nCoV and its implication for public health interventions.
\newblock {\em Journal of clinical medicine} {\bf 2005}, {\em 9},~462.  
 doi:{\href{https://doi.org/10.3390/jcm9020462}{\detokenize{10.3390/jcm9020462}}}. 


\bibitem{5_5}
Soucie, J.M.
\newblock Public health surveillance and data collection: general principles and impact on hemophilia care.
\newblock {\em Hematology} {\bf 2012}, {\em 17},~s144--s146.  
 doi:{\href{https://doi.org/10.1179/102453312X13336169156537}{\detokenize{10.1179/102453312X13336169156537}}}. 

\bibitem{6_6}
Dong, E.; Du, H., Gardner, L.
\newblock An interactive web-based dashboard to track COVID-19 in real time.
\newblock {\em The Lancet infectious diseases} {\bf 2020}, {\em 20},~533--534.  
 doi:{\href{https://doi.org/10.1016/S1473-3099(20)30120-1}{\detokenize{10.1016/S1473-3099(20)30120-1}}}.   

\bibitem{7_7}
Stern, A. M.; Markel, H.
\newblock International efforts to control infectious diseases, 1851 to the present.
\newblock {\em JAMA} {\bf 2004}, {\em 292},~1474--1479.  
 doi:{\href{https://doi.org/10.1001/jama.292.12.1474}{\detokenize{10.1001/jama.292.12.1474}}}.   


\bibitem{newman2003structure}
M.~Newman.
\newblock The structure and function of complex networks.
\newblock {\em SIAM review}, 45(2):167--256, 2003.

\bibitem{an2014synchronization}
X-L. An, L.~Zhang, Y-Z. Li, and J.-G. Zhang.
\newblock Synchronization analysis of complex networks with multi-weights and
  its application in public traffic network.
\newblock {\em Physica A: Statistical Mechanics and its Applications},
  412:149--156, 2014.

\bibitem{saxena}
C. Saxena, M. N. Doja, and T. Ahmad. 
\newblock Group based centrality for immunization of complex networks. 
\newblock {\em Physica A: Statistical Mechanics and its Applications}, 508, 35-47, 2018.

\bibitem{wang2010new}
J-W. Wang, L.L Rong, and T-Z. Guo.
\newblock A new measure method of network node importance based on local
  characteristics.
\newblock {\em Journal of Dalian University of Technology}, 50(5):822--826,
  2010.

\bibitem{lu2016vital}
L.~L{\"u}, D.~Chen, X-L. Ren, Q-M. Zhang, Y-C Zhang, and T.~Zhou.
\newblock Vital nodes identification in complex networks.
\newblock {\em Physics Reports}, 650:1--63, 2016.

\bibitem{guess2019less}
Guess, A.; Nagler, J.; Tucker, J.
\newblock Less than you think: Prevalence and predictors of fake news
  dissemination on Facebook.
\newblock {\em Sci. Adv.} {\bf 2019}, {\em 5},~eaau4586.

\bibitem{yang2012towards}
Yang, X.; Yang, L.
\newblock Towards the epidemiological modeling of computer viruses.
\newblock {\em Discret. Dyn. Nat. Soc.} {\bf 2012}, {\em 2012}.

\bibitem{almasi2019measuring}
Almasi, S.; Hu, T.
\newblock Measuring the importance of vertices in the weighted human disease
  network.
\newblock {\em PLoS ONE} {\bf 2019}, {\em 14},~e0205936.

\bibitem{Crossley11583}
Crossley, N.; Mechelli, A.; V{\'e}rtes, P.; Winton-Brown, T.; Patel, A.;
  Ginestet, C.; McGuire, P.; Bullmore, E.
\newblock Cognitive relevance of the community structure of the human brain
  functional coactivation network.
\newblock {\em Proc. Natl. Acad. Sci. USA} {\bf 2013},
  {\em 110},~11583--11588.

\bibitem{xu2019node}
H.~Xu, J.~Zhang, J.~Yang, and L.~Lun.
\newblock Node importance ranking of complex network based on degree and
  network density.
\newblock {\em International Journal of Performability Engineering}, 15(3),
  2019.

\bibitem{sciarra}
C. Sciarra, G. Chiarotti, F. Laio, and L. Ridolfi. 
\newblock A change of perspective in network centrality. 
\newblock {\em Scientific reports}, 8(1), 1-9, 2018.

\bibitem{magelinski}
T. Magelinski, M. Bartulovic and K.M Carley. 
\newblock Measuring Node Contribution to Community Structure with Modularity Vitality. 
\newblock {\em IEEE Transactions on Network Science and Engineering}, 2021.


\bibitem{zachary}
W.~W. Zachary.
\newblock An information flow model for conflict and fission in small groups.
\newblock {\em Journal of Anthropological Research}, 33(4):452--473, 1977.

\bibitem{wei2015weighted283}
B.~Wei, J.~Liu, D.~Wei, C.~Gao, and Y.~Deng.
\newblock Weighted k-shell decomposition for complex networks based on
  potential edge weights.
\newblock {\em Physica A: Statistical Mechanics and its Applications},
  420:277--283, 2015.

\bibitem{ghalmane}
Z. Ghalmane,  M. El Hassouni, and H. Cherifi. 
\newblock Immunization of networks with non-overlapping community structure. 
\newblock {\em Social Network Analysis and Mining}, 9(1), 1-22, 2019.

\bibitem{ghalmane2}
Z. Ghalmane,  C. Cherifi, H. Cherifi, and M. El Hassouni.  
\newblock Centrality in complex networks with overlapping community structure. 
\newblock {\em Scientific reports}, 9(1), 1-29, 2019.

\bibitem{yang2019novel}
Y.~Yang, L.~Yu, X.~Wang, Z.~Zhou, Y.~Chen, and T.~Kou.
\newblock A novel method to evaluate node importance in complex networks.
\newblock {\em Physica A: Statistical Mechanics and its Applications},
  526:121118, 2019.

\bibitem{jia2011improved}
W.~Jia-sheng, W.~Xiao-ping, Y.~Bo, and G.~Jiang-wei.
\newblock Improved method of node importance evaluation based on node
  contraction in complex networks.
\newblock {\em Procedia Engineering}, 15:1600--1604, 2011.

\bibitem{mo2019identifying}
H.~Mo and Y.~Deng.
\newblock Identifying node importance based on evidence theory in complex
  networks.
\newblock {\em Physica A: Statistical Mechanics and its Applications},
  529:121538, 2019.

\bibitem{barrat2004architecture}
A.~Barrat, M.~Barthelemy, R.~Pastor-Satorras, and A.~Vespignani.
\newblock The architecture of complex weighted networks.
\newblock {\em Proceedings of the national academy of sciences},
  101(11):3747--3752, 2004.

\bibitem{lu2016h}
L.~L{\"u}, T.~Zhou, Q-M. Zhang, and H.~E. Stanley.
\newblock The h-index of a network node and its relation to degree and
  coreness.
\newblock {\em Nature communications}, 7(1):1--7, 2016.

\bibitem{garas2012k281}
A.~Garas, Fr. Schweitzer, and S.~Havlin.
\newblock A k-shell decomposition method for weighted networks.
\newblock {\em New Journal of Physics}, 14(8):083030, 2012.

\bibitem{brandes2001faster}
U.~Brandes.
\newblock A faster algorithm for betweenness centrality.
\newblock {\em Journal of mathematical sociology}, 25(2):163--177, 2001.

\bibitem{newman2001scientific}
M.~Newman.
\newblock Scientific collaboration networks. ii. shortest paths, weighted
  networks, and centrality.
\newblock {\em Physical review E}, 64(1):016132, 2001.

\bibitem{PhysRevE.75.021102}
Q.~Ou, Y.D. Jin, T.~Zhou, B-H. Wang, and B-Q. Yin.
\newblock Power-law strength-degree correlation from resource-allocation
  dynamics on weighted networks.
\newblock {\em Phys. Rev. E}, 75:021102, Feb 2007.

\bibitem{ahmad}
A. Ahmad, T. Ahmad, and A. Bhatt.
\newblock HWSMCB: A community-based hybrid approach for identifying influential nodes in the social network. 
\newblock {\em Physica A: Statistical Mechanics and its Applications},
 545, 123590, 2020.

\bibitem{qi2012laplacian}
X.~Qi, E.~Fuller, Q.~Wu, Y.~Wu, and C-Q. Zhang.
\newblock Laplacian centrality: A new centrality measure for weighted networks.
\newblock {\em Information Sciences}, 194:240--253, 2012.

\bibitem{SKIBSKI2019151}
O.~Skibski, T.~Rahwan, T.~Michalak, and M.~Yokoo.
\newblock Attachment centrality: Measure for connectivity in networks.
\newblock {\em Artificial Intelligence}, 274:151 -- 179, 2019.

\bibitem{ijcai2017-59}
J.~Sosnowska and O.~Skibski.
\newblock Attachment centrality for weighted graphs.
\newblock In {\em Proceedings of the Twenty-Sixth International Joint
  Conference on Artificial Intelligence, {IJCAI-17}}, pages 416--422, 2017.

\bibitem{yang2017mining}
Y.~Yang, G.~Xie, and J.~Xie.
\newblock Mining important nodes in directed weighted complex networks.
\newblock {\em Discrete Dynamics in Nature and Society}, 2017, 2017.

\bibitem{nr-aaai15}
R.A. Rossi and N.~K. Ahmed.
\newblock The network data repository with interactive graph analytics and
  visualization.
\newblock In {\em Proceedings of the Twenty-Ninth AAAI Conference on Artificial
  Intelligence}, 2015.

\bibitem{Musial}
Musial. K., Juszczyszyn, K.  
\newblock Properties of bridge nodes in social networks. 
\newblock {\em In International Conference on Computational Collective Intelligence },
357--364. Springer, Berlin, Heidelberg, 2019.

\bibitem{colizza2007}
V.~Colizza, R.~Pastor-Satorras, and A.~Vespignani.
\newblock Reaction--diffusion processes and metapopulation models in
  heterogeneous networks.
\newblock {\em Nature Physics}, 3(4):276--282, 2007.

\bibitem{vitonicoruso}
V.~Latora, V.~Nicosia, and G.~Russo.
\newblock {\em Complex Networks: Principles, Methods and Applications}.
\newblock Cambridge University Press, 1 edition, 2017.

\bibitem{RUBINOV20101059}
M.~Rubinov and O.~Sporns.
\newblock Complex network measures of brain connectivity: Uses and
  interpretations.
\newblock {\em NeuroImage}, 52(3):1059 -- 1069, 2010.

\bibitem{lai2004}
Y.C. Lai, A.~Motter, and T.~Nishikawa.
\newblock {\em Attacks and cascades in complex networks}, pages 299--310.
\newblock 2004.

\bibitem{demongeot2013archimedean}
Demongeot, J.; Ghassani, M.; Rachdi, M.; Ouassou, I.; Taramasco, C.
\newblock Archimedean copula and contagion modeling in epidemiology.
\newblock {\em Netw. Heterog. Media} {\bf 2013}, {\em 8},~149.

\bibitem{chen2015node}
Chen, C.; Tong, H.; Prakash, B.A.; Tsourakakis, C.E.; Eliassi-Rad, T.;
  Faloutsos, C.; Chau, D.
\newblock Node immunization on large graphs: Theory and algorithms.
\newblock {\em IEEE Trans. Knowl. Data Eng.} {\bf 2015},
  {\em 28},~113--126.

\bibitem{zhang2015data}
Zhang, Y.; Prakash, B.A.
\newblock Data-aware vaccine allocation over large networks.
\newblock {\em ACM Trans. Knowl. Discov. Data (TKDD)} {\bf
  2015}, {\em 10},~1--32.

\bibitem{song2015node}
Song, C.; Hsu, W.; Lee, M.
\newblock Node immunization over infectious period.
\newblock  In Proceedings of the 24th ACM International on Conference on  Information and Knowledge Management, Melbourne, Australia, 19--23 October 2015; pp. 831--840.

\bibitem{wijayanto2017flow}
Wijayanto, A.W.; Murata, T.
\newblock Flow-aware vertex protection strategy on large social networks.
\newblock  In {Proceedings of the }2017 IEEE/ACM International Conference on Advances in Social
  Networks Analysis and Mining (ASONAM), Sydney, NSW, Australia, 31 July--3 August 2017; pp. 58--63.

\bibitem{imrich2000product}
Imrich, W.; Klavzar, S.
\newblock {\em Product Graphs: Structure and Recognition}; Wiley: Hoboken, NJ, USA, 2000.

\bibitem{liu_2011}
Liu, J.; Zhang, T.
\newblock Epidemic spreading of an SEIRS model in scale-free networks.
\newblock {\em Commun. Nonlinear Sci. Numer. Simul.} {\bf 2011}, {\em 16},~3375--3384.

\bibitem{Mao_Xing_2009}
\newblock Modelling the spread of sexually transmitted diseases on scale-free networks.
Mao Xing, L.; Jiong R.,
\newblock Modelling the spread of sexually transmitted diseases on scale-free networks.
 \newblock {\em Chin. Phys. B} {\bf 2009}, {\em 1} ,~2115--2120.

\bibitem{song2020massive}
Song, W.; Zang, P.; Ding, Z.; Fang, X.; Zhu, L.; Zhu, Y.; Bao, C.; Chen, F.; Wu, M.; Peng, Z.,
\newblock  Massive migration promotes the early spread of COVID-19 in China: A study based on a scale-free network. \newblock {\em Infect. Dis. Poverty} {\bf 2020}, {\em 9} ,~1--8.

\bibitem{schneeberger2004scale}
Schneeberger, A.;  Mercer, C.; Gregson, S.; Ferguson, N.; Nyamukapa, C.; Anderson, R.; Johnson, A.; Garnett, G.
\newblock  Scale-free networks and sexually transmitted diseases: A description of observed patterns of sexual contacts in Britain and Zimbabwe. 
\newblock {\em Sex. Transm. Dis.} {\bf 2004}, {\em 31} ,~380--387.

\bibitem{moreno2003disease}
Moreno, Y.; Vazquez, A.
\newblock  Disease spreading in structured scale-free networks. 
\newblock {\em  Eur. Phys. J. B-Condens. Matter Complex Syst.} {\bf 2003}, {\em 31} ,~265--271.

\bibitem{yang2019dynamics}
Yang, P. Wang, Y.
\newblock  Dynamics for an SEIRS epidemic model with time delay on a scale-free network. 
\newblock {\em Phys. A Stat. Mech. Its Appl.} {\bf 2019}, {\em 527} ,~121290.

\bibitem{ke2006immunization}
Ke, H.; Yi, T.
\newblock  Immunization for scale-free networks by random walker.
\newblock {\em Chin. Phys.} {\bf 2006}, {\em 15} ,~2782.


\end{thebibliography}
\end{document}